\documentclass[11pt,a4paper]{article}
\usepackage[margin=2cm]{geometry}
\RequirePackage[round, authoryear]{natbib}
\usepackage[labelfont=bf]{caption}
\usepackage[utf8]{inputenc}
\usepackage[T1]{fontenc}
\usepackage{array}
\usepackage{multirow}
\usepackage{amsmath}
\usepackage{amssymb}
\usepackage{stmaryrd}
\usepackage{graphicx}
\usepackage{wasysym}
\usepackage{diagbox}
\usepackage{amsthm}
\usepackage{enumerate}
\usepackage[shortlabels]{enumitem}
\usepackage[colorlinks,
linkcolor=blue,
anchorcolor=blue,
citecolor=blue]{hyperref}
\usepackage{float}
\usepackage{graphicx}
\usepackage{placeins}
\usepackage{makecell}
\usepackage{xr}
\usepackage{cases}
\usepackage{mathtools}
\usepackage[math]{cellspace}
\usepackage{bigints}
\usepackage{setspace}
\usepackage{enumitem}
\usepackage{times, cleveref}

\newcommand*\dif{\mathop{}\!\mathrm{d}}

\theoremstyle{definition}

\newtheorem{prop}{Proposition}
\newtheorem{rem}{Remark}

\title{ 
	\textbf{Real and complexified configuration spaces for spherical 4-bar linkages}
}
	
\author{Zeyuan He, Kentaro Hayakawa and Makoto Ohsaki \\ \small Department of Architecture and Architectural Engineering, Kyoto University \\ \small he.zeyuan.8u@kyoto-u.ac.jp, se.hayakawa@archi.kyoto-u.ac.jp, ohsaki@archi.kyoto-u.ac.jp}
\date{}

\begin{document}\maketitle

This document provides a comprehensive library of symbolic parametrized expressions for the real and complexified configuration spaces of a spherical 4-bar linkage. Building on the foundational work of \citet[Section 2]{izmestiev_classification_2016}, the library extends the available expressions by incorporating all four folding angles across every possible combination of linkage lengths. It also includes the polynomial relationship between the diagonals (spherical arcs). Additionally, a fully functional MATLAB app script \citet{he_elliptic-fun-based_2023-1} is provided, offering tools for visualization and parametrization. The derivations are meticulously detailed to ensure accessibility for researchers from a broad spectrum of disciplines.

The spherical 4-bar linkage is a fundamental component in single-degree-of-freedom mechanisms that have been pivotal in engineering applications for over two centuries. When each spherical linkage is replaced by a sectoral rigid panel and each joint by a rotational hinge, the result is a comparable panel-hinge mechanism. In origami engineering, such mechanisms are identified as degree-4 single-vertex rigid origami, which has attracted significant interest in both theoretical research and practical implementation. Theoretical and applied studies on spherical 4-bar linkages or degree-4 single-vertex rigid origami are extensive. This note focuses on theoretical foundations, offering thorough analytical calculations and insights into various forms of algebraic analysis. For additional details, see \citet{chen_fourth_1983} for fourth-order synthesis of adjacent folding angles and the analysis of the pole's locus between the coupler and fixed linkage; \citet{gibson_movable_1988-1, gibson_movable_1988} for results on configuration space properties, singularities, and reductions with varying linkage lengths; and \citet{khimshiashvili_complex_2011} for contributions on moduli spaces, cross-ratio calculations, and generalized Heron polynomials for determining the area of spherical quadrilaterals. Within the context of origami engineering, we also highlight \citet{waitukaitis_origami_2016}, \citet{zimmermann_rigid_2020}, and \citet{foschi_explicit_2022}, which delve into analytical calculations for degree-4 vertices using spherical trigonometry through case-specific discussions.

By synthesizing these contributions, this document serves as a vital resource for engineers, architects, and mathematicians interested in symbolic parametrized expressions for configuring mechanisms based on 4-bar linkages.

\section{Modelling} \label{section: spherical modelling}

In this section we will set up the notation for a spherical 4-bar linkage, which is shown in Figure \ref{fig: degree-4 vertex}. We use $\alpha, ~\beta, ~\gamma, ~\delta \in \mathbb{R}^+$ to describe the \textbf{magnitudes of sector angles} and four \textbf{folding angles} $\rho_x, ~\rho_y, ~\rho_z, ~\rho_w \in \mathbb{R} \slash 2\pi$ to describe the configuration. $\mathbb{R} \slash 2\pi$ means $a, ~b \in \mathbb{R}$ are equivalent if $b = a + 2k\pi, ~k \in \mathbb{Z}$, which implies that we take $-\pi$ and $\pi$ as the same rotational angle.

The tangents of half of the folding angles are defined as 
\begin{equation}
	\begin{gathered}
		x=\tan \dfrac{\rho_x}{2}, ~~y=\tan \dfrac{\rho_y}{2}, ~~z=\tan \dfrac{\rho_z}{2}, ~~w=\tan \dfrac{\rho_w}{2} \\
		x, ~y, ~z, ~w \in \mathbb{R} \cup \{\infty\}
	\end{gathered}
\end{equation}
$\mathbb{R} \cup \{\infty\}$ is the \textbf{one-point compactification} of the real line, which `glues' $-\infty$ and $\infty$ to the same point. Moreover, in further algebraic discussions, $u$ and $v$ are the lengths of spherical diagonals.

\begin{figure} [!tb]
	\noindent \begin{centering}
		\includegraphics[width=1\linewidth]{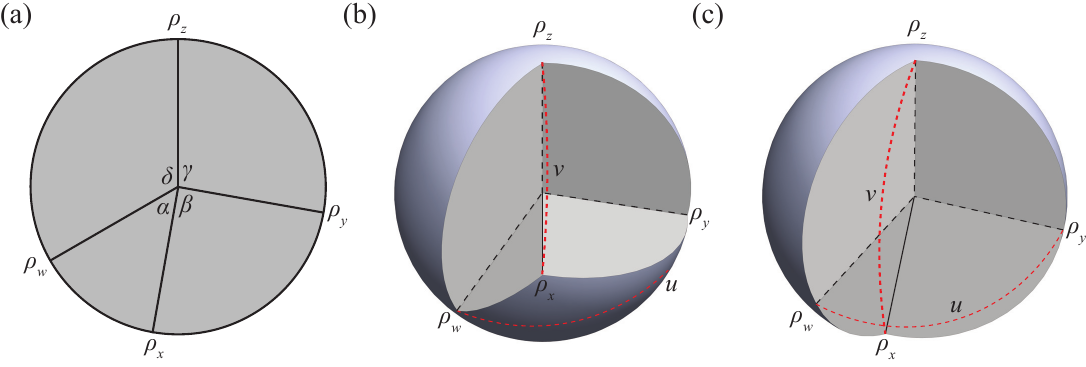}
		\par \end{centering}
	
	\caption{\label{fig: degree-4 vertex}(a) A degree-4 single-vertex rigid origami (not necessarily planar as shown here). We label the sector angles counter-clockwise as $\alpha$, $\beta$, $\gamma$, $\delta$; and the folding angles counter-clockwise as $\rho_x$, $\rho_y$, $\rho_z$, $\rho_w$. The tangent of half of these folding angles are $x, ~y, ~z, ~w$. $u$ and $v$ are two spherical diagonals of this vertex, each of which segments the spherical quadrilateral to spherical triangles. (b) and (c) show two non-trivial folded states with the outside edges of the single-vertex drawn on a sphere as arcs of great circles, assuming the panel corresponding to $\gamma$ is fixed when changing the magnitude of $z$. The mountain and valley creases are shown in solid and dashed lines. Generically, for a $z$ there are two sets of folding angles $\{x, ~y, ~w\}$, where the panels corresponding to $\alpha$ and $\beta$ of the two folded states are symmetric with respect to $u$.}
\end{figure}

\begin{rem}
	If considering the `stacking sequence' of these linkages, it is also reasonable to require $x = -\infty$ to be different from $x = + \infty$. This definition will make $x \in [-\infty, +\infty]$, which is called the \textit{two-point compactification} of the real line.
\end{rem}

First, we have the following proposition on the range of sector angles.
\begin{prop} \label{prop: sector angle range}
	$\alpha, ~\beta, ~\gamma, ~\delta$ are the sector angles of a degree-4 single-vertex if and only if
	\begin{equation} \label{eq: sector angle range}
		\begin{cases}
			\begin{gathered}
				\alpha, ~\beta, ~\gamma, ~\delta \in (0, \pi) \\
				\alpha < \beta + \gamma + \delta < \alpha + 2\pi \\
				\beta < \gamma + \delta + \alpha < \beta + 2\pi \\
				\gamma < \delta + \alpha + \beta < \gamma + 2\pi \\
				\delta < \alpha + \beta + \gamma  < \delta + 2\pi
			\end{gathered}
		\end{cases}
	\end{equation}
\end{prop}

\begin{proof}
	Necessity: Suppose we have a degree-4 single vertex,
	\begin{equation} 
		\begin{cases}
			\begin{gathered}
				\alpha, ~\beta, ~\gamma, ~\delta \\
				\pi-\alpha, ~\pi-\beta, ~\gamma, ~\delta \\
				\alpha, ~\pi-\beta, ~\pi-\gamma, ~\delta \\
				\alpha, ~\beta, ~\pi-\gamma, ~\pi-\delta \\
				\pi-\alpha, ~\beta, ~\gamma, ~\pi-\delta
			\end{gathered}
		\end{cases}
	\end{equation}
	are all spherical quadrilaterals, hence Equation \eqref{eq: sector angle range} holds.
	
	Sufficiency: Consider the four spherical triangles $(\alpha, ~\beta, ~u)$, $(\gamma, ~\delta, ~u)$, $(\beta, \gamma, ~v)$ and $(\delta, ~\alpha, ~v)$ when $\alpha, ~\beta, ~\gamma, ~\delta \in (0, \pi)$. $u, ~v$ should satisfy the following condition:
	\begin{equation}
		\begin{cases}
			\begin{gathered}
				|\alpha-\beta| < u < \min(\alpha+\beta,~ 2\pi-\alpha-\beta) \\
				|\gamma-\delta| < u < \min(\gamma+\delta, ~2\pi-\gamma-\delta) \\
				|\beta-\gamma| < v < \min(\beta+\gamma, ~2\pi-\beta-\gamma) \\
				|\delta-\alpha| < v < \min(\delta+\alpha, ~2\pi-\delta-\alpha) 
			\end{gathered}
		\end{cases}
	\end{equation}
	The sufficient condition for $\alpha, ~\beta, ~\gamma, ~\delta$ to be the sector angles of a degree-4 single-vertex is the range of $u$ and $v$ being non-empty, which means:
	\begin{equation}
		\begin{cases}
			\begin{gathered}
				|\alpha-\beta|  < \min(\gamma+\delta, ~2\pi-\gamma-\delta) \\
				|\gamma-\delta| < \min(\alpha+\beta, ~2\pi-\alpha-\beta) \\
				|\beta-\gamma| < \min(\delta+\alpha, ~2\pi-\delta-\alpha) \\
				|\delta-\alpha| < \min(\beta+\gamma, ~2\pi-\beta-\gamma) \\ 
			\end{gathered}
		\end{cases}
	\end{equation}
	This leads to Equation \eqref{eq: sector angle range}.
\end{proof}

Next we will give the polynomial relation between:
\begin{enumerate} [label={[\arabic*]}]
	\item adjacent rotational angles $x$ and $y$; $x$ and $w$.
	\item opposite rotational angles $x$ and $z$.
	\item lengths of diagonals $u$ and $v$.
\end{enumerate}
then conduct a detailed case-by-case discussion.

\section{Relation between adjacent folding angles}

\begin{figure} [!tb]
	\noindent \begin{centering}
		\includegraphics[width=0.5\linewidth]{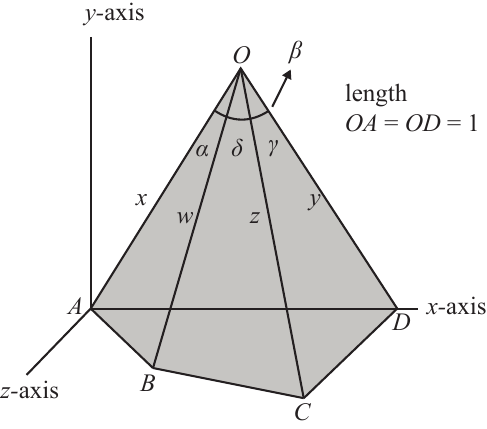}
		\par \end{centering}
	
	\caption{\label{fig: derivation vertex}This figure shows a degree-4 single-vertex rigid origami embedded in a 3-dimensional Euclidean coordinate system. Here $\angle AOB = \alpha$, $\angle BOC = \delta$, $\angle COD = \gamma$, $\angle DOA = \beta$. These sector angles $\alpha, ~\beta, ~\gamma, ~\delta$ satisfy the constraint in Proposition \ref{prop: sector angle range}. The folding angles on $OA, ~OB, ~OC, ~OD$ are $ \rho_x, ~\rho_w, ~\rho_z, ~\rho_y$. Whether these folding angles are all positive or all negative depends on the orientation specified. Triangle $OAD$ is set on the $xy$-plane with $OA=OD=1$. Next we set $\angle OAB = \angle ODC = \pi/2$ to make $O, ~A, ~B, ~C, ~D$ determined under given $\alpha, ~\beta, ~\gamma, ~\delta$. Hence $AB=\tan \alpha$, $CD=\tan \gamma$.}
\end{figure}

One method to derive the relation between adjacent folding angles is to build an Euclidean coordinate system, as shown in Figure \ref{fig: derivation vertex}. With the provided dimensions we could write the coordinates of $B$ and $C$ as:
\begin{equation*}
	\begin{gathered}
		B \left(\tan \alpha \cos (\pi+\rho_x) \cos \frac{\beta}{2}, -\tan \alpha \cos (\pi+\rho_x) \sin \frac{\beta}{2}, \tan \alpha \sin (\pi+\rho_x) \right) \\
		\Rightarrow B \left(- \dfrac{1-x^2}{1+x^2} \tan \alpha \cos \dfrac{\beta}{2} , \dfrac{1-x^2}{1+x^2} \tan \alpha \sin \dfrac{\beta}{2}, - \dfrac{2x}{1+x^2} \tan \alpha \right) \\
		C \left(2 \sin \frac{\beta}{2}-\tan \gamma \cos (\pi+\rho_y) \cos \frac{\beta}{2}, -\tan \gamma \cos (\pi+\rho_y) \sin \frac{\beta}{2}, \tan \gamma \sin (\pi+\rho_y) \right) \\
		\Rightarrow C \left( 2 \sin \frac{\beta}{2}+\dfrac{1-y^2}{1+y^2} \tan \gamma \cos \dfrac{\beta}{2} , \dfrac{1-y^2}{1+y^2} \tan \gamma \sin \dfrac{\beta}{2}, - \dfrac{2y}{1+y^2} \tan \gamma \right)
	\end{gathered}
\end{equation*}
Within triangle OBC we could get the relation between $x$ and $y$:
\begin{equation*}
	\begin{aligned}
		& \left(2 \sin \frac{\beta}{2}+\dfrac{1-x^2}{1+x^2} \tan \alpha \cos \dfrac{\beta}{2} + \dfrac{1-y^2}{1+y^2} \tan \gamma \cos \dfrac{\beta}{2}\right)^2 \\
		+ & \left(\dfrac{1-x^2}{1+x^2} \tan \alpha \sin \dfrac{\beta}{2}-\dfrac{1-y^2}{1+y^2} \tan \gamma \sin \dfrac{\beta}{2}\right)^2 \\
		+ & \left(\dfrac{2x}{1+x^2} \tan \alpha-\dfrac{2y}{1+y^2} \tan \gamma \right)^2  \\
		= & \dfrac{1}{\cos^2 \alpha} + \dfrac{1}{\cos^2 \gamma} - \dfrac{2 \cos \delta}{\cos \alpha \cos \gamma}
	\end{aligned}
\end{equation*}
The next step is to simplify the above equation to a polynomial equation on $x$ and $y$, the result is:
\begin{equation*} \label{eq: degree-4 vertex pre}
	\begin{aligned}
		& \left( \cos (\alpha - \beta + \gamma) - \cos \delta \right) x^2 y^2 \\ 
		+ & \left( \cos (\alpha - \beta - \gamma) - \cos \delta \right) x^2 \\
		+ & \left(4 \sin \alpha \sin \gamma \right) xy \\
		+ & \left( \cos (\alpha + \beta - \gamma ) - \cos \delta \right) y^2 \\
		+ & \left( \cos (\alpha + \beta + \gamma ) - \cos \delta \right) =0
	\end{aligned}
\end{equation*}
We further simplify the above equation to the polynomial form below:
\begin{equation} \label{eq: degree-4 vertex adjacent}
	f(\alpha, ~\beta, ~\gamma, ~\delta, ~x, ~y) = f_{22} x^2y^2 + f_{20} x^2 + 2f_{11}xy + f_{02} y^2 + f_{00} = 0
\end{equation}
\begin{equation*}
	\begin{gathered}
		f_{22} = \sin \dfrac{\alpha - \beta + \gamma + \delta}{2} \sin \dfrac{\alpha - \beta + \gamma - \delta}{2} = \sin (\sigma-\beta) \sin (\sigma-\beta-\delta) \\
		f_{20} = \sin \dfrac{ - \alpha + \beta + \gamma + \delta}{2} \sin \dfrac{ - \alpha + \beta + \gamma - \delta}{2} = \sin (\sigma-\alpha) \sin (\sigma-\alpha-\delta) \\
		f_{11} = - \sin \alpha \sin \gamma \\
		f_{02} = \sin \dfrac{-\alpha - \beta + \gamma + \delta}{2} \sin \dfrac{-\alpha - \beta + \gamma - \delta}{2} = \sin (\sigma-\gamma) \sin (\sigma-\gamma-\delta) \\
		f_{00} = \sin \dfrac{\alpha + \beta + \gamma + \delta}{2} \sin \dfrac{\alpha + \beta + \gamma - \delta}{2} = \sin \sigma \sin (\sigma-\delta) \\
		\sigma =  \dfrac{\alpha + \beta + \gamma + \delta}{2} \in (0, 2\pi) \\
	\end{gathered}
\end{equation*}
$\sigma$ is the \textit{semi-perimeter} of a spherical 4-bar linkage. Note that $\alpha, ~\beta, ~\gamma, ~\delta$ satisfy Proposition \ref{prop: sector angle range}, and we allow $x, ~y \in \mathbb{R} \cup \{\infty\}$ in Equation \eqref{eq: degree-4 vertex adjacent}. 

Further we will classify Equation \eqref{eq: degree-4 vertex adjacent} based on the \textbf{degree of degeneracy} --- how many of $f_{22}, ~f_{20}, ~f_{02}, ~f_{00}$ are zero ($f_{11} < 0$). From Proposition \ref{prop: sector angle range} we have:
\begin{equation*}
	\sigma-\alpha \in (0, \pi), ~~ \sigma-\beta \in (0, \pi), ~~ \sigma-\gamma \in (0, \pi), ~~ \sigma-\delta \in (0, \pi) \\
\end{equation*}
We classify a spherical 4-bar linkage in the following way:
\begin{enumerate} [label={[\arabic*]}]
	\item a \textit{square}, if:
	\begin{equation*}
		f_{22}=0, ~~ f_{20}=0,~~f_{02}=0, ~~f_{00} = 0 ~~ \Leftrightarrow ~~  \alpha=\beta=\gamma=\delta=\pi/2
	\end{equation*}
	\item a \textit{rhombus}, if:
	\begin{equation*}
		f_{22}=0, ~~ f_{20}=0,~~f_{02}=0, ~~f_{00} \neq 0 ~~ \Leftrightarrow ~~  \alpha=\beta=\gamma=\delta \neq \pi/2
	\end{equation*}
	\item a \textit{cross}, if:
	\begin{equation*}
		f_{22} \neq 0, ~~ f_{20}=0,~~f_{02}=0, ~~f_{00} = 0 ~~ \Leftrightarrow ~~  \alpha= \pi - \beta = \gamma = \pi - \delta \neq \pi/2
	\end{equation*}
	\item of type \textit{Miura I}, if:
	\begin{equation*}
		f_{22} = 0, ~~ f_{20} \neq 0,~~f_{02}=0, ~~f_{00} = 0 ~~ \Leftrightarrow ~~  \alpha= \pi - \beta = \pi - \gamma = \delta \neq \pi/2
	\end{equation*}
	\item of type \textit{Miura II}, if:
	\begin{equation*}
		f_{22} = 0, ~~ f_{20} = 0,~~f_{02} \neq 0, ~~f_{00} = 0 ~~ \Leftrightarrow ~~  \alpha = \beta = \pi - \gamma = \pi - \delta \neq \pi/2
	\end{equation*}
	\item an \textit{isogram}, if: 
	\begin{equation*}
		f_{22} \neq 0, ~~ f_{20} = 0,~~f_{02} = 0, ~~f_{00} \neq 0 ~~ \Leftrightarrow ~~ \gamma=\alpha, ~~ \delta=\beta, ~~\beta \neq \alpha, ~~\alpha + \beta \neq \pi
	\end{equation*}
	\item an \textit{anti-isogram}, if: 
	\begin{equation*}
		f_{22} = 0, ~~ f_{20} \neq 0,~~f_{02} \neq 0, ~~f_{00} = 0 ~~ \Leftrightarrow ~~ \gamma=\pi - \alpha, ~~ \delta= \pi - \beta, ~~\beta \neq \alpha, ~~\alpha + \beta \neq \pi
	\end{equation*}
	\item a \textit{deltoid I}, if:
	\begin{equation*}
		f_{22} = 0, ~~ f_{20} \neq 0,~~f_{02} = 0, ~~f_{00} \neq 0 ~~ \Leftrightarrow ~~ \delta=\alpha, ~~ \gamma=\beta, ~~\beta \neq \alpha, ~~\alpha + \beta \neq \pi
	\end{equation*}
	\item an \textit{anti-deltoid I}, if:
	\begin{equation*}
		f_{22} \neq 0, ~~ f_{20} = 0,~~f_{02} \neq 0, ~~f_{00} = 0 ~~ \Leftrightarrow ~~ \delta= \pi - \alpha, ~~ \gamma = \pi - \beta, ~~\beta \neq \alpha, ~~\alpha + \beta \neq \pi
	\end{equation*}
	\item a \textit{deltoid II}, if:
	\begin{equation*}
		f_{22} = 0, ~~ f_{20} = 0,~~f_{02} \neq 0, ~~f_{00} \neq 0 ~~ \Leftrightarrow ~~ \alpha=\beta, ~~ \delta=\gamma, ~~\gamma \neq \beta, ~~\beta + \gamma \neq \pi
	\end{equation*}
	\item an \textit{anti-deltoid II}, if:
	\begin{equation*}
		f_{22} \neq 0, ~~ f_{20} \neq 0,~~f_{02} = 0, ~~f_{00} = 0 ~~ \Leftrightarrow ~~ \alpha= \pi - \beta, ~~ \delta= \pi - \gamma, ~~\gamma \neq \beta, ~~\beta + \gamma \neq \pi
	\end{equation*}
	\item of type \textit{conic I}, if:
	\begin{equation*} 
		f_{22} = 0, ~~ f_{20} \neq 0,~~f_{02} \neq 0, ~~f_{00} \neq 0  ~~ \Leftrightarrow ~~ 
		\begin{dcases}
			\alpha - \beta + \gamma - \delta = 0 \\
			\alpha - \beta - \gamma + \delta \neq 0 \\
			\alpha + \beta - \gamma - \delta \neq 0 \\
			\alpha + \beta + \gamma + \delta \neq 2\pi
		\end{dcases}
	\end{equation*}
	\item of type \textit{conic II}, if:
	\begin{equation*} 
		f_{22} \neq 0, ~~ f_{20} = 0,~~f_{02} \neq 0, ~~f_{00} \neq 0  ~~ \Leftrightarrow ~~ 
		\begin{dcases}
			\alpha - \beta + \gamma - \delta \neq 0 \\
			\alpha - \beta - \gamma + \delta = 0 \\
			\alpha + \beta - \gamma - \delta \neq 0 \\
			\alpha + \beta + \gamma + \delta \neq 2\pi
		\end{dcases}
	\end{equation*}
	\item of type \textit{conic III}, if:
	\begin{equation*} 
		f_{22} \neq 0, ~~ f_{20} \neq 0,~~f_{02} = 0, ~~f_{00} \neq 0  ~~ \Leftrightarrow ~~ 
		\begin{dcases}
			\alpha - \beta + \gamma - \delta \neq 0 \\
			\alpha - \beta - \gamma + \delta \neq 0 \\
			\alpha + \beta - \gamma - \delta = 0 \\
			\alpha + \beta + \gamma + \delta \neq 2\pi
		\end{dcases}
	\end{equation*}
	\item of type \textit{conic IV}, if:
	\begin{equation*} 
		f_{22} \neq 0, ~~ f_{20} \neq 0,~~f_{02} \neq 0, ~~f_{00} = 0  ~~ \Leftrightarrow ~~ 
		\begin{dcases}
			\alpha - \beta + \gamma - \delta \neq 0 \\
			\alpha - \beta - \gamma + \delta \neq 0 \\
			\alpha + \beta - \gamma - \delta \neq 0 \\
			\alpha + \beta + \gamma + \delta = 2\pi
		\end{dcases}
	\end{equation*}
	\item of type \textit{elliptic}, if:
	\begin{equation*} 
		f_{22} \neq 0, ~~ f_{20} \neq 0,~~f_{02} \neq 0, ~~f_{00} \neq 0  ~~ \Leftrightarrow ~~ 
		\begin{dcases}
			\alpha - \beta + \gamma - \delta \neq 0 \\
			\alpha - \beta - \gamma + \delta \neq 0 \\
			\alpha + \beta - \gamma - \delta \neq 0 \\
			\alpha + \beta + \gamma + \delta \neq 2\pi
		\end{dcases}
	\end{equation*}
	Further, if the following relation is satisfied:
	\begin{equation*}
		\cos \alpha \cos \gamma = \cos \beta \cos \delta 
	\end{equation*}
	We call this type \textit{orthodiagonal} since the diagonals are orthogonal on the sphere (sum of the square of opposite arc lengths are equal), and we will analyze its special configuration space later.
\end{enumerate}
Similarly, for the standard form on $x$ and $w$, we have the following relation:
\begin{equation} \label{eq: degree-4 vertex 2}
	f(\beta, ~\alpha, ~\delta, ~\gamma, ~x, ~w) = 0
\end{equation}

\section{Relation between opposite folding angles}

In terms of $x$ and $z$, we could derive the relation between $z^2$ and $x^2$ from spherical trigonometry.
\begin{equation*}
	\begin{gathered}
		\cos u = \cos \delta \cos \gamma - \sin \delta \sin \gamma \dfrac{1-z^2}{1+z^2} = \cos \alpha \cos \beta - \sin \alpha \sin \beta \dfrac{1-x^2}{1+x^2}
	\end{gathered}
\end{equation*}
Next step is to transfer them to polynomials,
\begin{equation*}
	\begin{aligned}
		& \left( \cos(\alpha - \beta) - \cos(\delta - \gamma) \right) x^2 z^2 \\ 
		+ & \left( \cos(\alpha - \beta) - \cos(\delta + \gamma) \right) x^2 \\
		+ & \left( \cos(\alpha + \beta) - \cos(\delta - \gamma) \right) z^2 \\
		+ & \left( \cos(\alpha + \beta) - \cos(\delta + \gamma) \right) =0
	\end{aligned}
\end{equation*}
which leads to the simplified polynomial form on $z$ and $x$ below:
\begin{equation} \label{eq: opposite folding angle}
	g_{22}x^2z^2+g_{20}x^2+g_{02}z^2+g_{00}=0
\end{equation}
where
\begin{equation*} 
	\begin{gathered}
		g_{22}= \sin \dfrac{ -\alpha + \beta + \gamma - \delta }{2} \sin \dfrac{ \alpha - \beta + \gamma - \delta }{2} = \sin (\sigma-\alpha-\delta) \sin (\sigma-\beta-\delta) \\
		g_{20}= \sin \dfrac{ -\alpha + \beta + \gamma + \delta }{2} \sin \dfrac{ \alpha - \beta + \gamma + \delta }{2} = \sin (\sigma-\alpha) \sin (\sigma-\beta) \\
		g_{02}= \sin \dfrac{ -\alpha - \beta + \gamma - \delta }{2} \sin \dfrac{ \alpha + \beta + \gamma - \delta }{2} = -\sin (\sigma-\gamma) \sin (\sigma-\delta)\\
		g_{00}= \sin \dfrac{ \alpha + \beta + \gamma + \delta }{2} \sin \dfrac{ -\alpha - \beta + \gamma + \delta }{2} = \sin \sigma \sin (\sigma - \alpha - \beta) \\
		\sigma=\dfrac{\alpha+\beta+\gamma+\delta}{2}
	\end{gathered}
\end{equation*}
The identities listed in Proposition \ref{prop: identities folding angle} would help to understand the symmetry of coefficients in Equation \eqref{eq: opposite folding angle}.

It is clear that $g_{20}>0$ and $g_{02}<0$. Note that $x$ and $z$ can be separated if $g_{22}x^2+g_{02} \neq 0$:
\begin{equation}
	z^2=-\dfrac{g_{20}x^2+g_{00}}{g_{22}x^2+g_{02}}
\end{equation}
Equation \eqref{eq: opposite folding angle} after a case-by-case study is:
\begin{enumerate} [label={[\arabic*]}]
	\item square
	\begin{equation*}
		\begin{aligned}
			\alpha=\beta=\gamma=\delta=\pi/2 & ~~ \Rightarrow ~~ g_{22}=0, ~~ g_{20}=-g_{02}=1, ~~ g_{00}=0 \\
			& ~~ \Rightarrow ~~ z = \pm x
		\end{aligned}
	\end{equation*}
	\item rhombus
	\begin{equation*}
		\begin{aligned}
			\alpha=\beta=\gamma=\delta \neq \dfrac{\pi}{2} & ~~ \Rightarrow ~~ g_{22}=0, ~~ g_{20}=-g_{02}=\sin^2 \alpha, ~~ g_{00}=0 \\
			& ~~ \Rightarrow ~~ z = \pm x
		\end{aligned}
	\end{equation*}
	\item cross
	\begin{equation*}
		\begin{aligned}
			\alpha= \pi - \beta = \gamma = \pi - \delta \neq \pi/2 & ~~ \Rightarrow ~~ g_{22}=0, ~~ g_{20}=-g_{02}=\sin^2 \alpha, ~~ g_{00}=0 \\
			& ~~ \Rightarrow ~~ z = \pm x
		\end{aligned}
	\end{equation*}
	\item Miura I
	\begin{equation*}
		\begin{aligned}
			\alpha= \pi - \beta = \pi - \gamma = \delta \neq \pi/2 & ~~ \Rightarrow ~~ g_{22}=0, ~~ g_{20}=-g_{02}=\sin^2 \alpha, ~~ g_{00}=0 \\
			& ~~ \Rightarrow ~~ z = \pm x
		\end{aligned}
	\end{equation*}
	\item Miura II
	\begin{equation*}
		\begin{aligned}
			\alpha = \beta = \pi - \gamma = \pi - \delta \neq \pi/2 & ~~ \Rightarrow ~~ g_{22}=0, ~~ g_{20}=-g_{02}=\sin^2 \alpha, ~~ g_{00}=0 \\
			& ~~ \Rightarrow ~~ z = \pm x
		\end{aligned}
	\end{equation*}
	\item isogram 
\begin{equation*}
	\begin{aligned}
		& \quad \quad \gamma=\alpha, ~~ \delta=\beta, ~~\beta \neq \alpha, ~~\alpha + \beta \neq \pi  \\
		& \Rightarrow ~~ g_{22}=0, ~~ g_{20}=-g_{02}=\sin \alpha \sin \beta, ~~g_{00}=0 \\
		& \Rightarrow ~~ z= \pm x
	\end{aligned}
\end{equation*}
\item anti-isogram 
\begin{equation*}
	\begin{aligned}
		& \quad \quad \gamma=\pi - \alpha, ~~ \delta= \pi - \beta, ~~\beta \neq \alpha, ~~\alpha + \beta \neq \pi \\
		& \Rightarrow ~~ g_{22}=0, ~~ g_{20}=-g_{02}=\sin \alpha \sin \beta, ~~g_{00}=0 \\
		& \Rightarrow ~~ z= \pm x
	\end{aligned}
\end{equation*}
\item deltoid I
\begin{equation*}
	\begin{aligned}
		& \quad \quad \delta=\alpha, ~~ \gamma=\beta, ~~\beta \neq \alpha, ~~\alpha + \beta \neq \pi \\
		& \Rightarrow ~~ g_{22}=0, ~~ g_{20}=-g_{02}=\sin \alpha \sin \beta, ~~g_{00}=0 \\
		& \Rightarrow ~~ z= \pm x
	\end{aligned}
\end{equation*}
\item anti-deltoid I
\begin{equation*}
	\begin{aligned}
		& \quad \quad \delta=\pi-\alpha, ~~ \gamma=\pi-\beta, ~~\beta \neq \alpha, ~~\alpha + \beta \neq \pi \\
		& \Rightarrow ~~ g_{22}=0, ~~ g_{20}=-g_{02}=\sin \alpha \sin \beta, ~~g_{00}=0 \\
		& \Rightarrow ~~ z= \pm x
	\end{aligned}
\end{equation*}
\item deltoid II
\begin{equation*}
	\begin{aligned}
		& \quad \quad \alpha=\beta, ~~ \delta=\gamma, ~~\gamma \neq \beta, ~~\beta + \gamma \neq \pi \\
		& \Rightarrow ~~ g_{22}=0, ~~ g_{20}=\sin^2 \gamma, ~~ g_{02}=-\sin^2 \beta, ~~g_{00} = \sin^2 \gamma - \sin^2 \beta \\
		& \Rightarrow ~~ \sin^2 \beta (z^2 +1) = \sin^2 \gamma (x^2+1)  
	\end{aligned}
\end{equation*}
\item anti-deltoid II
\begin{equation*}
	\begin{aligned}
		& \quad \quad \alpha= \pi - \beta, ~~ \delta= \pi - \gamma, ~~\gamma \neq \beta, ~~\beta + \gamma \neq \pi \\
		& \Rightarrow ~~ g_{22}= \sin^2 \beta - \sin^2 \gamma, ~~ g_{20}=\sin^2 \beta, ~~ g_{02}=-\sin^2 \gamma, ~~g_{00} = 0 \\
		& \Rightarrow ~~ \sin^2 \beta (z^{-2} +1) = \sin^2 \gamma (x^{-2}+1) 
	\end{aligned}
\end{equation*}
\item conic I
\begin{equation*}
	\begin{aligned}
		\begin{dcases}
			\alpha - \beta + \gamma - \delta = 0 \\
			\alpha - \beta - \gamma + \delta \neq 0 \\
			\alpha + \beta - \gamma - \delta \neq 0 \\
			\alpha + \beta + \gamma + \delta \neq 2\pi
		\end{dcases} & ~~ \Rightarrow ~~ 
		\begin{dcases*}
			g_{22}=0, ~~ g_{20}=\sin \gamma \sin \delta, ~~ g_{02}=-\sin \alpha \sin \beta \\ g_{00}=\sin(\alpha+\gamma) \sin(-\beta+\gamma)= -\sin \alpha \sin \beta+ \sin \gamma \sin \delta
		\end{dcases*} \\
		& ~~ \Rightarrow ~~ (z^2+1) \sin \alpha \sin \beta  = (x^2 +1 ) \sin \gamma \sin \delta 
	\end{aligned}
\end{equation*}
We refer to Proposition \ref{prop: identities folding angle} for helpful identities in the derivation. 
\item conic II
\begin{equation*}
	\begin{aligned}
		\begin{dcases}
			\alpha - \beta + \gamma - \delta \neq 0 \\
			\alpha - \beta - \gamma + \delta = 0 \\
			\alpha + \beta - \gamma - \delta \neq 0 \\
			\alpha + \beta + \gamma + \delta \neq 2\pi
		\end{dcases} & ~~ \Rightarrow ~~ 
		\begin{dcases*}
			g_{22}=0, ~~ g_{20}=\sin \gamma \sin \delta, ~~ g_{02}=-\sin \alpha \sin \beta \\ g_{00}=\sin(\alpha+\delta) \sin(-\beta+\delta)= -\sin \alpha \sin \beta+ \sin \gamma \sin \delta
		\end{dcases*} \\
		& ~~ \Rightarrow ~~ (z^2+1) \sin \alpha \sin \beta  = (x^2 +1 ) \sin \gamma \sin \delta 
	\end{aligned}
\end{equation*}
\item conic III
\begin{equation*}
	\begin{aligned}
		\begin{dcases}
			\alpha - \beta + \gamma - \delta \neq 0 \\
			\alpha - \beta - \gamma + \delta \neq 0 \\
			\alpha + \beta - \gamma - \delta = 0 \\
			\alpha + \beta + \gamma + \delta \neq 2\pi
		\end{dcases} & ~~ \Rightarrow ~~ 
		\begin{dcases*}
			g_{22}= \sin(\beta-\gamma) \sin(\alpha-\gamma) = \sin \alpha \sin \beta- \sin \gamma \sin \delta \\ g_{20}=\sin \alpha \sin \beta, ~~ g_{02}=-\sin \gamma \sin \delta , ~~ g_{00} = 0
		\end{dcases*} \\
		& ~~ \Rightarrow ~~ (z^{-2}+1) \sin \alpha \sin \beta = (x^{-2} +1) \sin \gamma \sin \delta 
	\end{aligned}
\end{equation*}
\item conic IV
\begin{equation*}
	\begin{aligned}
		\begin{dcases}
			\alpha - \beta + \gamma - \delta \neq 0 \\
			\alpha - \beta - \gamma + \delta \neq 0 \\
			\alpha + \beta - \gamma - \delta \neq 0 \\
			\alpha + \beta + \gamma + \delta = 2\pi
		\end{dcases} & ~~ \Rightarrow ~~ 
		\begin{dcases*}
			g_{22}= \sin(\beta+\gamma) \sin(\alpha+\gamma) = \sin \alpha \sin \beta- \sin \gamma \sin \delta \\ g_{20}=\sin \alpha \sin \beta, ~~ g_{02}=-\sin \gamma \sin \delta , ~~ g_{00} = 0
		\end{dcases*} \\
		& ~~ \Rightarrow ~~ (z^{-2}+1) \sin \alpha \sin \beta = (x^{-2} +1)\sin \gamma \sin \delta 
	\end{aligned}
\end{equation*}
\item Elliptic. See Section \ref{section: spherical elliptic} since the expression is not reduced.
\end{enumerate}

\section{Relation between the lengths of diagonals}

The polynomial relation between the spherical lengths of the diagonals $u$ and $v$ holds significant importance in our study and finds frequent application.

To simplify our analysis, it is necessary to examine the relationship between $u$ and $v$ independently of the folding angles. An invaluable method for this purpose is the spherical adaptation of the \textit{Cayley-Menger determinant}, applied to the four vertices of a spherical 4-bar linkage. The criterion for this linkage to reside on a sphere is as follows:
\begin{equation*}
	\begin{gathered}
		\left|
		\begin{matrix}
			\dfrac{1}{2} & 1 & 1 & 1 & 1 \\
			1 & 0 & 2(1-\cos \alpha) & 2(1-\cos u) & 2(1-\cos \delta) \\
			1 & 2(1-\cos \alpha) & 0 & 2(1-\cos \beta) &  2(1-\cos v) \\
			1 & 2(1-\cos u) & 2(1-\cos \beta) & 0 & 2(1-\cos \gamma) \\
			1 & 2(1-\cos \delta) & 2(1-\cos v) & 2(1-\cos \gamma) & 0
		\end{matrix}
		\right| = 0 \\[10pt]
	\end{gathered}
\end{equation*}
After full simplification, the result is an order 3 polynomial over $1 - \cos u$ and $1- \cos v$:
\begin{equation} \label{eq: relation between diagonals spherical}
	\begin{aligned}
		h(\alpha, ~\beta, ~\gamma, ~\delta, ~u, ~ v) & = (1- \cos u)^2(1-\cos v)^2 \\
		& \quad -2(1- \cos u)^2(1-\cos v)-2(1- \cos u)(1-\cos v)^2 \\
		& \quad +h_{11}(1- \cos u)(1-\cos v) \\
		& \quad +h_{10}(1- \cos u)+h_{01}(1- \cos v)+h_{00}=0
	\end{aligned}
\end{equation}
where
\begin{equation*}
	\begin{gathered}
		h_{11}=2 (2 -\cos \alpha \cos \gamma - \cos \beta \cos \delta) \\
		h_{10}=2(\cos \beta-\cos \gamma)(\cos \delta-\cos \alpha) \\
		h_{01}=2(\cos \alpha-\cos \beta)(\cos \gamma-\cos \delta) \\
		\begin{aligned}
			h_{00} & = (\cos \alpha \cos \gamma - \cos \beta \cos \delta)^2 - (\cos \alpha - \cos \beta + \cos \gamma - \cos \delta)^2 \\
			& = \left( (1- \cos \alpha) (1- \cos \gamma) - (1- \cos \beta) (1- \cos \delta) \right) \\ & \quad  \left( (1+ \cos \alpha) (1+ \cos \gamma) - (1+ \cos \beta) (1+ \cos \delta) \right)
		\end{aligned}
	\end{gathered}
\end{equation*}
Here we use two simple examples to show how to make use of Equation \eqref{eq: relation between diagonals spherical}:
\begin{enumerate} [label={[\arabic*]}]
	\item square
	\begin{equation*}
		\begin{aligned}
			\alpha=\beta=\gamma=\delta=\pi/2 & ~~ \Rightarrow ~~ h_{11}=4, ~~ h_{10}=h_{01}=h_{00} = 0 \\
			& ~~ \Rightarrow ~~ (1 - \cos u)(1 - \cos v)(1 + \cos u)(1 + \cos v) = 0 \\
			& ~~ \Rightarrow ~~ u \equiv 0 ~~\mathrm{or}~~ v \equiv 0 ~~\mathrm{or}~~ u \equiv \pi ~~\mathrm{or}~~ v \equiv \pi \\
		\end{aligned}
	\end{equation*}
	\item rhombus
	\begin{equation*}
		\begin{aligned}
			\alpha=\beta=\gamma=\delta \neq \dfrac{\pi}{2} & ~~ \Rightarrow ~~ h_{11}=4-4 \cos^2 \alpha, ~~ h_{10}=h_{01}=h_{00} = 0 \\
			& ~~ \Rightarrow ~~ (1 - \cos u)(1 - \cos v)\left((1 + \cos u)(1 + \cos v)-4 \cos^2 \alpha\right) = 0 \\
			& ~~ \Rightarrow ~~ u \equiv 0 ~~\mathrm{or}~~ v \equiv 0 ~~\mathrm{or}~~ (1 + \cos u)(1 + \cos v) = 4 \cos^2 \alpha \\
		\end{aligned}
	\end{equation*}
\end{enumerate}

Equation \eqref{eq: relation between diagonals spherical} leads to a useful proposition:
\begin{prop}
	There is a one-to-one correspondence between the two spherical 4-bar linkages below:
	\begin{equation}
		(\alpha, ~\beta, ~\gamma, ~\delta, ~u, ~v) \Leftrightarrow (\sigma-\alpha, ~\sigma-\beta, ~\sigma-\gamma, ~\sigma-\delta, ~u, ~v)
	\end{equation}
	That is to say, for every spherical 4-bar linkage with sector angles $\alpha, ~\beta, ~\gamma, ~\delta$ and diagonal lengths $u, ~v$, there is another spherical 4-bar linkage with sector angles $\sigma-\alpha, ~\sigma-\beta, ~\sigma-\gamma, ~\sigma-\delta$ and the same diagonal lengths $u, ~v$. We say they are \textit{conjugate} spherical 4-bar linkages.
\end{prop}

By leveraging the identities outlined in Proposition \ref{prop: identities folding angle} and performing direct symbolic computations, it becomes evident that the coefficients $h_{11}, ~h_{10}, ~h_{01}, ~h_{00}$ remain unchanged when transitioning a spherical 4-bar linkage to its conjugate.

\section{Helpful identities and sign convention} \label{section: sign convention folding angle}

This section presents several useful identities that aid in analysing the polynomial relations discussed above.

\begin{prop} \label{prop: identities folding angle}
	Some identities over $\alpha, ~\beta, ~\gamma, ~ \delta$ and $\sigma-\alpha, ~\sigma-\beta, ~ \sigma-\gamma, ~\sigma-\delta$. Note that $\sigma=(\alpha+\beta+\gamma+\delta)/2$ is the semi-perimeter without further clarification.
	\begin{enumerate} [label={[\arabic*]}]
		\item 
		\begin{equation*}
			\sigma-\alpha-\beta =  -(\sigma-\gamma-\delta)
		\end{equation*}
		\item
		\begin{equation*}
			\begin{aligned}
				\sin(\sigma-\alpha) \sin(\sigma-\beta)- \sin \alpha \sin \beta 
				& = \dfrac{1}{2} \left(\cos (\alpha + \beta) - \cos (\gamma + \delta) \right) \\
				& = \sin \sigma \sin(\sigma-\alpha-\beta) \\
				& = -\sin \sigma \sin(\sigma-\gamma-\delta) \\
				& = \sin \gamma \sin \delta - \sin (\sigma-\gamma) \sin (\sigma-\delta)
			\end{aligned}
		\end{equation*} 
		\item
		\begin{equation*}
			\begin{aligned}
				\sin (\sigma-\alpha) \sin (\sigma-\beta)- \sin \gamma \sin \delta & = \dfrac{1}{2} \left(\cos (\alpha - \beta) - \cos (\gamma - \delta) \right) \\
				& = \sin (\sigma-\alpha-\gamma) \sin (\sigma-\beta-\gamma) \\
				& = - \sin (\sigma-\gamma-\alpha) \sin (\sigma-\delta-\alpha) \\
				& = \sin \alpha \sin \beta- \sin (\sigma-\gamma) \sin (\sigma-\delta)
			\end{aligned}
		\end{equation*}
		\item 
		\begin{equation*}
			\sin \alpha \sin \beta + \sin \gamma \sin \delta = \sin (\sigma-\alpha) \sin (\sigma-\beta) + \sin (\sigma-\gamma) \sin (\sigma-\delta) 
		\end{equation*}
		\begin{equation*}
			\begin{aligned}
				& \sin \alpha \sin \beta \sin \gamma \sin \delta - \sin (\sigma-\alpha) \sin (\sigma-\beta) \sin (\sigma-\gamma) \sin (\sigma-\delta) \\
				= &  [\sin(\sigma-\alpha) \sin (\sigma-\beta)+ \sin (\sigma-\gamma) \sin (\sigma-\delta)- \sin \gamma \sin \delta] \sin \gamma \sin \delta \\ & \quad -\sin (\sigma-\alpha) \sin (\sigma-\beta) \sin (\sigma-\gamma) \sin (\sigma-\delta) \\ 
				= &  (\sin (\sigma-\alpha) \sin (\sigma-\beta)- \sin \gamma \sin \delta)( \sin \gamma \sin \delta-\sin (\sigma-\gamma) \sin (\sigma-\delta)) \\
				= &  \sin \sigma \sin (\sigma-\alpha-\beta) \sin (\sigma-\alpha-\gamma) \sin (\sigma-\beta-\gamma)			
			\end{aligned}
		\end{equation*}
		\item
		\begin{equation*}
			\begin{aligned}
				\cos(\sigma-\alpha) \cos(\sigma-\beta)- \cos \alpha \cos \beta 
				& = \dfrac{1}{2} \left(-\cos (\alpha + \beta) + \cos (\gamma + \delta) \right) \\
				& = -\sin \sigma \sin(\sigma-\alpha-\beta) \\
				& = \sin \sigma \sin(\sigma-\gamma-\delta) \\
				& = \cos \gamma \cos \delta - \cos (\sigma-\gamma) \cos (\sigma-\delta)
			\end{aligned}
		\end{equation*} 
		\item
		\begin{equation*}
			\begin{aligned}
				\cos (\sigma-\alpha) \cos (\sigma-\beta)- \cos \gamma \cos \delta & = \dfrac{1}{2} \left( \cos (\alpha - \beta) - \cos (\gamma - \delta) \right) \\
				& = \sin (\sigma-\alpha-\gamma) \sin (\sigma-\beta-\gamma) \\
				& = - \sin (\sigma-\gamma-\alpha) \sin (\sigma-\delta-\alpha) \\
				& = \cos \alpha \cos \beta- \cos (\sigma-\gamma) \cos (\sigma-\delta)
			\end{aligned}
		\end{equation*}
		\begin{equation*}
			\cos \alpha \cos \beta + \cos \gamma \cos \delta = \cos (\sigma-\alpha) \cos (\sigma-\beta) + \cos (\sigma-\gamma) \cos (\sigma-\delta) 
		\end{equation*}
		\item 
		\begin{equation*}
			\begin{aligned}
				& (\cos (\sigma-\alpha) - \cos (\sigma - \beta))(\cos (\sigma-\gamma) - \cos (\sigma - \delta)) \\
				= &4 \sin \dfrac{\alpha-\beta}{2} \sin \dfrac{\gamma + \delta}{2} \sin \dfrac{\alpha - \beta}{2} \sin \dfrac{\gamma + \delta}{2} \\
				= & (\cos \alpha - \cos \beta)(\cos \gamma - \cos \delta)
			\end{aligned}
		\end{equation*}
		\item 
		\begin{equation*}
			\begin{aligned}
				\sin \alpha \sin \beta - \sin \gamma \sin \delta = \cos(\sigma-\alpha)\cos(\sigma-\beta) - \cos(\sigma-\gamma)\cos(\sigma-\delta) \\
				\cos \alpha \cos \beta - \cos \gamma \cos \delta = \sin(\sigma-\alpha)\sin(\sigma-\beta) - \sin(\sigma-\gamma)\sin(\sigma-\delta) \\
			\end{aligned}
		\end{equation*}
		\item When $\cos \alpha\cos \gamma = \cos \beta \cos \delta$:
		\begin{equation*}
			\begin{gathered}
				\begin{aligned}
					& \quad \sin(\sigma-\beta)\sin(\sigma-\beta-\delta)  = \dfrac{1}{2} \left(\cos \delta - \cos(\alpha+\gamma-\beta) \right) \\
					& = \dfrac{\cos \beta\cos \delta - \cos \beta \cos(\alpha+\gamma-\beta)}{2\cos \beta} = \dfrac{\cos \alpha\cos \gamma - \cos \beta \cos(\alpha+\gamma-\beta)}{2\cos \beta}  \\
					& = \dfrac{\cos (\alpha - \gamma) - \cos (2\beta-\alpha-\gamma)}{4\cos \beta} = \dfrac{\sin (\beta - \alpha) \sin (\beta - \gamma)}{2\cos \beta} 
				\end{aligned} \\
				\sin (\sigma-\alpha) \sin (\sigma-\alpha-\delta) = \dfrac{\sin (\beta - \alpha) \sin (\beta + \gamma)}{2\cos \beta} \\
				\sin (\sigma-\gamma) \sin (\sigma-\gamma-\delta) = \dfrac{\sin (\beta + \alpha) \sin (\beta - \gamma)}{2\cos \beta} \\
				\sin \sigma \sin (\sigma-\delta) = \dfrac{\sin (\beta + \alpha) \sin (\beta + \gamma)}{2\cos \beta} \\
				\sin(\sigma - \alpha)\sin(\sigma - \gamma) = \dfrac{1}{2}(\sin \alpha \sin \gamma + \sin \beta \sin \delta) = \sin(\sigma - \beta)\sin(\sigma - \delta)
			\end{gathered}
		\end{equation*}
	\end{enumerate}
	These identities also hold for all the permutations over $\alpha, ~\beta, ~\gamma, ~\delta$.
\end{prop}

We will also provide some conventions and notations used throughout this note:
\begin{equation} \label{eq: sqrt sign convention}
	\sqrt{a}=
	\begin{cases}
		\begin{aligned}
			& \sqrt{a} ~(a \ge 0) \\
			& i \sqrt{-a} ~(a<0)
		\end{aligned}
	\end{cases}
\end{equation}
\begin{equation*}
	\begin{gathered}
		p_x=\sqrt{\dfrac{\sin \alpha \sin \beta}{\sin (\sigma-\alpha) \sin(\sigma-\beta)}-1}, ~~ p_y=\sqrt{\dfrac{\sin \beta \sin \gamma}{\sin (\sigma-\beta) \sin (\sigma-\gamma)}-1} \\
		p_z=\sqrt{\dfrac{\sin \gamma \sin \delta}{\sin (\sigma- \gamma) \sin (\sigma-\delta)}-1}, ~~ p_w=\sqrt{\dfrac{\sin \delta \sin \alpha}{\sin (\sigma-\delta) \sin (\sigma-\alpha)}-1}
	\end{gathered}
\end{equation*}

We conclude this section by presenting a significant transformation referred to as switching a strip. The term is inspired by a concept proposed by Prof. Ivan Izmestiev:
\begin{prop} \label{prop: switch a strip}
	One-to-one correspondence between two spherical quadrilaterals.
	\begin{enumerate} [label={[\arabic*]}]
		\item $(\alpha, ~\beta, ~\gamma, ~\delta, ~x, ~y, ~z, ~w) \Leftrightarrow (\alpha, ~\beta, ~\pi-\gamma, ~\pi-\delta, ~x, ~-y^{-1}, ~-z, ~-w^{-1})$
		\item $(\alpha, ~\beta, ~\gamma, ~\delta, ~x, ~y, ~z, ~w) \Leftrightarrow (\pi-\alpha, ~\beta, ~\gamma, ~\pi-\delta, ~-x^{-1}, ~y, ~-z^{-1}, ~-w)$ 
		\item $(\alpha, ~\beta, ~\gamma, ~\delta, ~x, ~y, ~z, ~w) \Leftrightarrow (\pi-\alpha, ~\pi - \beta, ~\gamma, ~\delta, ~-x, ~-y^{-1}, ~z, ~-w^{-1})$ 
		\item $(\alpha, ~\beta, ~\gamma, ~\delta, ~x, ~y, ~z, ~w) \Leftrightarrow (\alpha, ~\pi - \beta, ~\pi - \gamma, ~\delta, ~-x^{-1}, ~-y, ~-z^{-1}, ~w)$ 
	\end{enumerate}
\end{prop}

Generically, given a flex of $x$,  there will be two corresponding solutions for $y$ from Equation \eqref{eq: degree-4 vertex adjacent}, two solutions for $z$ based on Equation Equation \eqref{eq: opposite folding angle}, and two solutions for $w$ from Equation \eqref{eq: degree-4 vertex 2}. Determining which combination of $\{x, ~y, ~z, ~w\}$ constitutes the correct solution requires further validation. The discussions in the earlier sections are insufficient to fully address solutions at infinity within this polynomial system. For a comprehensive analysis, it is necessary to bihomogenize the equations mentioned above, as detailed in \citet[Section 6]{he_real_2023}.

The following sections will explore both finite solutions and solutions at infinity, addressing each case individually. 

\section{Square: finite solution} 

The condition on sector angles is:
\begin{equation*}
	f_{22}=0, ~~ f_{20}=0,~~f_{02}=0, ~~f_{00} = 0 ~~ \Leftrightarrow ~~  \alpha=\beta=\gamma=\delta=\pi/2
\end{equation*}
which means:
\begin{equation*}
	\begin{dcases}
		xy = 0 \\
		z = \pm x \\
		xw = 0
	\end{dcases}
\end{equation*}
After post-examination, there are two branches of flexes. Each branch will be diffeomorphic to a circle $S^1$ with the parametrization below:
\begin{equation*}
	\begin{cases}
		x \in \mathbb{R} \cup \{\infty\} \\
		y = 0 \\
		z = x \\
		w = 0
	\end{cases} \mathrm{or} \quad \begin{cases}
		x = 0 \\
		y \in \mathbb{R} \cup \{\infty\} \\
		z = 0 \\
		w = y
	\end{cases} 
\end{equation*}
which are just simply folded along collinear opposite creases.

\section{Rhombus: finite solution}

The condition on sector angles is
\begin{equation*}
	f_{22}=0, ~~ f_{20}=0,~~f_{02}=0, ~~f_{00} \neq 0 ~~ \Leftrightarrow ~~  \alpha=\beta=\gamma=\delta \neq \pi/2
\end{equation*}
which means:
\begin{equation*}
	\begin{aligned}
		\begin{dcases}
			xy = \cos \alpha \\
			z = \pm x \\
			xw = \cos \alpha 
		\end{dcases} ~~ \Rightarrow ~~ & \begin{dcases}
			y = \dfrac{\cos \alpha}{x} \\
			z = \pm x \\
			w = \dfrac{\cos \alpha}{x} 
		\end{dcases}
	\end{aligned}
\end{equation*}
From post-examination, there is a single branch of flex diffeomorphic to a circle $S^1$ with the parametrization below:
\begin{equation*}
	\begin{dcases}
		x \in \mathbb{R} \cup \{\infty\} \\
		y = \dfrac{\cos \alpha}{x} \\
		z = x \\
		w = \dfrac{\cos \alpha}{x} 
	\end{dcases}
\end{equation*}
which is exactly a rhombus on a sphere.

\section{Cross: finite solution}

The condition on sector angles is:
\begin{equation*}
	f_{22} \neq 0, ~~ f_{20}=0,~~f_{02}=0, ~~f_{00} = 0 ~~ \Leftrightarrow ~~  \alpha= \pi - \beta = \gamma = \pi - \delta \neq \pi/2
\end{equation*}
which means:
\begin{equation*}
	\begin{dcases}
		xy(xy \cos \alpha + 1) = 0 \\
		z = \pm x \\
		xw(xw \cos \alpha - 1) = 0
	\end{dcases}
\end{equation*}
After post-examination, there are three branches of flexes. Each branch will be diffeomorphic to a circle $S^1$ with the parametrization below:
\begin{equation*}
	\begin{cases}
		x \in \mathbb{R} \cup \{\infty\} \\
		y = 0 \\
		z = x \\
		w = 0
	\end{cases} \mathrm{or} \quad \begin{cases}
		x = 0 \\
		y \in \mathbb{R} \cup \{\infty\} \\
		z = 0 \\
		w = y
	\end{cases}  \mathrm{or} \quad  \begin{cases}
		x \in \mathbb{R} \cup \{\infty\} \\
		y = -\dfrac{1}{x\cos \alpha} \\
		z = -x \\
		w = \dfrac{1}{x\cos \alpha}
	\end{cases}
\end{equation*}
The first two are just simply folded along collinear opposite creases. The third has a self-intersected `butterfly' shape. 

\section{Miura I: finite solution}

The condition on sector angles is:
\begin{equation*}
	f_{22} = 0, ~~ f_{20} \neq 0,~~f_{02}=0, ~~f_{00} = 0 ~~ \Leftrightarrow ~~  \alpha= \pi - \beta = \pi - \gamma = \delta \neq \pi/2
\end{equation*}
which means:
\begin{equation*}
	\begin{dcases}
		x (x \cos \alpha - y) = 0 \\
		z = \pm x \\
		x (x \cos \alpha + w) = 0
	\end{dcases}
\end{equation*}
After post-examination, there are two branches of flexes. Each branch will be diffeomorphic to a circle $S^1$ with the parametrization below:
\begin{equation*}
	\begin{cases}
		x = 0 \\
		y \in \mathbb{R} \cup \{\infty\} \\
		z = 0 \\
		w = y
	\end{cases}  \mathrm{or} \quad  \begin{cases}
		x \in \mathbb{R} \cup \{\infty\} \\
		y = x\cos \alpha \\
		z = x \\
		w = -x\cos \alpha
	\end{cases}
\end{equation*}

\section{Miura II: finite solution}

The condition on sector angles is:
\begin{equation*}
	f_{22} = 0, ~~ f_{20} = 0,~~f_{02} \neq 0, ~~f_{00} = 0 ~~ \Leftrightarrow ~~  \alpha = \beta = \pi - \gamma = \pi - \delta \neq \pi/2
\end{equation*}
which means:
\begin{equation*}
	\begin{dcases}
		y (y \cos \alpha + x) = 0 \\
		z = \pm x \\
		w (w \cos \alpha + x) = 0
	\end{dcases}
\end{equation*}
After post-examination, there are two branches of flexes. Each branch will be diffeomorphic to a circle $S^1$ with the parametrization below:
\begin{equation*}
	\begin{cases}
		x \in \mathbb{R} \cup \{\infty\} \\
		y = 0 \\
		z = x \\
		w = 0
	\end{cases}  \mathrm{or} \quad  \begin{cases}
		x \in \mathbb{R} \cup \{\infty\} \\
		y = - \dfrac{x}{\cos \alpha} \\
		z = -x \\
		w = - \dfrac{x}{\cos \alpha}
	\end{cases}
\end{equation*}

\section{Isogram: finite solution}

The condition on sector angles is:
\begin{equation*}
	f_{22} \neq 0, ~~ f_{20} = 0,~~f_{02} = 0, ~~f_{00} \neq 0 ~~ \Leftrightarrow ~~ \gamma=\alpha, ~~ \delta=\beta, ~~\beta \neq \alpha, ~~\alpha + \beta \neq \pi
\end{equation*}
which means:
\begin{equation*}
	\begin{aligned}
		\begin{cases}
			\sin(\alpha-\beta) x^2y^2-2 \sin \alpha xy+\sin(\alpha+\beta)=0 \\
			z = \pm x \\
			\sin(\beta-\alpha)x^2w^2-2 \sin \beta xw+\sin (\alpha+\beta) =0 
		\end{cases} 
		\Rightarrow ~~ \begin{dcases}
			xy = \dfrac{\sin \alpha \pm \sin \beta}{\sin (\alpha - \beta)} \\
			z = \pm x \\
			xw = \dfrac{\sin \beta \pm \sin \alpha}{\sin (\beta - \alpha)} \\
		\end{dcases} 
	\end{aligned}
\end{equation*}
From post-examination, there are only two branches. Each branch will be diffeomorphic to a circle $S^1$ with the parametrization below:
\begin{equation*}
	\begin{dcases}
		x \in \mathbb{R} \cup \{\infty\} \\
		y = \dfrac{\cos \frac{\alpha + \beta}{2}}{x\cos \frac{\alpha - \beta}{2}} \\
		z = x \\
		w = \dfrac{\cos \frac{\beta + \alpha}{2}}{x\cos \frac{\beta - \alpha}{2}} 
	\end{dcases} \quad \mathrm{or} \quad \begin{dcases}
		x \in \mathbb{R} \cup \{\infty\} \\
		y = \dfrac{\sin \frac{\alpha + \beta}{2}}{x \sin \frac{\alpha - \beta}{2}} \\
		z = -x \\
		w = \dfrac{\sin \frac{\beta + \alpha}{2}}{x\sin \frac{\beta - \alpha}{2}}
	\end{dcases}	
\end{equation*}
The first branch has a non-self-intersecting convex `car wiper' motion, while the second branch has a constantly self-intersecting `butterfly' shape.  

\section{Anti-isogram: finite solution}

The condition on sector angles is:
\begin{equation*}
	f_{22} = 0, ~~ f_{20} \neq 0,~~f_{02} \neq 0, ~~f_{00} = 0 ~~ \Leftrightarrow ~~ \gamma=\pi - \alpha, ~~ \delta= \pi - \beta, ~~\beta \neq \alpha, ~~\alpha + \beta \neq \pi
\end{equation*}
which means:
\begin{equation*}
	\begin{aligned}
		\begin{cases}
			\sin(\alpha-\beta) x^2 + 2 \sin \alpha xy + \sin(\alpha+\beta) y^2=0 \\
			z = \pm x \\
			\sin(\beta-\alpha) x^2 + 2 \sin \beta xw+\sin (\beta+\alpha) w^2 =0 
		\end{cases} 
		\Rightarrow ~~ \begin{dcases}
			\dfrac{x}{y} = \dfrac{ - \sin \alpha \pm \sin \beta}{\sin (\alpha - \beta)} \\
			z = \pm x \\
			\dfrac{x}{w} = \dfrac{ - \sin \beta \pm \sin \alpha}{\sin (\beta - \alpha)} \\
		\end{dcases} 
	\end{aligned}
\end{equation*}
From post-examination, there are only two branches. Each branch will be diffeomorphic to a circle $S^1$ with the parametrization below:
\begin{equation*}
	\begin{dcases}
		x \in \mathbb{R} \cup \{\infty\} \\
		y = -\dfrac{x\sin \frac{\alpha - \beta}{2}}{\sin \frac{\alpha + \beta}{2}} \\
		z = x \\
		w = -\dfrac{x\sin \frac{\beta - \alpha}{2}}{\sin \frac{\beta + \alpha}{2}} 
	\end{dcases} \quad \mathrm{or} \quad \begin{dcases}
		x \in \mathbb{R} \cup \{\infty\} \\
		y = -\dfrac{x\cos \frac{\alpha - \beta}{2}}{\cos \frac{\alpha + \beta}{2}} \\
		z = -x \\
		w = -\dfrac{x\cos \frac{\beta - \alpha}{2}}{\cos \frac{\beta + \alpha}{2}}
	\end{dcases}	
\end{equation*}

\section{Deltoid I: finite solution}

The condition on sector angles is:	
\begin{equation*}
	f_{22} = 0, ~~ f_{20} \neq 0,~~f_{02} = 0, ~~f_{00} \neq 0 ~~ \Leftrightarrow ~~ \delta=\alpha, ~~ \gamma=\beta, ~~\beta \neq \alpha, ~~\alpha + \beta \neq \pi
\end{equation*}
which means:
\begin{equation*}
	\begin{aligned}
		\begin{cases}
			\sin(\beta-\alpha) x^2-2 \sin \alpha xy+ \sin (\beta+\alpha) =0 \\
			z = \pm x \\
			\sin(\alpha-\beta)x^2-2 \sin \beta xw+ \sin(\alpha+\beta) =0 
		\end{cases} 
		~~ \Rightarrow ~~  \begin{dcases}
			y = \dfrac{\sin(\beta-\alpha)x^2+\sin(\beta+\alpha)}{2 x \sin \alpha } \\
			z = \pm x \\
			w = \dfrac{\sin(\alpha-\beta)x^2 + \sin (\alpha+\beta)}{2 x \sin \beta}
		\end{dcases} 
	\end{aligned}
\end{equation*}
From post-examination, there is only a single branch. Each branch will be diffeomorphic to a circle $S^1$ with the parametrization below:
\begin{equation*}
	\begin{dcases}
		y = \dfrac{\sin(\beta-\alpha)x^2+\sin(\beta+\alpha)}{2 x \sin \alpha } \\
		z = x \\
		w = \dfrac{\sin(\alpha-\beta)x^2+\sin(\alpha+\beta)}{2 x \sin \beta}
	\end{dcases}
\end{equation*}
This branch has a non-self-intersecting convex `car wiper' motion.

\section{Anti-deltoid I: finite solution}

The condition on sector angles is:
\begin{equation*}
	f_{22} \neq 0, ~~ f_{20} = 0,~~f_{02} \neq 0, ~~f_{00} = 0 ~~ \Leftrightarrow ~~ \delta= \pi - \alpha, ~~ \gamma = \pi - \beta, ~~\beta \neq \alpha, ~~\alpha + \beta \neq \pi
\end{equation*}
which means:
\begin{equation*}
	\begin{aligned}
		& \begin{cases}
			\sin(\beta-\alpha) x^2y^2 + 2 \sin \alpha xy+ \sin (\beta+\alpha)y^2 =0 \\
			z = \pm x \\
			\sin(\alpha-\beta) x^2w^2 + 2 \sin \beta xw+ \sin (\alpha+\beta)w^2 =0  
		\end{cases}  \\
		\Rightarrow ~~  & \begin{dcases}
			y = 0 \\
			z = \pm x \\
			w = 0
		\end{dcases} ~~ \mathrm{or} ~~ \begin{dcases}
			y^{-1} = -\dfrac{\sin(\beta-\alpha)x^2+\sin(\beta+\alpha)}{2 x \sin \alpha } \\
			z = \pm x \\
			w^{-1} = -\dfrac{\sin(\alpha-\beta)x^2 + \sin (\alpha+\beta)}{2 x \sin \beta}
		\end{dcases}
	\end{aligned}
\end{equation*}
From post-examination, there is only a single branch. Each branch will be diffeomorphic to a circle $S^1$ with the parametrization below:
\begin{equation*}
	\begin{dcases}
		x \in \mathbb{R} \cup \{\infty\} \\
		y = 0 \\
		z = x \\
		w = 0
	\end{dcases} ~~ \mathrm{or} ~~ \begin{dcases}
		x \in \mathbb{R} \cup \{\infty\} \\
		y^{-1} = -\dfrac{\sin(\beta-\alpha)x^2+\sin(\beta+\alpha)}{2 x \sin \alpha } \\
		z = - x \\
		w^{-1} = -\dfrac{\sin(\alpha-\beta)x^2 + \sin (\alpha+\beta)}{2 x \sin \beta}
	\end{dcases}
\end{equation*}
The expression above can also be obtained after `switching a strip' from Deltoid I.	
\begin{equation*}
	(\alpha, ~\beta, ~\gamma, ~\delta, ~x, ~y, ~z, ~w) \Leftrightarrow (\alpha, ~\beta, ~\pi-\gamma, ~\pi-\delta, ~x, ~-y^{-1}, ~-z, ~-w^{-1})
\end{equation*}

\section{Deltoid II: finite solution}

The condition on sector angles is 	
\begin{equation*}
	f_{22} = 0, ~~ f_{20} = 0,~~f_{02} \neq 0, ~~f_{00} \neq 0 ~~ \Leftrightarrow ~~ \alpha=\beta, ~~ \delta=\gamma, ~~\gamma \neq \beta, ~~\beta + \gamma \neq \pi
\end{equation*}
which means:
\begin{equation*}
	\begin{aligned}
		\begin{cases}
			\sin (\beta-\gamma)y^2-2 \sin \gamma xy + \sin (\beta+\gamma) =0 \\
			\sin^2 \beta (z^2 +1) = \sin^2 \gamma (x^2+1) \\
			\sin (\beta-\gamma)w^2-2 \sin \gamma xw+ \sin (\beta+\gamma) =0 
		\end{cases}
		\Rightarrow ~~ \begin{dcases}
			y = \dfrac{x \sin \gamma \pm \sqrt{\sin^2 \gamma(x^2+1)-\sin^2 \beta}}{\sin(\beta-\gamma)} \\
			z = \pm \dfrac{\sqrt{\sin^2 \beta(x^2+1)-\sin^2 \gamma}}{\sin \beta} \\
			w = \dfrac{x \sin \gamma \pm \sqrt{\sin^2 \gamma(x^2+1)-\sin^2 \beta}}{\sin(\beta-\gamma)}
		\end{dcases} 
	\end{aligned}
\end{equation*}
Only when $x^2 \ge \dfrac{\sin^2 \beta}{\sin ^2 \gamma}-1$ is there a real solution. From post-examination, there are only two solutions:
\begin{equation}
	\begin{dcases}
		y = \dfrac{x \sin \gamma + \sqrt{\sin^2 \gamma(x^2+1)-\sin^2 \beta}}{\sin(\beta-\gamma)} \\
		z = - \dfrac{\sqrt{\sin^2 \beta(x^2+1)-\sin^2 \gamma}}{\sin \beta} \\
		w = \dfrac{x \sin \gamma + \sqrt{\sin^2 \gamma(x^2+1)-\sin^2 \beta}}{\sin(\beta-\gamma)}
	\end{dcases} ~~ \mathrm{or} ~~
	\begin{dcases}
		y = \dfrac{x \sin \gamma - \sqrt{\sin^2 \gamma(x^2+1)-\sin^2 \beta}}{\sin(\beta-\gamma)} \\
		z = \dfrac{\sqrt{\sin^2 \beta(x^2+1)-\sin^2 \gamma}}{\sin \beta} \\
		w = \dfrac{x \sin \gamma - \sqrt{\sin^2 \gamma(x^2+1)-\sin^2 \beta}}{\sin(\beta-\gamma)}
	\end{dcases}
\end{equation}
where
\begin{equation*}
	\begin{aligned}
		\begin{dcases}
			x \in \left(-\infty, -\sqrt{\dfrac{\sin^2 \beta}{\sin^2 \gamma}-1} \right] \bigcup \left[\sqrt{\dfrac{\sin^2 \beta}{\sin^2 \gamma}-1}, +\infty \right) &  \mathrm{when} ~~ \sin \beta> \sin \gamma \\
			x \in \mathbb{R} &  \mathrm{when} ~~ \sin \beta< \sin \gamma
		\end{dcases}
	\end{aligned}
\end{equation*}
A more advanced expression is setting (referring to Subsection \ref{section: sign convention folding angle})
\begin{equation*}
	\begin{aligned}
		p_x=\sqrt{\dfrac{\sin^2 \beta}{\sin^2 \gamma}-1}, ~~
		\begin{dcases}
			x = \pm p_x \cosh s, ~~ s \in (0, +\infty) &  \mathrm{when} ~~ \sin \beta> \sin \gamma \\
			x = \pm i p_x \sinh s, ~~ s \in (0, +\infty) &  \mathrm{when} ~~ \sin \beta< \sin \gamma
		\end{dcases}
	\end{aligned}
\end{equation*}
We set $s \in (0, +\infty)$ in order for sign confirmation when calculating the square root. 

When $\sin \beta > \sin \gamma$,
\begin{equation*}
	\begin{aligned}
		&\begin{dcases}
			x = \sqrt{\dfrac{\sin^2 \beta}{\sin^2 \gamma}-1} \cosh s \\
			y = \sqrt{\dfrac{\tan \beta + \tan \gamma}{\tan \beta - \tan \gamma}} e^{s} \\
			z = -\sqrt{1-\dfrac{\sin^2 \gamma}{\sin^2 \beta }} \sinh s \\
			w= \sqrt{\dfrac{\tan \beta + \tan \gamma}{\tan \beta - \tan \gamma}} e^{s}
		\end{dcases} 
		~~ \mathrm{and} ~~ \begin{dcases}
			x = -\sqrt{\dfrac{\sin^2 \beta }{\sin^2 \gamma}-1} \cosh s \\
			y = -\sqrt{\dfrac{\tan \beta + \tan \gamma}{\tan \beta - \tan \gamma}} e^{s} \\
			z = \sqrt{1-\dfrac{\sin^2 \gamma}{\sin^2 \beta }} \sinh s \\
			w= -\sqrt{\dfrac{\tan \beta + \tan \gamma}{\tan \beta - \tan \gamma}} e^{s}
		\end{dcases} \\
	\end{aligned}
\end{equation*}
\begin{equation*}
	\begin{aligned}
		& ~~ \mathrm{or} ~~ \begin{dcases}
			x = \sqrt{\dfrac{\sin^2 \beta }{\sin^2 \gamma}-1} \cosh s \\
			y = \sqrt{\dfrac{\tan \beta + \tan \gamma}{\tan \beta - \tan \gamma}} e^{-s} \\
			z = \sqrt{1-\dfrac{\sin^2 \gamma}{\sin^2 \beta }} \sinh s \\
			w= \sqrt{\dfrac{\tan \beta + \tan \gamma}{\tan \beta - \tan \gamma}} e^{-s}
		\end{dcases} 
		~~ \mathrm{and} ~~ \begin{dcases}
			x = -\sqrt{\dfrac{\sin^2 \beta }{\sin^2 \gamma}-1} \cosh s \\
			y = -\sqrt{\dfrac{\tan \beta + \tan \gamma}{\tan \beta - \tan \gamma}} e^{-s} \\
			z = -\sqrt{1-\dfrac{\sin^2 \gamma}{\sin^2 \beta }} \sinh s \\
			w= -\sqrt{\dfrac{\tan \beta + \tan \gamma}{\tan \beta - \tan \gamma}} e^{-s}
		\end{dcases} 
	\end{aligned} , ~~ s \in (0, +\infty)
\end{equation*}
When $\sin \beta < \sin \gamma$
\begin{equation*}
	\begin{aligned}
		& \begin{dcases}
			x = -\sqrt{1-\dfrac{\sin^2 \beta }{\sin^2 \gamma}} \sinh s \\
			y = \sqrt{\dfrac{\tan \gamma + \tan \beta}{\tan \gamma - \tan \beta}} e^{s} \\
			z = \sqrt{\dfrac{\sin^2 \gamma}{\sin^2 \beta }-1} \cosh s \\
			w= \sqrt{\dfrac{\tan \gamma + \tan \beta}{\tan \gamma - \tan \beta}} e^{s}
		\end{dcases} 
		~~ \mathrm{and} ~~ \begin{dcases}
			x = \sqrt{1-\dfrac{\sin^2 \beta }{\sin^2 \gamma}} \sinh s \\
			y = -\sqrt{\dfrac{\tan \gamma + \tan \beta}{\tan \gamma - \tan \beta}} e^{s} \\
			z = -\sqrt{\dfrac{\sin^2 \gamma}{\sin^2 \beta }-1} \cosh s \\
			w= -\sqrt{\dfrac{\tan \gamma + \tan \beta}{\tan \gamma - \tan \beta}} e^{s}
		\end{dcases}
	\end{aligned}
\end{equation*}
\begin{equation*}
	\begin{aligned}
		& ~~ \mathrm{or} ~~ \begin{dcases}
			x = \sqrt{1-\dfrac{\sin^2 \beta }{\sin^2 \gamma}} \sinh s \\
			y = \sqrt{\dfrac{\tan \gamma + \tan \beta}{\tan \gamma - \tan \beta}} e^{-s} \\
			z = \sqrt{\dfrac{\sin^2 \gamma}{\sin^2 \beta }-1} \cosh s \\
			w= \sqrt{\dfrac{\tan \gamma + \tan \beta}{\tan \gamma - \tan \beta}} e^{-s}
		\end{dcases}
		~~ \mathrm{and} ~~ \begin{dcases}
			x = -\sqrt{1-\dfrac{\sin^2 \beta }{\sin^2 \gamma}} \sinh s \\
			y = -\sqrt{\dfrac{\tan \gamma + \tan \beta}{\tan \gamma - \tan \beta}} e^{-s} \\
			z = -\sqrt{\dfrac{\sin^2 \gamma}{\sin^2 \beta }-1} \cosh s \\
			w= -\sqrt{\dfrac{\tan \gamma + \tan \beta}{\tan \gamma - \tan \beta}} e^{-s}
		\end{dcases} 
	\end{aligned} , ~~ s \in (0, +\infty)
\end{equation*}

Actually we could write the expressions in a more compact form: 

When $\sin \beta > \sin \gamma$:
\begin{equation*}
	\begin{dcases}
		x = \sqrt{\dfrac{\sin^2 \beta }{\sin^2 \gamma}-1} \cosh s \\
		y = \sqrt{\dfrac{\tan \beta + \tan \gamma}{\tan \beta - \tan \gamma}} e^{-s} \\
		z = \sqrt{1-\dfrac{\sin^2 \gamma}{\sin^2 \beta }} \sinh s \\
		w= \sqrt{\dfrac{\tan \beta + \tan \gamma}{\tan \beta - \tan \gamma}} e^{-s}
	\end{dcases}, ~~ \mathrm{and} ~~
	\begin{dcases}
		x = -\sqrt{\dfrac{\sin^2 \beta }{\sin^2 \gamma}-1} \cosh s \\
		y = -\sqrt{\dfrac{\tan \beta + \tan \gamma}{\tan \beta - \tan \gamma}} 	e^{-s} \\
		z = -\sqrt{1-\dfrac{\sin^2 \gamma}{\sin^2 \beta }} \sinh s \\
		w= -\sqrt{\dfrac{\tan \beta + \tan \gamma}{\tan \beta - \tan \gamma}} e^{-s}
	\end{dcases} , ~~ s \in \mathbb{R} 
\end{equation*}
When $\sin \gamma < \sin \beta$:
\begin{equation*}
	\begin{dcases}
		x = \sqrt{1-\dfrac{\sin^2 \beta }{\sin^2 \gamma}} \sinh s \\
		y = \sqrt{\dfrac{\tan \gamma + \tan \beta}{\tan \gamma - \tan \beta}} e^{-s} \\
		z = \sqrt{\dfrac{\sin^2 \gamma}{\sin^2 \beta }-1} \cosh s \\
		w= \sqrt{\dfrac{\tan \gamma + \tan \beta}{\tan \gamma - \tan \beta}} e^{-s}
	\end{dcases} ~~\mathrm{and} ~~\begin{dcases}
		x = -\sqrt{1-\dfrac{\sin^2 \beta }{\sin^2 \gamma}} \sinh s \\
		y = -\sqrt{\dfrac{\tan \gamma + \tan \beta}{\tan \gamma - \tan \beta}} e^{-s} \\
		z = -\sqrt{\dfrac{\sin^2 \gamma}{\sin^2 \beta }-1} \cosh s \\
		w= -\sqrt{\dfrac{\tan \gamma + \tan \beta}{\tan \gamma - \tan \beta}} e^{-s}
	\end{dcases}, ~~ s \in \mathbb{R} 
\end{equation*}

Further, we want to do an \textbf{analytical continuation} to the above real solutions towards the complex field, which, as we will see later, is a more helpful and unified expression for future symbolic calculations. Let  
\begin{equation*}
	\begin{aligned}
		x=p_x \cos t  
	\end{aligned}
\end{equation*}
When $\sin \beta > \sin \gamma$, 
\begin{equation*}
	\begin{dcases}
		x = \sqrt{\dfrac{\sin^2 \beta }{\sin^2 \gamma}-1} \cos t \\
		y = \sqrt{\dfrac{\tan \beta + \tan \gamma}{\tan \beta - \tan \gamma}} e^{it} \\
		z = -i\sqrt{1-\dfrac{\sin^2 \gamma}{\sin^2 \beta }} \sin t = i\sqrt{1-\dfrac{\sin^2 \gamma}{\sin^2 \beta }} \cos \left(t+\dfrac{\pi}{2}\right) \\
		w= \sqrt{\dfrac{\tan \beta + \tan \gamma}{\tan \beta - \tan \gamma}} e^{it}
	\end{dcases} 	 
\end{equation*}
When $\sin \beta < \sin \gamma$,
\begin{equation*}
	\begin{dcases}
		x = i \sqrt{1-\dfrac{\sin^2 \beta }{\sin^2 \gamma}} \cos t \\
		y = \sqrt{\dfrac{\tan \gamma + \tan \beta}{\tan \gamma - \tan \beta}} e^{it} \\
		z = \sqrt{\dfrac{\sin^2 \gamma}{\sin^2 \beta }-1} \sin t = \sqrt{\dfrac{\sin^2 \gamma}{\sin^2 \beta }-1} \cos \left(t-\dfrac{\pi}{2} \right) \\
		w= \sqrt{\dfrac{\tan \gamma + \tan \beta}{\tan \gamma - \tan \beta}} e^{it}
	\end{dcases} 	 
\end{equation*}
By carefully choosing $t$ for real configurations, with the sign convention in Equation \eqref{section: sign convention folding angle}, it is possible to write the two branches in a unified form. Each branch will be diffeomorphic to $\mathbb{R}$ with the parametrization below:
\begin{equation} \label{eq: spherical deltoid 2}
	\begin{dcases}
		x = p_x \cos t \\
		y = \mathrm{sign}(\pi - \sigma) \sqrt{\dfrac{\tan \beta+\tan \gamma}{\tan \beta - \tan \gamma}} e^{it} \\
		z = p_z \cos \left(t + \dfrac{\pi}{2}\right) \\
		w= \mathrm{sign}(\pi - \sigma) \sqrt{\dfrac{\tan \beta+\tan \gamma}{\tan \beta - \tan \gamma}} e^{it} \\
		\sigma = \beta + \gamma
	\end{dcases} 	 
\end{equation}
\begin{equation*}
	p_x = \sqrt{\dfrac{\sin^2 \beta}{\sin^2 \gamma}-1}, ~~ p_z=\sqrt{\dfrac{\sin^2 \gamma}{\sin^2 \beta}-1}
\end{equation*}
Note that
\begin{equation*}
	\dfrac{\sin (\beta+\gamma)}{\sin(\beta-\gamma)} = \dfrac{\tan \beta+\tan \gamma}{\tan \beta - \tan \gamma}
\end{equation*}
The choices of $t$ for real configurations are provided below:
\bgroup
\def\arraystretch{2}
\begin{center}
	\begin{tabular}{|c|c|c|}
		\hline
		\makecell{Choice of $t$ for \\ real configurations} &  Branch 1 & Branch 2 \\ 
		\hline
		\makecell{$p_x \in \mathbb{R}^+$ \\ $p_z \in i\mathbb{R}^+$}  & \makecell {$is$, $s \in (0, +\infty)$ \\ $\pi + is$, $s \in (0, +\infty)$ \\ when $xz \neq 0$, $xz>0$ \\ this branch is continuous at $x = \infty$ \\ and `snaps' at $x = 0$} & \makecell {$is$, $s \in (-\infty, 0)$ \\ $\pi + is$, $s \in (-\infty, 0)$ \\ when $xz \neq 0$, $xz<0$ \\ this branch is continuous at $x = \infty$ \\ and `snaps' at $x = 0$} \\
		\hline
		\makecell{$p_x \in i\mathbb{R}^+$ \\ $p_z \in \mathbb{R}^+$} & \makecell {$3\pi/2 + is$, $s \in (0, +\infty)$ \\ $\pi/2 + is$, $s \in (0, +\infty)$ \\ when $xz \neq 0$, $xz>0$ \\ this branch is continuous at $x = \infty$ \\ and `snaps' at $x = 0$} & \makecell {$\pi/2 + is$, $s \in (-\infty, 0)$ \\ $3\pi/2 + is$, $s \in (-\infty, 0)$ \\ when $xz \neq 0$, $xz<0$ \\ this branch is continuous at $x = \infty$ \\ and `snaps' at $x = 0$} \\  
		\hline      
	\end{tabular}
\end{center}
\egroup

\section{Anti-deltoid II: finite solution}

The condition on sector angles is 	
\begin{equation*}
	f_{22} \neq 0, ~~ f_{20} \neq 0,~~f_{02} = 0, ~~f_{00} = 0 ~~ \Leftrightarrow ~~ \alpha= \pi - \beta, ~~ \delta= \pi - \gamma, ~~\gamma \neq \beta, ~~\beta + \gamma \neq \pi
\end{equation*}
which means:
\begin{equation*}
	\begin{aligned}
		\begin{cases}
			\sin (\beta-\gamma)x^2y^2+2 \sin \gamma xy + \sin (\beta+\gamma)x^2 =0 \\
			\sin^2 \beta (z^{-2} +1) = \sin^2 \gamma (x^{-2}+1) \\
			\sin (\beta-\gamma)x^2w^2-2 \sin \gamma xw + \sin (\beta+\gamma)x^2 =0
		\end{cases} 
	\end{aligned}
\end{equation*}
By `switching a strip' from Deltoid II:
\begin{equation*}
	(\alpha, ~\beta, ~\gamma, ~\delta, ~x, ~y, ~z, ~w) \Leftrightarrow (\pi-\alpha, ~\beta, ~\gamma, ~\pi-\delta, ~-x^{-1}, ~y, ~-z^{-1}, ~-w)
\end{equation*}
We could also directly write down the two branches. Each branch will be diffeomorphic to $\mathbb{R}$ with the parametrization below:
\begin{equation} \label{eq: spherical anti deltoid 2}
	\begin{dcases}
		x^{-1} = -p_x \cos t \\
		y = \mathrm{sign}(\pi-\sigma) \sqrt{\dfrac{\tan \beta+\tan \gamma}{\tan \beta - \tan \gamma}} e^{it} \\
		z^{-1} = p_z \cos \left(t- \dfrac{\pi}{2}\right) \\
		w= \mathrm{sign}(\sigma-\pi) \sqrt{\dfrac{\tan \beta+\tan \gamma}{\tan \beta - \tan \gamma}} e^{it} \\
		\sigma = \beta + \gamma
	\end{dcases} 	 
\end{equation}
Here we continue to use $\sigma$ to make the expressions consistent with the deltoid II, while here $\sigma$ does not mean the semi-perimeter.
\begin{equation*}
	p_x = \sqrt{\dfrac{\sin^2 \beta}{\sin^2 \gamma}-1}, ~~ p_z=\sqrt{\dfrac{\sin^2 \gamma}{\sin^2 \beta}-1}
\end{equation*}
Note that
\begin{equation*}
	\dfrac{\sin (\beta+\gamma)}{\sin(\beta-\gamma)} = \dfrac{\tan \beta+\tan \gamma}{\tan \beta - \tan \gamma}
\end{equation*}
The choices of $t$ for real configurations are provided below:
\bgroup
\def\arraystretch{2}
\begin{center}
	\begin{tabular}{|c|c|c|}
		\hline
		\makecell{Choice of $t$ for \\ real configurations} &  Branch 1 & Branch 2 \\ 
		\hline
		\makecell{$p_x \in \mathbb{R}^+$ \\ $p_z \in i\mathbb{R}^+$}  & \makecell {$is$, $s \in (0, +\infty) $ \\ $\pi + is$, $s \in (0, +\infty) $ \\ when $xz \neq 0$, $xz>0$ \\ this branch is continuous at $x = 0$ \\ and `snaps' at $x = \infty$} & \makecell {$is$, $s \in (-\infty, 0)$ \\ $\pi + is$, $s \in (-\infty, 0)$ \\ when $xz \neq 0$, $xz<0$ \\ this branch is continuous at $x = 0$ \\ and `snaps' at $x = \infty$} \\
		\hline
		\makecell{$p_x \in i\mathbb{R}^+$ \\ $p_z \in \mathbb{R}^+$} & \makecell {$3\pi/2 + is$, $s \in (0, +\infty) $ \\ $\pi/2 + is$, $s \in (0, +\infty) $ \\ when $xz \neq 0$, $xz>0$ \\ this branch is continuous at $x = 0$ \\ and `snaps' at $x = \infty$} & \makecell {$\pi/2 + is$, $s \in (-\infty, 0)$ \\ $3\pi/2 + is$, $s \in (-\infty, 0)$ \\ when $xz \neq 0$, $xz<0$ \\ this branch is continuous at $x = 0$ \\ and `snaps' at $x = \infty$} \\  
		\hline      
	\end{tabular}
\end{center}
\egroup

\section{Conic I: finite solution}

The condition on sector angles is:
\begin{equation*} 
	f_{22} = 0, ~~ f_{20} \neq 0,~~f_{02} \neq 0, ~~f_{00} \neq 0  ~~ \Leftrightarrow ~~ 
	\begin{dcases}
		\alpha - \beta + \gamma - \delta = 0 \\
		\alpha - \beta - \gamma + \delta \neq 0 \\
		\alpha + \beta - \gamma - \delta \neq 0 \\
		\alpha + \beta + \gamma + \delta \neq 2\pi
	\end{dcases}
\end{equation*}
which implies $\sigma=\alpha+\gamma=\beta+\delta$, and:
\begin{equation*} 
	\begin{aligned}
		\begin{cases}
			y^2 \sin \alpha \sin (\beta-\gamma)-2xy\sin \alpha \sin \gamma + x^2 \sin \gamma \sin (\beta-\alpha)+ \sin \beta \sin (\alpha+\gamma) =0 \\
			\sin \alpha \sin \beta (z^2+1)= \sin \gamma \sin \delta (x^2+1) \\
			\sin \beta w^2 \sin (\alpha-\delta)-2xw\sin \beta \sin \delta + x^2 \sin \delta \sin (\alpha-\beta)+\sin \alpha \sin (\beta+\delta)=0
		\end{cases}
	\end{aligned}
\end{equation*}
First, from the case-by-case discussion on Equation \eqref{eq: opposite folding angle},
\begin{equation*}
	\begin{aligned}
		& \quad \quad ~~ \sin \alpha \sin \beta (z^2+1)=\sin\gamma \sin\delta (x^2+1) \\ & ~~ \Leftrightarrow ~~ x^2 \sin\gamma \sin\delta - z^2 \sin\alpha \sin\beta = \sin\alpha \sin\beta - \sin\gamma \sin\delta
		\\
		& ~~ \Leftrightarrow ~~ \dfrac{x^2 \sin\gamma \sin\delta}{\sin\alpha \sin\beta-\sin\gamma \sin\delta} +  \dfrac{z^2 \sin\alpha \sin\beta }{\sin\gamma \sin\delta -\sin\alpha \sin\beta} = 1  \\
		& ~~ \Leftrightarrow ~~ \dfrac{x^2}{\dfrac{\sin\alpha \sin\beta}{\sin\gamma \sin\delta}-1} + \dfrac{z^2}{\dfrac{\sin\gamma \sin\delta}{\sin\alpha \sin\beta}-1} = 1 \\
		&  ~~ \Leftrightarrow ~~ \dfrac{x^2}{p_x^2} + \dfrac{z^2}{p_z^2} = 1
	\end{aligned}
\end{equation*}
This equation is always hyperbolic. Note that the amplitudes have the form below:
\begin{equation*}
	\begin{gathered}
		p_x=\sqrt{\dfrac{\sin \alpha \sin \beta}{\sin \gamma \sin \delta}-1}, ~~ p_y=\sqrt{\dfrac{\sin \beta \sin \gamma}{\sin \delta \sin \alpha}-1} \\
		p_z=\sqrt{\dfrac{\sin \gamma \sin \delta}{\sin \alpha \sin \beta}-1}, ~~ p_w=\sqrt{\dfrac{\sin \delta \sin \alpha}{\sin \beta \sin \gamma}-1}
	\end{gathered}
\end{equation*}

It is worth mentioning that the magnitudes of $\alpha, ~ \beta, ~\gamma, ~\delta$ plays an important role in determining the signs of all the expressions below, as shown in the following table.

\bgroup
\def\arraystretch{2}
\begin{center}
	\begin{tabular}{|c|c|c|}
		\hline
		& $\sigma < \pi$ & $\sigma > \pi$ \\ 
		\hline
		$\beta = \mathrm{max}(\alpha, ~\beta, ~\gamma, ~\delta)$ & 
		\makecell{$\sin \alpha \sin \beta > \sin \gamma \sin \delta$ \\
			$\sin \beta \sin \gamma > \sin \delta \sin \alpha$ \\
			$\sin \alpha \sin \gamma > \sin \beta \sin \delta$}
		&  \makecell{$\sin \alpha \sin \beta < \sin \gamma \sin \delta$ \\
			$\sin \beta \sin \gamma < \sin \delta \sin \alpha$ \\
			$\sin \alpha \sin \gamma > \sin \beta \sin \delta$} \\
		\hline
		$\alpha = \mathrm{max}(\alpha, ~\beta, ~\gamma, ~\delta)$ & \makecell{$\sin \alpha \sin \beta > \sin \gamma \sin \delta$ \\
			$\sin \beta \sin \gamma < \sin \delta \sin \alpha$ \\
			$\sin \alpha \sin \gamma < \sin \beta \sin \delta$} &  \makecell{$\sin \alpha \sin \beta < \sin \gamma \sin \delta$ \\
			$\sin \beta \sin \gamma > \sin \delta \sin \alpha$ \\
			$\sin \alpha \sin \gamma < \sin \beta \sin \delta$} \\
		\hline
		$\delta= \mathrm{max}(\alpha, ~\beta, ~\gamma, ~\delta)$ & \makecell{$\sin \alpha \sin \beta < \sin \gamma \sin \delta$ \\
			$\sin \beta \sin \gamma < \sin \delta \sin \alpha$ \\
			$\sin \alpha \sin \gamma > \sin \beta \sin \delta$} & \makecell{$\sin \alpha \sin \beta > \sin \gamma \sin \delta$ \\
			$\sin \beta \sin \gamma > \sin \delta \sin \alpha$ \\
			$\sin \alpha \sin \gamma > \sin \beta \sin \delta$} \\
		\hline
		$\gamma = \mathrm{max}(\alpha, ~\beta, ~\gamma, ~\delta)$ & \makecell{$\sin \alpha \sin \beta < \sin \gamma \sin \delta$ \\
			$\sin \beta \sin \gamma > \sin \delta \sin \alpha$ \\
			$\sin \alpha \sin \gamma < \sin \beta \sin \delta$} &  \makecell{$\sin \alpha \sin \beta > \sin \gamma \sin \delta$ \\
			$\sin \beta \sin \gamma < \sin \delta \sin \alpha$ \\
			$\sin \alpha \sin \gamma < \sin \beta \sin \delta$} \\
		\hline
	\end{tabular}
\end{center}
\egroup
This table is derived from Section \ref{section: sign convention folding angle}. For example using:
\begin{equation*}
	\begin{gathered}
		\sin \alpha \sin \beta - \sin \gamma \sin \delta = \sin \sigma \sin (\alpha + \beta - \sigma) \\
		\sin \beta \sin \gamma - \sin \delta \sin \alpha = \sin \sigma \sin (\beta + \gamma - \sigma) \\
		\sin \alpha \sin \gamma - \sin \beta \sin \delta = \sin (\alpha + \beta - \sigma) \sin (\beta + \gamma - \sigma) \\
	\end{gathered}
\end{equation*}

We will take one of the above cases as an example to show the derivation, when $\beta = \mathrm{max}$ and $\sigma < \pi$,
\begin{equation*}
	x = \pm \sqrt{\dfrac{\sin \alpha \sin\beta}{\sin\gamma \sin\delta}-1} \cosh s, \quad z = \pm \sqrt{1-\dfrac{\sin\gamma \sin\delta}{\sin\alpha \sin\beta}} \sinh s, ~~ s \in (0, +\infty)
\end{equation*}
Substitute $x$ into the quadratic equation with respect to $y$ and $w$, we have
\begin{equation*}
	\begin{gathered}
		\sin \alpha \sin (\beta-\gamma) y^2 \mp 2y\sin \alpha \sin \gamma \sqrt{\dfrac{\sin \alpha \sin \beta}{\sin \gamma \sin\delta}-1} \cosh s \\ + \dfrac{\sin(\beta-\alpha)(\sin\alpha \sin\beta -\sin\gamma \sin\delta) \cosh^2 s}{\sin\delta}+\sin\beta\sin(\alpha+\gamma) =0 \\
		\sin \beta \sin (\alpha-\delta)y^2 \mp 2\sin\beta\sin\delta y \sqrt{\dfrac{\sin\alpha \sin\beta}{\sin\gamma \sin\delta}-1} \cosh s \\ + \dfrac{\sin(\alpha-\beta)(\sin\alpha \sin\beta -\sin\gamma \sin\delta) \cosh^2 s}{\sin\beta}+\sin\delta\sin(\beta+\delta) =0
	\end{gathered}
\end{equation*}
and the solutions are 
\begin{equation*}
	\begin{gathered}
		y= \sqrt{\dfrac{\sin(\beta+\delta)}{\sin(\beta-\gamma)}} \left( \pm \sqrt{\dfrac{\sin\gamma}{\sin\delta}}  \cosh s \pm \sqrt{\dfrac{\sin\beta}{\sin\alpha}} \sinh s \right) \\
		w= \sqrt{\dfrac{\sin (\alpha+\gamma)}{\sin(\alpha-\delta)}} \left( \pm \sqrt{\dfrac{\sin\delta}{\sin\gamma}}  \cosh s \pm \sqrt{\dfrac{\sin\alpha}{\sin\beta}} \sinh s \right) 
	\end{gathered}
\end{equation*}
Take one sign choice of $y$ and $w$ as an example:
\begin{equation*}
	\begin{aligned}
		y_1 & = \sqrt{\dfrac{\sin (\beta+\delta)}{\sin(\beta-\gamma)}} \left( \sqrt{\dfrac{\sin \gamma}{\sin \delta}}  \cosh s + \sqrt{\dfrac{\sin \beta}{\sin \alpha}} \sinh s \right) \\
		& = \sqrt{\dfrac{\sin (\beta+\delta)}{\sin(\beta-\gamma)}\left(\dfrac{\sin \gamma}{\sin \delta}-\dfrac{\sin \beta}{\sin \alpha}\right)} \left(\sqrt{\dfrac{\sin \alpha \sin\gamma}{ \sin\alpha \sin\gamma - \sin\beta \sin\delta}} \cosh s \right.\\
		& \quad + \left. \sqrt{\dfrac{\sin\beta \sin\delta}{ \sin\alpha \sin\gamma - \sin\beta \sin\delta}} \sinh s \right)  \\
		& = \sqrt{\dfrac{\sin \beta \sin \gamma}{\sin \delta \sin \alpha }-1} \left(\sqrt{\dfrac{\sin \alpha \sin \gamma}{ \sin\alpha \sin\gamma - \sin\beta \sin\delta}} \cosh s \right. \\
		& \quad + \left. \sqrt{\dfrac{\sin\beta \sin\delta}{ \sin\alpha \sin\gamma - \sin\beta \sin\delta}} \sinh s \right) \\
		& = \sqrt{\dfrac{\sin\beta \sin\gamma}{\sin\delta \sin\alpha}-1} \cosh (s + \theta_1') \\
	\end{aligned}
\end{equation*}
\begin{equation*}
	\begin{aligned}
		w_1 & = \sqrt{\dfrac{\sin (\alpha+\gamma)}{\sin (\alpha-\delta)}} \left( \sqrt{\dfrac{\sin \delta}{\sin \gamma}}  \cosh s + \sqrt{\dfrac{\sin \alpha}{\sin \beta}} \sinh s \right) \\
		& = \sqrt{\dfrac{\sin (\alpha+\gamma)}{\sin (\alpha-\delta) }\left(\dfrac{\sin\alpha}{\sin\beta}-\dfrac{\sin\delta}{\sin\gamma}\right)} \left(\sqrt{\dfrac{\sin\beta \sin\delta}{ \sin\alpha \sin\gamma - \sin\beta \sin\delta}} \cosh s \right.\\ & \left. \quad + \sqrt{\dfrac{\sin\alpha \sin\gamma}{ \sin\alpha \sin\gamma - \sin\beta \sin\delta}} \sinh s \right)  \\
		& = \sqrt{1-\dfrac{\sin \delta \sin \alpha}{\sin \beta \sin \gamma}} \left(\sqrt{\dfrac{\sin \beta \sin \delta}{ \sin\alpha \sin\gamma - \sin\beta \sin\delta}} \cosh s + \right. \\ & \quad \left. \sqrt{\dfrac{\sin\alpha \sin\gamma}{ \sin\alpha \sin\gamma - \sin\beta \sin\delta}} \sinh s \right) \\
		& = \sqrt{1-\dfrac{\sin\delta \sin\alpha}{\sin\beta \sin\gamma}} \sinh (s + \theta_1')
	\end{aligned}
\end{equation*}
\begin{equation*}
	\tanh \theta_1' = \sqrt{\dfrac{\sin \beta \sin \delta}{\sin \alpha \sin \gamma}} ~~ \Rightarrow ~~ \theta_1' = \dfrac{1}{2} \ln \left( \dfrac{\sqrt{\sin \beta \sin \delta} + \sqrt{\sin \alpha \sin \gamma}}{\sqrt{\sin \beta \sin \delta} - \sqrt{\sin \alpha \sin \gamma}}\right)
\end{equation*}
We use $\theta_1'$ here since the phase shift $\theta_1$ will be defined later. After post-examination on sign choices, the solutions are $s \in (0, +\infty)$:
\begin{equation*}
	\begin{aligned}
		&\begin{dcases}
			x = \sqrt{\dfrac{\sin \alpha \sin \beta}{\sin \gamma \sin \delta}-1} \cosh s \\
			y = \sqrt{\dfrac{\sin \beta \sin \gamma}{\sin \delta \sin \alpha }-1} \cosh (s + \theta_1') \\
			z = -\sqrt{1-\dfrac{\sin \gamma \sin \delta}{\sin \alpha \sin \beta}} \sinh s \\
			w= \sqrt{1-\dfrac{\sin \delta\sin  \alpha}{\sin \beta \sin \gamma}} \sinh (s + \theta_1')
		\end{dcases} 
		~~ \mathrm{and} ~~ \begin{dcases}
			x = -\sqrt{\dfrac{\sin \alpha \sin \beta}{\sin \gamma \sin \delta}-1} \cosh s \\
			y = -\sqrt{\dfrac{\sin \beta \sin \gamma}{\sin \delta \sin \alpha }-1} \cosh (s + \theta_1') \\
			z = \sqrt{1-\dfrac{\sin \gamma \sin \delta}{\sin \alpha \sin \beta}} \sinh s \\
			w= -\sqrt{1-\dfrac{\sin \delta \sin \alpha}{\sin \beta \sin \gamma}} \sinh (s + \theta_1')
		\end{dcases}
	\end{aligned}
\end{equation*}
\begin{equation*}
	\begin{aligned}
		\mathrm{and} ~~ \begin{dcases}
			x = \sqrt{\dfrac{\sin \alpha \sin \beta}{\sin \gamma \sin \delta}-1} \cosh s \\
			y = \sqrt{\dfrac{\sin \beta \sin \gamma}{\sin \delta \sin \alpha }-1} \cosh (-s + \theta_1') \\
			z = \sqrt{1-\dfrac{\sin \gamma \sin \delta}{\sin \alpha \sin \beta}} \sinh s \\
			w= \sqrt{1-\dfrac{\sin \delta \sin \alpha}{\sin \beta \sin \gamma}} \sinh (-s + \theta_1')
		\end{dcases} 
		~~ \mathrm{and} ~~ \begin{dcases}
			x = -\sqrt{\dfrac{\sin \alpha \sin \beta}{\sin \gamma \sin \delta}-1} \cosh s \\
			y = -\sqrt{\dfrac{\sin \beta \sin \gamma}{\sin \delta \sin \alpha }-1} \cosh (-s + \theta_1') \\
			z = -\sqrt{1-\dfrac{\sin \gamma \sin \delta}{\sin \alpha\sin  \beta}} \sinh s \\
			w= -\sqrt{1-\dfrac{\sin\delta \sin\alpha}{\sin\beta \sin\gamma}} \sinh (-s + \theta_1')
		\end{dcases} 
	\end{aligned}
\end{equation*}

Note that when $\sigma > \pi$ there is a sign change, which is shown below. Using the same technique as Deltoid II, after listing all the cases and apply analytical continuation we could directly write the final expressions:
\begin{equation}
	\begin{dcases}
		x = p_x \cos t \\
		y = \mathrm{sign}(\pi - \sigma) p_y \cos \left(t - \theta_1 \right) \\
		z = p_z \cos \left(t + \dfrac{\pi}{2}\right) \\
		w = \mathrm{sign}(\pi - \sigma) p_w \cos \left(t - \theta_2 \right) \\
	\end{dcases}
\end{equation}

\bgroup
\def\arraystretch{2}
\begin{center}
	\begin{tabular}{|c|c|c|}
		\hline
		\makecell{Choice of $t$ for \\ real configurations} &  Branch 1 & Branch 2 \\ 
		\hline
		\makecell{$p_x \in \mathbb{R}^+$ \\ $p_z \in i\mathbb{R}^+$}  & \makecell {$is$, $s \in (0, +\infty) $ \\ $\pi + is$, $s \in (0, +\infty) $ \\ when $xz \neq 0$, $xz>0$ \\ this branch is continuous at $x = \infty$ \\ and `snaps' at $x = 0$} & \makecell {$is$, $s \in (-\infty, 0)$ \\ $\pi + is$, $s \in (-\infty, 0)$ \\ when $xz \neq 0$, $xz<0$ \\ this branch is continuous at $x = \infty$ \\ and `snaps' at $x = 0$} \\
		\hline
		\makecell{$p_x \in i\mathbb{R}^+$ \\ $p_z \in \mathbb{R}^+$} & \makecell {$3\pi/2 + is$, $s \in (0, +\infty) $ \\ $\pi/2 + is$, $s \in (0, +\infty) $ \\ when $xz \neq 0$, $xz>0$ \\ this branch is continuous at $x = \infty$ \\ and `snaps' at $x = 0$} & \makecell {$\pi/2 + is$, $s \in (-\infty, 0)$ \\ $3\pi/2 + is$, $s \in (-\infty, 0)$ \\ when $xz \neq 0$, $xz<0$ \\ this branch is continuous at $x = \infty$ \\ and `snaps' at $x = 0$} \\  
		\hline      
	\end{tabular}
\end{center}
\egroup

We could clearly see that in the complex field, $x, ~y, ~ z, ~w$ are trigonometric functions with specific amplitudes and phase shifts. Further,
\begin{equation*}
	\begin{gathered}
		\tan \theta_1 = i \sqrt{\dfrac{\sin \beta \sin \delta}{\sin \alpha \sin \gamma}} \\
		\tan \theta_2 = i \sqrt{\dfrac{\sin \alpha \sin \gamma}{\sin \beta \sin \delta}}
	\end{gathered}
\end{equation*}
\bgroup
\def\arraystretch{2}
\begin{center}
	\begin{tabular}{|c|c|c|}
		\hline
		&  $\theta_1$ & $\theta_2$ \\
		\hline
		$p_x \in \mathbb{R}^+, ~p_y \in \mathbb{R}^+$ & $ \ln \left|\dfrac{\sqrt{\sin \alpha \sin \gamma}+\sqrt{\sin \beta \sin \delta}}{\sqrt{\sin \alpha \sin \gamma}-\sqrt{\sin \beta \sin \delta}}\right|$ & $\dfrac{\pi}{2} + \theta_1$ \\
		\hline
		$p_y \in \mathbb{R}^+, ~p_z \in \mathbb{R}^+$ & $-\dfrac{\pi}{2} + \ln \left|\dfrac{\sqrt{\sin \alpha \sin \gamma}+\sqrt{\sin \beta \sin \delta}}{\sqrt{\sin \alpha \sin \gamma}-\sqrt{\sin \beta \sin \delta}}\right|$ & $\dfrac{\pi}{2} + \theta_1$\\ 
		\hline 
		$p_z \in \mathbb{R}^+, ~p_w \in \mathbb{R}^+$ & $ \ln \left|\dfrac{\sqrt{\sin \alpha \sin \gamma}+\sqrt{\sin \beta \sin \delta}}{\sqrt{\sin \alpha \sin \gamma}-\sqrt{\sin \beta \sin \delta}}\right|$ & $-\dfrac{\pi}{2} + \theta_1$ \\
		\hline
		$p_w \in \mathbb{R}^+, ~p_x \in \mathbb{R}^+$ & $ \dfrac{\pi}{2} + \ln \left|\dfrac{\sqrt{\sin \alpha \sin \gamma}+\sqrt{\sin \beta \sin \delta}}{\sqrt{\sin \alpha \sin \gamma}-\sqrt{\sin \beta \sin \delta}}\right|$ & $-\dfrac{\pi}{2} + \theta_1$ \\
		\hline      
	\end{tabular}
\end{center}
\egroup

\section{Conic II: finite solution}

The condition on sector angles is 
\begin{equation*} 
	f_{22} \neq 0, ~~ f_{20} = 0,~~f_{02} \neq 0, ~~f_{00} \neq 0  ~~ \Leftrightarrow ~~ 
	\begin{dcases}
		\alpha - \beta + \gamma - \delta \neq 0 \\
		\alpha - \beta - \gamma + \delta = 0 \\
		\alpha + \beta - \gamma - \delta \neq 0 \\
		\alpha + \beta + \gamma + \delta \neq 2\pi
	\end{dcases}
\end{equation*}
We could transfer the result for Conic I to Conic II from switching a strip:
\begin{equation}
	(\alpha, ~\beta, ~\gamma, ~\delta, ~x, ~y, ~z, ~w) \rightarrow (\alpha, ~\beta, ~\pi-\gamma, ~\pi-\delta, ~x, ~-y^{-1}, ~-z, ~-w^{-1}) \\
\end{equation}
Now:
\begin{equation*}
	\sigma = \dfrac{\alpha + \beta - \gamma - \delta}{2} + \pi
\end{equation*}
Another thing needs to be noted is the switch of branch 1 and branch 2 when switching a strip.
\begin{equation}
	\begin{dcases}
		x = p_x \cos t \\
		y^{-1} = \mathrm{sign}(\sigma - \pi) p_y \cos \left(t - \theta_1 \right) \\
		z = p_z \cos \left(t - \dfrac{\pi}{2}\right) \\
		w^{-1} = \mathrm{sign}(\sigma - \pi) p_w \cos \left(t - \theta_2 \right) \\
	\end{dcases}
\end{equation}
\begin{equation*}
	\begin{gathered}
		p_x=\sqrt{\dfrac{\sin \alpha \sin \beta}{\sin \gamma \sin \delta}-1}, ~~ p_y=\sqrt{\dfrac{\sin \beta \sin \gamma}{\sin \delta \sin \alpha}-1} \\
		p_z=\sqrt{\dfrac{\sin \gamma \sin \delta}{\sin \alpha \sin \beta}-1}, ~~ p_w=\sqrt{\dfrac{\sin \delta \sin \alpha}{\sin \beta \sin \gamma}-1}
	\end{gathered}
\end{equation*}

\bgroup
\def\arraystretch{2}
\begin{center}
	\begin{tabular}{|c|c|c|}
		\hline
		\makecell{Choice of $t$ for \\ real configurations} &  Branch 1 & Branch 2 \\ 
		\hline
		\makecell{$p_x \in \mathbb{R}^+$ \\ $p_z \in i\mathbb{R}^+$}  & \makecell {$is$, $s \in (-\infty, 0)$ \\ $\pi + is$, $s \in (-\infty, 0)$ \\ when $xz \neq 0$, $xz>0$ \\ this branch is continuous at $x = \infty$ \\ and `snaps' at $x = 0$} & \makecell {$is$, $s \in (0, +\infty) $ \\ $\pi + is$, $s \in (0, +\infty)$ \\ when $xz \neq 0$, $xz<0$ \\ this branch is continuous at $x = \infty$ \\ and `snaps' at $x = 0$} \\
		\hline
		\makecell{$p_x \in i\mathbb{R}^+$ \\ $p_z \in \mathbb{R}^+$} & \makecell {$\pi/2 + is$, $s \in (-\infty, 0)$ \\ $3\pi/2 + is$, $s \in (-\infty, 0)$ \\ when $xz \neq 0$, $xz>0$ \\ this branch is continuous at $x = \infty$ \\ and `snaps' at $x = 0$} & \makecell {$3\pi/2 + is$, $s \in (0, +\infty) $ \\ $\pi/2 + is$, $s \in (0, +\infty) $ \\ when $xz \neq 0$, $xz<0$ \\ this branch is continuous at $x = \infty$ \\ and `snaps' at $x = 0$}\\  
		\hline      
	\end{tabular}
\end{center}
\egroup
For the phase shift:
\bgroup
\def\arraystretch{2}
\begin{center}
	\begin{tabular}{|c|c|c|}
		\hline
		&  $\theta_1$ & $\theta_2$ \\
		\hline
		$p_x \in \mathbb{R}^+, ~p_y \in \mathbb{R}^+$ & $ \ln \left|\dfrac{\sqrt{\sin \alpha \sin \gamma}+\sqrt{\sin \beta \sin \delta}}{\sqrt{\sin \alpha \sin \gamma}-\sqrt{\sin \beta \sin \delta}}\right|$ & $\dfrac{\pi}{2} + \theta_1$ \\
		\hline
		$p_y \in \mathbb{R}^+, ~p_z \in \mathbb{R}^+$ & $-\dfrac{\pi}{2} + \ln \left|\dfrac{\sqrt{\sin \alpha \sin \gamma}+\sqrt{\sin \beta \sin \delta}}{\sqrt{\sin \alpha \sin \gamma}-\sqrt{\sin \beta \sin \delta}}\right|$ & $\dfrac{\pi}{2} + \theta_1$\\ 
		\hline 
		$p_z \in \mathbb{R}^+, ~p_w \in \mathbb{R}^+$ & $ \ln \left|\dfrac{\sqrt{\sin \alpha \sin \gamma}+\sqrt{\sin \beta \sin \delta}}{\sqrt{\sin \alpha \sin \gamma}-\sqrt{\sin \beta \sin \delta}}\right|$ & $-\dfrac{\pi}{2} + \theta_1$ \\
		\hline
		$p_w \in \mathbb{R}^+, ~p_x \in \mathbb{R}^+$ & $ \dfrac{\pi}{2} + \ln \left|\dfrac{\sqrt{\sin \alpha \sin \gamma}+\sqrt{\sin \beta \sin \delta}}{\sqrt{\sin \alpha \sin \gamma}-\sqrt{\sin \beta \sin \delta}}\right|$ & $-\dfrac{\pi}{2} + \theta_1$ \\
		\hline      
	\end{tabular}
\end{center}
\egroup

Please refer to the table below on how to determine the maximum sector angle from the range of amplitudes $p_x, ~p_y, ~p_z, ~p_w$.
\bgroup
\def\arraystretch{2}
\begin{center}
	\begin{tabular}{|c|c|c|}
		\hline
		& $\sigma < \pi$ & $\sigma > \pi$ \\ 
		\hline
		$\beta = \mathrm{max}(\alpha, ~\beta, ~\pi - \gamma, ~\pi - \delta)$ & 
		\makecell{$\sin \alpha \sin \beta > \sin \gamma \sin \delta$ \\
			$\sin \beta \sin \gamma > \sin \delta \sin \alpha$ \\
			$\sin \alpha \sin \gamma > \sin \beta \sin \delta$}
		&  \makecell{$\sin \alpha \sin \beta < \sin \gamma \sin \delta$ \\
			$\sin \beta \sin \gamma < \sin \delta \sin \alpha$ \\
			$\sin \alpha \sin \gamma > \sin \beta \sin \delta$} \\
		\hline
		$\alpha = \mathrm{max}(\alpha, ~\beta, ~\pi - \gamma, ~\pi - \delta)$ & \makecell{$\sin \alpha \sin \beta > \sin \gamma \sin \delta$ \\
			$\sin \beta \sin \gamma < \sin \delta \sin \alpha$ \\
			$\sin \alpha \sin \gamma < \sin \beta \sin \delta$} &  \makecell{$\sin \alpha \sin \beta < \sin \gamma \sin \delta$ \\
			$\sin \beta \sin \gamma > \sin \delta \sin \alpha$ \\
			$\sin \alpha \sin \gamma < \sin \beta \sin \delta$} \\
		\hline
		$\pi - \delta= \mathrm{max}(\alpha, ~\beta, ~\pi - \gamma, ~\pi - \delta)$ & \makecell{$\sin \alpha \sin \beta < \sin \gamma \sin \delta$ \\
			$\sin \beta \sin \gamma < \sin \delta \sin \alpha$ \\
			$\sin \alpha \sin \gamma > \sin \beta \sin \delta$} & \makecell{$\sin \alpha \sin \beta > \sin \gamma \sin \delta$ \\
			$\sin \beta \sin \gamma > \sin \delta \sin \alpha$ \\
			$\sin \alpha \sin \gamma > \sin \beta \sin \delta$} \\
		\hline
		$\pi - \gamma = \mathrm{max}(\alpha, ~\beta, ~\pi - \gamma, ~\pi - \delta)$ & \makecell{$\sin \alpha \sin \beta < \sin \gamma \sin \delta$ \\
			$\sin \beta \sin \gamma > \sin \delta \sin \alpha$ \\
			$\sin \alpha \sin \gamma < \sin \beta \sin \delta$} &  \makecell{$\sin \alpha \sin \beta > \sin \gamma \sin \delta$ \\
			$\sin \beta \sin \gamma < \sin \delta \sin \alpha$ \\
			$\sin \alpha \sin \gamma < \sin \beta \sin \delta$} \\
		\hline
	\end{tabular}
\end{center}
\egroup

\section{Conic III: finite solution}

The condition on sector angles is 
\begin{equation*} 
	f_{22} \neq 0, ~~ f_{20} \neq 0,~~f_{02} = 0, ~~f_{00} \neq 0  ~~ \Leftrightarrow ~~ 
	\begin{dcases}
		\alpha - \beta + \gamma - \delta \neq 0 \\
		\alpha - \beta - \gamma + \delta \neq 0 \\
		\alpha + \beta - \gamma - \delta = 0 \\
		\alpha + \beta + \gamma + \delta \neq 2\pi
	\end{dcases}
\end{equation*}
We could transfer the result for Conic I to Conic III from switching a strip:
\begin{equation}
	(\alpha, ~\beta, ~\gamma, ~\delta, ~x, ~y, ~z, ~w) \rightarrow (\pi-\alpha, ~\beta, ~\gamma, ~\pi-\delta, ~-x^{-1}, ~y, ~-z^{-1}, ~-w) \\
\end{equation}
Now:
\begin{equation*}
	\sigma = \dfrac{-\alpha + \beta + \gamma - \delta}{2} + \pi
\end{equation*}
and we could directly write the result for Conic III with a little careful checking on the signs of the expressions:
\begin{equation}
	\begin{dcases}
		x^{-1} = -p_x \cos t \\
		y = \mathrm{sign}(\pi - \sigma) p_y \cos \left(t - \theta_1 \right) \\
		z^{-1} = p_z \cos \left(t - \dfrac{\pi}{2}\right) \\
		w = \mathrm{sign}(\sigma - \pi) p_w \cos \left(t - \theta_2 \right) \\ 
	\end{dcases}
\end{equation}
\begin{equation*}
	\begin{gathered}
		p_x=\sqrt{\dfrac{\sin \alpha \sin \beta}{\sin \gamma \sin \delta}-1}, ~~ p_y=\sqrt{\dfrac{\sin \beta \sin \gamma}{\sin \delta \sin \alpha}-1} \\
		p_z=\sqrt{\dfrac{\sin \gamma \sin \delta}{\sin \alpha \sin \beta}-1}, ~~ p_w=\sqrt{\dfrac{\sin \delta \sin \alpha}{\sin \beta \sin \gamma}-1}
	\end{gathered}
\end{equation*}

\bgroup
\def\arraystretch{2}
\begin{center}
	\begin{tabular}{|c|c|c|}
		\hline
		\makecell{Choice of $t$ for \\ real configurations} &  Branch 1 & Branch 2 \\ 
		\hline
		\makecell{$p_x \in \mathbb{R}^+$ \\ $p_z \in i\mathbb{R}^+$}  & \makecell {$is$, $s \in (0, +\infty)$ \\ $\pi + is$, $s \in (0, +\infty) $ \\ when $xz \neq 0$, $xz>0$ \\ this branch is continuous at $x = 0$ \\ and `snaps' at $x = \infty$} & \makecell {$is$, $s \in (-\infty, 0)$ \\ $\pi + is$, $s \in (-\infty, 0)$ \\ when $xz \neq 0$, $xz<0$ \\ this branch is continuous at $x = 0$ \\ and `snaps' at $x = \infty$} \\
		\hline
		\makecell{$p_x \in i\mathbb{R}^+$ \\ $p_z \in \mathbb{R}^+$} & \makecell {$3\pi/2 + is$, $s \in (0, +\infty)$ \\ $\pi/2 + is$, $s \in (0, +\infty)$ \\ when $xz \neq 0$, $xz>0$ \\ this branch is continuous at $x = 0$ \\ and `snaps' at $x = \infty$} & \makecell {$\pi/2 + is$, $s \in (-\infty, 0)$ \\ $3\pi/2 + is$, $s \in (-\infty, 0)$ \\ when $xz \neq 0$, $xz<0$ \\ this branch is continuous at $x = 0$ \\ and `snaps' at $x = \infty$} \\  
		\hline      
	\end{tabular}
\end{center}
\egroup
For the phase shift:
\bgroup
\def\arraystretch{2}
\begin{center}
	\begin{tabular}{|c|c|c|}
		\hline
		&  $\theta_1$ & $\theta_2$ \\
		\hline
		$p_x \in \mathbb{R}^+, ~p_y \in \mathbb{R}^+$ & $ \ln \left|\dfrac{\sqrt{\sin \alpha \sin \gamma}+\sqrt{\sin \beta \sin \delta}}{\sqrt{\sin \alpha \sin \gamma}-\sqrt{\sin \beta \sin \delta}}\right|$ & $\dfrac{\pi}{2} + \theta_1$ \\
		\hline
		$p_y \in \mathbb{R}^+, ~p_z \in \mathbb{R}^+$ & $-\dfrac{\pi}{2} + \ln \left|\dfrac{\sqrt{\sin \alpha \sin \gamma}+\sqrt{\sin \beta \sin \delta}}{\sqrt{\sin \alpha \sin \gamma}-\sqrt{\sin \beta \sin \delta}}\right|$ & $\dfrac{\pi}{2} + \theta_1$\\ 
		\hline 
		$p_z \in \mathbb{R}^+, ~p_w \in \mathbb{R}^+$ & $ \ln \left|\dfrac{\sqrt{\sin \alpha \sin \gamma}+\sqrt{\sin \beta \sin \delta}}{\sqrt{\sin \alpha \sin \gamma}-\sqrt{\sin \beta \sin \delta}}\right|$ & $-\dfrac{\pi}{2} + \theta_1$ \\
		\hline
		$p_w \in \mathbb{R}^+, ~p_x \in \mathbb{R}^+$ & $ \dfrac{\pi}{2} + \ln \left|\dfrac{\sqrt{\sin \alpha \sin \gamma}+\sqrt{\sin \beta \sin \delta}}{\sqrt{\sin \alpha \sin \gamma}-\sqrt{\sin \beta \sin \delta}}\right|$ & $-\dfrac{\pi}{2} + \theta_1$ \\
		\hline      
	\end{tabular}
\end{center}
\egroup

Please refer to the table below on how to determine the maximum sector angle from the range of amplitudes $p_x, ~p_y, ~p_z, ~p_w$.
\bgroup
\def\arraystretch{2}
\begin{center}
	\begin{tabular}{|c|c|c|}
		\hline
		& $\sigma < \pi$ & $\sigma > \pi$ \\ 
		\hline
		$\beta = \mathrm{max}(\pi - \alpha, ~\beta, ~\gamma, ~\pi - \delta)$ & 
		\makecell{$\sin \alpha \sin \beta > \sin \gamma \sin \delta$ \\
			$\sin \beta \sin \gamma > \sin \delta \sin \alpha$ \\
			$\sin \alpha \sin \gamma > \sin \beta \sin \delta$}
		&  \makecell{$\sin \alpha \sin \beta < \sin \gamma \sin \delta$ \\
			$\sin \beta \sin \gamma < \sin \delta \sin \alpha$ \\
			$\sin \alpha \sin \gamma > \sin \beta \sin \delta$} \\
		\hline
		$\pi - \alpha = \mathrm{max}(\pi - \alpha, ~\beta, ~\gamma, ~\pi - \delta)$ & \makecell{$\sin \alpha \sin \beta > \sin \gamma \sin \delta$ \\
			$\sin \beta \sin \gamma < \sin \delta \sin \alpha$ \\
			$\sin \alpha \sin \gamma < \sin \beta \sin \delta$} &  \makecell{$\sin \alpha \sin \beta < \sin \gamma \sin \delta$ \\
			$\sin \beta \sin \gamma > \sin \delta \sin \alpha$ \\
			$\sin \alpha \sin \gamma < \sin \beta \sin \delta$} \\
		\hline
		$\pi - \delta= \mathrm{max}(\pi - \alpha, ~\beta, ~\gamma, ~\pi - \delta)$ & \makecell{$\sin \alpha \sin \beta < \sin \gamma \sin \delta$ \\
			$\sin \beta \sin \gamma < \sin \delta \sin \alpha$ \\
			$\sin \alpha \sin \gamma > \sin \beta \sin \delta$} & \makecell{$\sin \alpha \sin \beta > \sin \gamma \sin \delta$ \\
			$\sin \beta \sin \gamma > \sin \delta \sin \alpha$ \\
			$\sin \alpha \sin \gamma > \sin \beta \sin \delta$} \\
		\hline
		$\gamma = \mathrm{max}(\pi - \alpha, ~\beta, ~\gamma, ~\pi - \delta)$ & \makecell{$\sin \alpha \sin \beta < \sin \gamma \sin \delta$ \\
			$\sin \beta \sin \gamma > \sin \delta \sin \alpha$ \\
			$\sin \alpha \sin \gamma < \sin \beta \sin \delta$} &  \makecell{$\sin \alpha \sin \beta > \sin \gamma \sin \delta$ \\
			$\sin \beta \sin \gamma < \sin \delta \sin \alpha$ \\
			$\sin \alpha \sin \gamma < \sin \beta \sin \delta$} \\
		\hline
	\end{tabular}
\end{center}
\egroup

\section{Conic IV: real solution}

The condition on sector angles is:
\begin{equation*} 
	f_{22} \neq 0, ~~ f_{20} \neq 0,~~f_{02} \neq 0, ~~f_{00} = 0  ~~ \Leftrightarrow ~~ 
	\begin{dcases}
		\alpha - \beta + \gamma - \delta \neq 0 \\
		\alpha - \beta - \gamma + \delta \neq 0 \\
		\alpha + \beta - \gamma - \delta \neq 0 \\
		\alpha + \beta + \gamma + \delta = 2\pi
	\end{dcases}
\end{equation*}
We could transfer the result for Conic I to Conic IV from switching two strips, continuing from Conic II or III:
\begin{equation}
	(\alpha, ~\beta, ~\gamma, ~\delta, ~x, ~y, ~z, ~w) \rightarrow (\pi-\alpha, ~\beta, ~\pi-\gamma, ~\delta, ~-x^{-1}, ~-y^{-1}, ~z^{-1}, ~w^{-1}) \\
\end{equation}
Now:
\begin{equation*}
	\sigma = \dfrac{-\alpha + \beta - \gamma + \delta}{2} + \pi
\end{equation*}
and we could directly write the result for Conic IV with a little careful checking on the signs of the expressions:
\begin{equation}
	\begin{dcases}
		x^{-1} = -p_x \cos t \\
		y^{-1} = \mathrm{sign}(\sigma - \pi) p_y \cos \left(t - \theta_1 \right) \\
		z^{-1} = p_z \cos \left(t + \dfrac{\pi}{2}\right) \\
		w^{-1} = \mathrm{sign}(\pi - \sigma) p_w \cos \left(t - \theta_2 \right) \\
	\end{dcases}
\end{equation}
\begin{equation*}
	\begin{gathered}
		p_x=\sqrt{\dfrac{\sin \alpha \sin \beta}{\sin \gamma \sin \delta}-1}, ~~ p_y=\sqrt{\dfrac{\sin \beta \sin \gamma}{\sin \delta \sin \alpha}-1} \\
		p_z=\sqrt{\dfrac{\sin \gamma \sin \delta}{\sin \alpha \sin \beta}-1}, ~~ p_w=\sqrt{\dfrac{\sin \delta \sin \alpha}{\sin \beta \sin \gamma}-1}
	\end{gathered}
\end{equation*}

\bgroup
\def\arraystretch{2}
\begin{center}
	\begin{tabular}{|c|c|c|}
		\hline
		\makecell{Choice of $t$ for \\ real configurations} &  Branch 1 & Branch 2 \\ 
		\hline
		\makecell{$p_x \in \mathbb{R}^+$ \\ $p_z \in i\mathbb{R}^+$}  & \makecell {$is$, $s \in (-\infty, 0)$ \\ $\pi + is$, $s \in (-\infty, 0)$ \\ when $xz \neq 0$, $xz>0$ \\ this branch is continuous at $x = 0$ \\ and `snaps' at $x = \infty$} & \makecell {$is$, $s \in (0, +\infty) $ \\ $\pi + is$, $s \in (0, +\infty)$ \\ when $xz \neq 0$, $xz<0$ \\ this branch is continuous at $x = 0$ \\ and `snaps' at $x = \infty$} \\
		\hline
		\makecell{$p_x \in i\mathbb{R}^+$ \\ $p_z \in \mathbb{R}^+$} & \makecell {$\pi/2 + is$, $s \in (-\infty, 0)$ \\ $3\pi/2 + is$, $s \in (-\infty, 0)$ \\ when $xz \neq 0$, $xz>0$ \\ this branch is continuous at $x = 0$ \\ and `snaps' at $x = \infty$} & \makecell {$3\pi/2 + is$, $s \in (0, +\infty)$ \\ $\pi/2 + is$, $s \in (0, +\infty)$ \\ when $xz \neq 0$, $xz<0$ \\ this branch is continuous at $x = 0$ \\ and `snaps' at $x = \infty$} \\  
		\hline      
	\end{tabular}
\end{center}
\egroup
For the phase shift:
\bgroup
\def\arraystretch{2}
\begin{center}
	\begin{tabular}{|c|c|c|}
		\hline
		&  $\theta_1$ & $\theta_2$ \\
		\hline
		$p_x \in \mathbb{R}^+, ~p_y \in \mathbb{R}^+$ & $ \ln \left|\dfrac{\sqrt{\sin \alpha \sin \gamma}+\sqrt{\sin \beta \sin \delta}}{\sqrt{\sin \alpha \sin \gamma}-\sqrt{\sin \beta \sin \delta}}\right|$ & $\dfrac{\pi}{2} + \theta_1$ \\
		\hline
		$p_y \in \mathbb{R}^+, ~p_z \in \mathbb{R}^+$ & $-\dfrac{\pi}{2} + \ln \left|\dfrac{\sqrt{\sin \alpha \sin \gamma}+\sqrt{\sin \beta \sin \delta}}{\sqrt{\sin \alpha \sin \gamma}-\sqrt{\sin \beta \sin \delta}}\right|$ & $\dfrac{\pi}{2} + \theta_1$\\ 
		\hline 
		$p_z \in \mathbb{R}^+, ~p_w \in \mathbb{R}^+$ & $ \ln \left|\dfrac{\sqrt{\sin \alpha \sin \gamma}+\sqrt{\sin \beta \sin \delta}}{\sqrt{\sin \alpha \sin \gamma}-\sqrt{\sin \beta \sin \delta}}\right|$ & $-\dfrac{\pi}{2} + \theta_1$ \\
		\hline
		$p_w \in \mathbb{R}^+, ~p_x \in \mathbb{R}^+$ & $ \dfrac{\pi}{2} + \ln \left|\dfrac{\sqrt{\sin \alpha \sin \gamma}+\sqrt{\sin \beta \sin \delta}}{\sqrt{\sin \alpha \sin \gamma}-\sqrt{\sin \beta \sin \delta}}\right|$ & $-\dfrac{\pi}{2} + \theta_1$ \\
		\hline      
	\end{tabular}
\end{center}
\egroup

Please refer to the table below on how to determine the maximum sector angle from the range of amplitudes $p_x, ~p_y, ~p_z, ~p_w$.
\bgroup
\def\arraystretch{2}
\begin{center}
	\begin{tabular}{|c|c|c|}
		\hline
		& $\sigma < \pi$ & $\sigma > \pi$ \\ 
		\hline
		$\beta = \mathrm{max}(\pi -\alpha, ~\beta, ~\pi -\gamma, ~\delta)$ & 
		\makecell{$\sin \alpha \sin \beta > \sin \gamma \sin \delta$ \\
			$\sin \beta \sin \gamma > \sin \delta \sin \alpha$ \\
			$\sin \alpha \sin \gamma > \sin \beta \sin \delta$}
		&  \makecell{$\sin \alpha \sin \beta < \sin \gamma \sin \delta$ \\
			$\sin \beta \sin \gamma < \sin \delta \sin \alpha$ \\
			$\sin \alpha \sin \gamma > \sin \beta \sin \delta$} \\
		\hline
		$\pi - \alpha = \mathrm{max}(\pi -\alpha, ~\beta, ~\pi -\gamma, ~\delta)$ & \makecell{$\sin \alpha \sin \beta > \sin \gamma \sin \delta$ \\
			$\sin \beta \sin \gamma < \sin \delta \sin \alpha$ \\
			$\sin \alpha \sin \gamma < \sin \beta \sin \delta$} &  \makecell{$\sin \alpha \sin \beta < \sin \gamma \sin \delta$ \\
			$\sin \beta \sin \gamma > \sin \delta \sin \alpha$ \\
			$\sin \alpha \sin \gamma < \sin \beta \sin \delta$} \\
		\hline
		$\delta= \mathrm{max}(\pi -\alpha, ~\beta, ~\pi -\gamma, ~\delta)$ & \makecell{$\sin \alpha \sin \beta < \sin \gamma \sin \delta$ \\
			$\sin \beta \sin \gamma < \sin \delta \sin \alpha$ \\
			$\sin \alpha \sin \gamma > \sin \beta \sin \delta$} & \makecell{$\sin \alpha \sin \beta > \sin \gamma \sin \delta$ \\
			$\sin \beta \sin \gamma > \sin \delta \sin \alpha$ \\
			$\sin \alpha \sin \gamma > \sin \beta \sin \delta$} \\
		\hline
		$\pi - \gamma = \mathrm{max}(\pi -\alpha, ~\beta, ~\pi -\gamma, ~\delta)$ & \makecell{$\sin \alpha \sin \beta < \sin \gamma \sin \delta$ \\
			$\sin \beta \sin \gamma > \sin \delta \sin \alpha$ \\
			$\sin \alpha \sin \gamma < \sin \beta \sin \delta$} &  \makecell{$\sin \alpha \sin \beta > \sin \gamma \sin \delta$ \\
			$\sin \beta \sin \gamma < \sin \delta \sin \alpha$ \\
			$\sin \alpha \sin \gamma < \sin \beta \sin \delta$} \\
		\hline
	\end{tabular}
\end{center}
\egroup

\section{Elliptic: finite solution} \label{section: spherical elliptic}
This is the most general case with no extra condition on the sector angles. 
\begin{equation*} 
	f_{22} \neq 0, ~~ f_{20} \neq 0,~~f_{02} \neq 0, ~~f_{00} \neq 0  ~~ \Leftrightarrow ~~ 
	\begin{dcases}
		\alpha - \beta + \gamma - \delta \neq 0 \\
		\alpha - \beta - \gamma + \delta \neq 0 \\
		\alpha + \beta - \gamma - \delta \neq 0 \\
		\alpha + \beta + \gamma + \delta \neq 2\pi
	\end{dcases}
\end{equation*}
Similarly, let us start with Equation \eqref{eq: opposite folding angle}:
\begin{equation*} 
	\begin{gathered}
		g(\alpha, ~\beta, ~\gamma, ~\delta, ~x, ~ z)=g_{22}x^2z^2+g_{20}x^2+g_{02}z^2+g_{00}=0 \\
		g_{22} = \sin (\sigma-\alpha-\delta) \sin(\sigma-\beta-\delta), ~~
		g_{20} = \sin (\sigma-\alpha) \sin (\sigma-\beta) \\
		g_{02} = - \sin (\sigma-\gamma) \sin (\sigma-\delta), ~~
		g_{00} = \sin \sigma \sin (\sigma - \alpha - \beta) \\
		\sigma=\dfrac{\alpha+\beta+\gamma+\delta}{2}
	\end{gathered}
\end{equation*}
Let
\begin{equation}
	M = \dfrac{\sin \alpha \sin \beta \sin \gamma \sin \delta}{\sin (\sigma-\alpha) \sin (\sigma-\beta) \sin (\sigma-\gamma) \sin(\sigma-\delta)} \in (0, 1) \cup (1, +\infty)
\end{equation}
If using the amplitudes $p_x$ and $p_z$, from Section \ref{section: sign convention folding angle} we could see the above equation implies:
\begin{equation} \label{eq: spherical elliptic opposite}
	\begin{gathered}
		g(\alpha, ~\beta, ~\gamma, ~\delta, ~x, ~ z)=\left(M-1\right)\dfrac{x^2z^2}{p_x^2p_z^2}+\dfrac{x^2}{p_x^2}+\dfrac{z^2}{p_z^2} - 1 = 0\\
	\end{gathered}	
\end{equation} 
Further,
\begin{equation*}
	\begin{gathered}
		f(\alpha, ~\beta, ~\gamma, ~\delta, ~x, ~y) = f_{22} x^2y^2 + f_{20}x^2 + 2f_{11}xy + f_{02}y^2 + f_{00} = 0 \\
		f_{22} = \sin (\sigma-\beta) \sin (\sigma-\beta-\delta), ~~
		f_{20} = \sin (\sigma-\alpha) \sin(\sigma-\alpha-\delta) \\
		f_{11} = - \sin \alpha \sin \gamma, ~~ f_{02} =   \sin (\sigma-\gamma) \sin (\sigma-\gamma-\delta), ~~ f_{00}  = 
		\sin \sigma \sin (\sigma-\delta) \\
	\end{gathered}
\end{equation*}
Divide by $\dfrac{\sin \sigma \sin(\sigma-\alpha-\delta) \sin (\sigma-\gamma-\delta)}{\sin (\sigma-\beta)}$:
\begin{equation} \label{eq: spherical elliptic adjacent}
	\begin{aligned}
		f(\alpha, ~\beta, ~\gamma, ~\delta, ~x, ~y) & = \left( \dfrac{\sin \alpha \sin \gamma}{\sin(\sigma-\alpha)\sin(\sigma-\gamma)}-1 \right)\dfrac{x^2y^2}{p_x^2p_y^2} + \dfrac{x^2}{p_x^2} + \dfrac{y^2}{p_y^2} \\ & \quad +  \dfrac{\sin(\sigma-\beta)\sin(\sigma-\delta)}{\sin(\sigma-\alpha)\sin(\sigma-\gamma)-\sin \beta \sin\delta} \\ & \quad - \dfrac{2 \sin \alpha \sin \gamma}{\sqrt{\sin (\sigma-\alpha) \sin (\sigma-\gamma)(\sin (\sigma-\alpha) \sin (\sigma-\gamma)- \sin \beta \sin \delta)}}\dfrac{xy}{p_xp_y} = 0
	\end{aligned} \\
\end{equation} 
hence 
\begin{equation} \label{eq: spherical elliptic adjacent 2}
	\begin{aligned}
		f(\beta, ~\alpha, ~\delta, ~\gamma, ~x, ~w) & = \left( \dfrac{\sin \beta \sin \delta}{\sin (\sigma-\beta) \sin (\sigma-\delta)}-1 \right)\dfrac{x^2y^2}{p_x^2p_w^2} + \dfrac{x^2}{p_x^2} + \dfrac{w^2}{p_w^2} \\ & \quad + \dfrac{\sin(\sigma-\alpha)\sin(\sigma-\gamma)}{\sin(\sigma-\beta)\sin(\sigma-\delta)-\sin\alpha\sin\beta} \\ & \quad - \dfrac{2\sin\beta \sin\delta}{\sqrt{\sin(\sigma-\beta)\sin(\sigma-\delta)(\sin(\sigma-\beta)\sin(\sigma-\delta)-\sin\alpha\sin\gamma)}}\dfrac{xy}{p_xp_w} = 0
	\end{aligned} \\
\end{equation} 

It has the parametrization using elliptic functions defined below. We have provided a brief introduction on elliptic function in \citet[Section 14]{he_real_2023}.

\begin{prop}
	Consider the parametrization of Equation \eqref{eq: spherical elliptic opposite}, Equation \eqref{eq: spherical elliptic adjacent} and Equation \eqref{eq: spherical elliptic adjacent 2}.
	\begin{enumerate} [label={[\arabic*]}]
		\item When $M \in (1, +\infty)$:
		\begin{equation*}
			\begin{gathered}
				\dfrac{k^2}{1-k^2}\mathrm{cn}^2t\mathrm{cn}^2(t+K)+\mathrm{cn}^2t+\mathrm{cn}^2(t+K)-1=0 \\
				k = \sqrt{1-\dfrac{1}{M}}
			\end{gathered}
		\end{equation*}
		That is to say for Equation \eqref{eq: spherical elliptic opposite}:
		\begin{equation*}
			x = p_x \mathrm{cn} t, ~~z = p_z \mathrm{cn} (t+K)
		\end{equation*}
		\item When $M \in (0, 1)$:
		\begin{equation*}
			\begin{gathered}
				-k^2\mathrm{sn}^2t\mathrm{sn}^2(t+K)+\mathrm{sn}^2t+\mathrm{sn}^2(t+K)-1=0 \\
				k = \sqrt{1-M}
			\end{gathered}
		\end{equation*}
		That is to say for Equation \eqref{eq: spherical elliptic opposite}:
		\begin{equation*}
			x = p_x \mathrm{sn} t, ~~z = p_z \mathrm{sn} (t+K)
		\end{equation*}
		\item When $M \in (1, +\infty)$ and $\pi - \sigma > 0$:
		\begin{equation*}
			\begin{gathered}
				\dfrac{k^2\mathrm{sn}^2 \theta_1}{\mathrm{dn}^2 \theta_1}\mathrm{cn}^2 t\mathrm{cn}^2 (t-\theta_1)+\mathrm{cn}^2t-\dfrac{2\mathrm{cn}\theta_1}{\mathrm{dn}^2\theta_1}\mathrm{cn} t\mathrm{cn} (t-\theta_1)+\mathrm{cn}^2(t-\theta_1)+\dfrac{(k^2-1)\mathrm{sn}^2 \theta_1}{\mathrm{dn}^2 \theta_1}=0 \\
				k = \sqrt{1-\dfrac{1}{M}}
			\end{gathered}
		\end{equation*}
		Hence for the parametrization of Equation \eqref{eq: spherical elliptic adjacent} when $\alpha+\gamma<\sigma$:
		\begin{equation*}
			\begin{dcases}
				\dfrac{k^2\mathrm{sn}^2 \theta_1}{\mathrm{dn}^2 \theta_1} = \dfrac{\sin \alpha \sin \gamma}{\sin(\sigma-\alpha)\sin(\sigma-\gamma)}-1 \\
				\dfrac{\mathrm{cn}\theta_1}{\mathrm{dn}^2\theta_1} = \dfrac{\sin\alpha \sin\gamma}{\sqrt{\sin(\sigma-\alpha)\sin(\sigma-\gamma)(\sin(\sigma-\alpha)\sin(\sigma-\gamma)-\sin\beta\sin\delta)}}\\
				\dfrac{(k^2-1)\mathrm{sn}^2 \theta_1}{\mathrm{dn}^2 \theta_1} = \dfrac{\sin(\sigma-\beta)\sin(\sigma-\delta)}{\sin(\sigma-\alpha)\sin(\sigma-\gamma)-\sin\beta\sin\delta}
			\end{dcases}
		\end{equation*}
		The above equation determines $\theta_1$ in the form of:
		\begin{equation*}
			\mathrm{dn}\theta_1 = \sqrt{\dfrac{\sin(\sigma-\alpha)\sin(\sigma-\gamma)}{\sin\alpha \sin\gamma}}, ~~ \theta_1 \in (0, iK') \mathrm{~~when~~} \alpha + \gamma < \sigma \\
		\end{equation*}
		When $\alpha + \gamma > \sigma$, we need to do the following transformation:
		\begin{equation*}
			(\alpha, ~\beta, ~\gamma, ~\delta, ~x, ~y, ~z, ~w) \rightarrow (\delta, ~\alpha, ~\beta, ~\gamma, ~w, ~x, ~y, ~z)
		\end{equation*}
		and we could find 
		\begin{equation*}
			\mathrm{dn}\theta_1 = \sqrt{\dfrac{\sin(\sigma-\beta)\sin(\sigma-\delta)}{\sin\beta \sin\delta}}, ~~ \theta_1 \in (0, iK') \mathrm{~~when~~} \alpha + \gamma > \sigma \\
		\end{equation*}
		Note that when $\pi- \sigma < 0$ there is a sign change, which is shown in the final expression.
		\item When $M \in (0, 1)$ and $\pi - \sigma > 0$:
		\begin{equation*}
			\begin{gathered}
				- k^2 \mathrm{sn}^2 \theta_1 \mathrm{sn}^2 t\mathrm{sn}^2 (t-\theta_1)+\mathrm{sn}^2t-2\mathrm{cn}\theta_1 \mathrm{dn} \theta_1 \mathrm{sn} t\mathrm{sn} (t-\theta_1)+\mathrm{sn}^2(t-\theta_1)-\mathrm{sn}^2 \theta_1=0 \\
				k = \sqrt{1-M}
			\end{gathered}
		\end{equation*}
		Hence for the parametrization of Equation \eqref{eq: spherical elliptic adjacent}, when $\alpha + \gamma > \sigma $:
		\begin{equation*}
			\begin{dcases}
				-k^2\mathrm{sn}^2 \theta_1 = \dfrac{\sin\alpha \sin\gamma}{\sin(\sigma-\alpha)\sin(\sigma-\gamma)}-1 \\
				\mathrm{cn}\theta_1 \mathrm{dn} \theta_1 = \dfrac{\sin\alpha \sin\gamma}{\sqrt{\sin(\sigma-\alpha)\sin(\sigma-\gamma)(\sin(\sigma-\alpha)\sin(\sigma-\gamma)-\sin\beta\sin\delta)}} \\
				- \mathrm{sn}^2 \theta_1 = \dfrac{\sin(\sigma-\beta)\sin(\sigma-\delta)}{\sin(\sigma-\alpha)\sin(\sigma-\gamma)-\sin\beta\sin\delta} 
			\end{dcases}
		\end{equation*}
		The above equation determines $\theta_1$ in the form of:
		\begin{equation*}
			\mathrm{dn}\theta_1 = \sqrt{\dfrac{\sin \alpha \sin\gamma}{\sin(\sigma-\alpha)\sin(\sigma-\gamma)}}, ~~\theta_1 \in (0, iK') \mathrm{~~when~~} \alpha + \gamma > \sigma \\
		\end{equation*}
		When $\alpha + \gamma < \sigma$, we need to do the following transformation:
		\begin{equation*}
			(\alpha, ~\beta, ~\gamma, ~\delta, ~x, ~y, ~z, ~w) \rightarrow (\delta, ~\alpha, ~\beta, ~\gamma, ~w, ~x, ~y, ~z)
		\end{equation*}
		and we could find 
		\begin{equation*}
			\mathrm{dn}\theta_1 = \sqrt{\dfrac{\sin(\sigma-\beta)\sin(\sigma-\delta)}{\sin\beta \sin\delta}}, ~~ \theta_1 \in (0, iK') \mathrm{~~when~~} \alpha + \gamma < \sigma \\
		\end{equation*}
		Note that when $\pi- \sigma < 0$ there is a sign change, which is shown in the final expression.
	\end{enumerate}
\end{prop} 
The proof is also provided in \citet[Section 14]{he_real_2023}.

Subsequently, as in the derivation process presented in Conic I, it remains essential to determine the magnitudes of $\alpha, ~\beta, ~\gamma, ~\delta$.
\begin{prop}
	The relation between the magnitudes of $\alpha, ~ \beta, ~\gamma, ~\delta$ and the signs of amplitudes $p_x, ~p_y, ~p_z, ~p_w$. Note that $\mathrm{max} = \mathrm{max}(\alpha, ~ \beta, ~\gamma, ~\delta)$. $\mathrm{min} = \mathrm{min}(\alpha, ~ \beta, ~\gamma, ~\delta)$. 
	\begin{equation*}
		\begin{gathered}
			\mathrm{max} + \mathrm{min} > \sigma \Leftrightarrow \mathrm{sign}(\pi-\sigma)(M-1)>0 \\
			\mathrm{max} + \mathrm{min} < \sigma \Leftrightarrow \mathrm{sign}(\pi-\sigma)(M-1)<0
		\end{gathered}
	\end{equation*}
	\begin{enumerate} [label={[\arabic*]}]
		\item $\mathrm{sign}(\pi-\sigma)(M-1)>0$ and $\alpha + \gamma < \sigma$ imply $\beta = \max$ or $\delta = \max$:
		\begin{enumerate} [label={[\alph*]}]
			\item $\beta = \max \Leftrightarrow \alpha \beta > (\sigma-\alpha)(\sigma-\beta)$, $\beta \gamma > (\sigma-\beta)(\sigma-\gamma)$
			\item $\delta = \max \Leftrightarrow \alpha \beta < (\sigma-\alpha)(\sigma-\beta)$, $\beta \gamma < (\sigma-\beta)(\sigma-\gamma)$
		\end{enumerate}
		\item $\mathrm{sign}(\pi-\sigma)(M-1)>0$ and $\alpha + \gamma > \sigma$ imply $\alpha = \max$ or $\gamma = \max$:
		\begin{enumerate} [label={[\alph*]}]
			\item $\alpha = \max \Leftrightarrow \alpha \beta > (\sigma-\alpha)(\sigma-\beta)$, $\beta \gamma < (\sigma-\beta)(\sigma-\gamma)$
			\item $\gamma = \max \Leftrightarrow \alpha \beta < (\sigma-\alpha)(\sigma-\beta)$, $\beta \gamma > (\sigma-\beta)(\sigma-\gamma)$
		\end{enumerate}
		\item $\mathrm{sign}(\pi-\sigma)(M-1)<0$ and $\alpha + \gamma > \sigma$ imply $\beta = \min$ or $\delta = \min$:
		\begin{enumerate} [label={[\alph*]}]
			\item $\beta = \min \Leftrightarrow \alpha \beta < (\sigma-\alpha)(\sigma-\beta)$, $\beta \gamma < (\sigma-\beta)(\sigma-\gamma)$
			\item $\delta = \min \Leftrightarrow \alpha \beta > (\sigma-\alpha)(\sigma-\beta)$, $\beta \gamma > (\sigma-\beta)(\sigma-\gamma)$
		\end{enumerate}
		\item $\mathrm{sign}(\pi-\sigma)(M-1)<0$ and $\alpha + \gamma < \sigma$ imply $\alpha = \min$ or $\gamma = \min$:
		\begin{enumerate} [label={[\alph*]}]
			\item $\alpha = \min \Leftrightarrow \alpha \beta < (\sigma-\alpha)(\sigma-\beta)$, $\beta \gamma > (\sigma-\beta)(\sigma-\gamma)$
			\item $\gamma = \min \Leftrightarrow \alpha \beta > (\sigma-\alpha)(\sigma-\beta)$, $\beta \gamma < (\sigma-\beta)(\sigma-\gamma)$
		\end{enumerate}
	\end{enumerate}
\end{prop}

We are now prepared to present the solutions for both the real configuration space and its complexified counterpart.

If $M > 1$:
\begin{equation*}
	\begin{gathered}
		\begin{dcases}
			x = p_x \mathrm{cn}(t; ~k) \\
			y = \mathrm{sign}(\pi-\sigma)p_y \mathrm{cn}(t - \theta_1; ~k) \\
			z = p_z \mathrm{cn}(t+K; ~k) \\
			w = \mathrm{sign}(\pi-\sigma)p_w \mathrm{cn}(t - \theta_2; ~k) \\ 
		\end{dcases}, ~~ \begin{dcases}
			k = \sqrt{1-\dfrac{1}{M}}, ~~k' = \sqrt{\dfrac{1}{M}} \\
			K = \bigintsss_{0}^{\frac{\pi}{2}} \dfrac{\dif \psi}{\sqrt{1 - k^2 \sin^2 \psi}}
		\end{dcases} 
	\end{gathered}
\end{equation*}
To calculate $s$ with $x$:
\begin{equation*}
	\begin{gathered}
		\mathrm{cn}(is; ~k) = \dfrac{1}{\mathrm{cn}(s; ~k')}, ~~
		\mathrm{cn}(K + is; ~k) = -ik'\dfrac{\mathrm{sn}(s; ~k')}{\mathrm{dn}(s; ~k')} \\
	\end{gathered}
\end{equation*} 
\bgroup
\def\arraystretch{2}
\begin{center}
	\begin{tabular}{|c|c|c|}
		\hline
		\makecell{Choice of $t$ for \\ real configurations} &  Branch 1 & Branch 2 \\ 
		\hline
		\makecell{$p_x \in \mathbb{R}^+$ \\ $p_z \in i\mathbb{R}^+$} & \makecell {$is$, $s \in (0, +\infty)$ \\ $2K + is$, $s \in (0, +\infty)$ \\ when $xz \neq 0$, $xz>0$ \\ this branch `snaps' at $x = \infty$ \\ and `snaps' at $x = 0$} & \makecell {$is$, $s \in (-\infty, 0)$ \\ $2K + is$, $s \in (-\infty, 0)$ \\ when $xz \neq 0$, $xz<0$ \\ this branch `snaps' at $x = \infty$ \\ and `snaps' at $x = 0$} \\
		\hline
		\makecell{$p_x \in i\mathbb{R}^+$ \\ $p_z \in \mathbb{R}^+$} & \makecell {$3K + is$, $s \in (0, +\infty)$ \\ $K + is$, $s \in (0, +\infty)$ \\ when $xz \neq 0$, $xz>0$ \\ this branch `snaps' at $x = \infty$ \\ and `snaps' at $x = 0$} & \makecell {$K + is$, $s \in (-\infty, 0)$ \\ $3K + is$, $s \in (-\infty, 0)$ \\ when $xz \neq 0$, $xz<0$ \\ this branch `snaps' at $x = \infty$ \\ and `snaps' at $x = 0$} \\  
		\hline      
	\end{tabular}
\end{center}
for phase shift:
\begin{center}
	\begin{tabular}{|c|c|c|}
		\hline
		&  $\theta_1$ & $\theta_2$ \\
		\hline
		$p_x \in \mathbb{R}^+, ~p_y \in \mathbb{R}^+$ & $i\mathrm{dc}^{-1}\left(\sqrt{\dfrac{\sin(\sigma - \alpha)\sin(\sigma - \gamma)}{\sin \alpha \sin \gamma}};~k'\right)$ & $K + \theta_1$ \\
		\hline
		$p_y \in \mathbb{R}^+, ~p_z \in \mathbb{R}^+$ & $-K + i\mathrm{dc}^{-1}\left(\sqrt{\dfrac{\sin(\sigma - \beta)\sin(\sigma - \delta)}{\sin \beta \sin \delta}};~k'\right)$ & $K + \theta_1$\\ 
		\hline 
		$p_z \in \mathbb{R}^+, ~p_w \in \mathbb{R}^+$ & $ i\mathrm{dc}^{-1}\left(\sqrt{\dfrac{\sin(\sigma - \alpha)\sin(\sigma - \gamma)}{\sin \alpha \sin \gamma}};~k'\right)$ & $-K + \theta_1$ \\
		\hline
		$p_w \in \mathbb{R}^+, ~p_x \in \mathbb{R}^+$ & $ K + i\mathrm{dc}^{-1}\left(\sqrt{\dfrac{\sin(\sigma - \beta)\sin (\sigma - \delta)}{\sin \beta \sin \delta}};~k'\right)$ & $-K + \theta_1$ \\
		\hline      
	\end{tabular}
\end{center}
\begin{equation*}
	\begin{aligned}
		\mathrm{dn} \theta_1 = \dfrac{\mathrm{dn}(\theta_1';~k')}{\mathrm{cn}(\theta_1';~k')} = \mathrm{dc}(\theta_1';~k') \mathrm{~~when~~} \theta_1 \in (0, iK') 
	\end{aligned}
\end{equation*}
\egroup

If $M < 1$:
\begin{equation*}
	\begin{gathered}
		\begin{dcases}
			x = p_x \mathrm{sn}(t; ~k) \\
			y = \mathrm{sign}(\pi - \sigma)p_y \mathrm{sn}(t - \theta_1; ~k) \\
			z = p_z \mathrm{sn}(t+K; ~k) \\
			w = \mathrm{sign}(\pi - \sigma)p_w \mathrm{sn}(t - \theta_2; ~k) \\ 
		\end{dcases}, ~~ \begin{dcases}
			k = \sqrt{1-M}, ~~ k' = \sqrt{M}\\
			K = \bigintsss_{0}^{\frac{\pi}{2}} \dfrac{\dif \psi}{\sqrt{1 - k^2 \sin^2 \psi}}
		\end{dcases}
	\end{gathered}
\end{equation*}
To calculate $s$ with $x$:
\begin{equation*}
	\begin{gathered}
		\mathrm{sn}(K + is; ~k) = \dfrac{1}{\mathrm{dn}(s; ~k')}, ~~ \mathrm{sn}(is; ~k) = i \dfrac{\mathrm{sn}(s; ~k')}{\mathrm{cn}(s; ~k')}
	\end{gathered}
\end{equation*} 

\bgroup
\def\arraystretch{2}
\begin{center}
	\begin{tabular}{|c|c|c|}
		\hline
		\makecell{Choice of $t$ for \\ real configurations} &  Branch 1 & Branch 2 \\ 
		\hline
		\makecell{$p_x \in \mathbb{R}^+$ \\ $p_z \in i\mathbb{R}^+$} & \makecell {$K + is$, $s \in (0, +\infty)$ \\ $3K + is$, $s \in (0, +\infty)$ \\ when $xz \neq 0$, $xz>0$ \\ this branch `snaps' at $x = \infty$ \\ and `snaps' at $x = 0$} & \makecell {$K + is$, $s \in (-\infty, 0)$ \\ $3K + is$, $s \in (-\infty, 0)$ \\ when $xz \neq 0$, $xz<0$ \\ this branch `snaps' at $x = \infty$ \\ and `snaps' at $x = 0$} \\
		\hline
		\makecell{$p_x \in i\mathbb{R}^+$ \\ $p_z \in \mathbb{R}^+$} & \makecell {$is$, $s \in (0, +\infty)$ \\ $2K + is$, $s \in (0, +\infty)$ \\ when $xz \neq 0$, $xz>0$ \\ this branch `snaps' at $x = \infty$ \\ and `snaps' at $x = 0$} & \makecell {$2K + is$, $s \in (-\infty, 0)$ \\ $is$, $s \in (-\infty, 0)$ \\ when $xz \neq 0$, $xz<0$ \\ this branch `snaps' at $x = \infty$ \\ and `snaps' at $x = 0$} \\  
		\hline      
	\end{tabular}
\end{center}
for phase shift:
\begin{center}
	\begin{tabular}{|c|c|c|}
		\hline
		&  $\theta_1$ & $\theta_2$ \\
		\hline
		$p_x \in \mathbb{R}^+, ~p_y \in \mathbb{R}^+$ & $i\mathrm{dc}^{-1}\left(\sqrt{\dfrac{\sin \alpha \sin \gamma}{\sin(\sigma - \alpha)\sin(\sigma - \gamma)}};~k'\right)$ & $K + \theta_1$ \\
		\hline
		$p_y \in \mathbb{R}^+, ~p_z \in \mathbb{R}^+$ & $-K + i\mathrm{dc}^{-1}\left(\sqrt{\dfrac{\sin \beta \sin \delta}{\sin(\sigma - \beta)\sin(\sigma - \delta)}};~k'\right)$ & $K + \theta_1$\\ 
		\hline 
		$p_z \in \mathbb{R}^+, ~p_w \in \mathbb{R}^+$ & $ i\mathrm{dc}^{-1}\left(\sqrt{\dfrac{\sin \alpha \sin \gamma}{\sin(\sigma - \alpha)\sin(\sigma - \gamma)}};~k'\right)$ & $-K + \theta_1$ \\
		\hline
		$p_w \in \mathbb{R}^+, ~p_x \in \mathbb{R}^+$ & $ K + i\mathrm{dc}^{-1}\left(\sqrt{\dfrac{\sin \beta \sin \delta}{\sin (\sigma - \beta)\sin (\sigma - \delta)}};~k'\right)$ & $-K + \theta_1$ \\
		\hline      
	\end{tabular}
\end{center}
where
\begin{equation*}
	\begin{aligned}
		\mathrm{dn} \theta_1 = \dfrac{\mathrm{dn}(\theta_1';~k')}{\mathrm{cn}(\theta_1';~k')} = \mathrm{dc}(\theta_1';~k') \mathrm{~~when~~} \theta_1 \in (0, iK') 
	\end{aligned}
\end{equation*}
\egroup

The findings in Elliptic I and Elliptic II exhibit formal symmetry and align well with Deltoid II, Conic I, Conic II, and Conic III. 

\section{Orthodiagonal: finite solution}

The condition on sector angles is:
\begin{equation*} 
	f_{22} \neq 0, ~~ f_{20} \neq 0,~~f_{02} \neq 0, ~~f_{00} \neq 0  ~~ \Leftrightarrow ~~ 
	\begin{dcases}
		\alpha - \beta + \gamma - \delta \neq 0 \\
		\alpha - \beta - \gamma + \delta \neq 0 \\
		\alpha + \beta - \gamma - \delta \neq 0 \\
		\alpha + \beta + \gamma + \delta \neq 2\pi
	\end{dcases}
\end{equation*}
and
\begin{equation*}
	\cos \alpha \cos \gamma = \cos \beta \cos \delta \Leftrightarrow \sin(\sigma - \alpha)\sin(\sigma - \gamma) = \sin(\sigma - \beta)\sin(\sigma - \delta)
\end{equation*}
From the helpful identities listed in Proposition \ref{prop: identities folding angle}, we have the following simplified Equations \eqref{eq: degree-4 vertex adjacent}:
\begin{equation*}
	\begin{gathered}
		f(\alpha, ~\beta, ~\gamma, ~\delta, ~x, ~y) = f_{22} x^2y^2 + f_{20}x^2 + 2f_{11}xy + f_{02}y^2 + f_{00} = 0 \\
		f_{22} = \sin(\sigma-\beta)\sin(\sigma-\beta-\delta) = \dfrac{\sin (\beta - \alpha) \sin (\beta - \gamma)}{2\cos \beta} \\
		f_{20} = \sin(\sigma-\alpha)\sin(\sigma-\alpha-\delta) = \dfrac{\sin (\beta - \alpha) \sin (\beta + \gamma)}{2\cos \beta}\\
		f_{11} = - \sin \alpha \sin\gamma \\
		f_{02} = \sin(\sigma-\gamma)\sin(\sigma-\gamma-\delta) = \dfrac{\sin (\beta + \alpha) \sin (\beta - \gamma)}{2\cos \beta} \\
		f_{00}  = 
		\sin\sigma\sin(\sigma-\delta) = \dfrac{\sin (\beta + \alpha) \sin (\beta + \gamma)}{2\cos \beta} \\
		\sigma=\dfrac{\alpha+\beta+\gamma+\delta}{2}
	\end{gathered}
\end{equation*}
That is to say $x$ and $y$ are separable:
\begin{equation*}
	f = 0 \Leftrightarrow \left(\sin(\beta - \alpha)x + \dfrac{\sin(\beta + \alpha)}{x}  \right) \left(\sin(\beta - \gamma)y + \dfrac{\sin(\beta + \gamma)}{y}  \right) =  4 \sin \alpha \sin \gamma \cos \beta
\end{equation*}
In the parametrized expression, we could obtain $\cos \alpha \cos \gamma = \cos \beta \cos \delta \Rightarrow M<1$, since
\begin{equation*}
	\begin{aligned}
		1-M = & \dfrac{-\sin\sigma\sin(\sigma-\alpha-\beta)\sin(\sigma-\alpha-\gamma)\sin(\sigma-\beta-\gamma)}{\sin(\sigma-\alpha)\sin(\sigma-\beta)\sin(\sigma-\gamma)\sin(\sigma-\delta)} \\
		& = \dfrac{(\cos (\alpha - \beta)-\cos(\alpha - \beta))(\cos (\alpha + \beta)-\cos(\alpha + \beta))}{4} \\
		& = \left(\dfrac{\sin\alpha\sin\gamma - \sin\beta\sin\delta}{\sin\alpha\sin\gamma + \sin\beta\sin\delta}\right)^2
	\end{aligned}
\end{equation*}
and
\begin{equation*}
	k = \sqrt{1-M} = \dfrac{|\sin\alpha\sin\gamma - \sin\beta\sin\delta|}{\sin\alpha\sin\gamma + \sin\beta\sin\delta} 
\end{equation*}
Further we claim that $\theta_1'$ is half of the quarter period $K'$ of elliptic functions with modulus $k'$:
\begin{equation*}
	\theta_1' = \dfrac{K'}{2}
\end{equation*}
This is because for half-quarter-periods
\begin{equation*}
	\mathrm{dc}\left(\dfrac{K'}{2}; ~k'\right) = \sqrt{1 + k}
\end{equation*}
If $\sin\alpha\sin\gamma > \sin\beta\sin\delta$, 
\begin{equation*}
	\begin{aligned}
		\mathrm{dc}\left(\dfrac{K'}{2}; ~k'\right) = \sqrt{\dfrac{2 \sin \alpha \sin \gamma}{\sin\alpha\sin\gamma + \sin\beta\sin\delta}} = \sqrt{\dfrac{\sin\alpha \sin\gamma}{\sin(\sigma - \alpha)\sin(\sigma - \gamma)}}
	\end{aligned}
\end{equation*}
If $\sin\alpha\sin\gamma < \sin\beta\sin\delta$, 
\begin{equation*}
	\begin{aligned}
		\mathrm{dc}\left(\dfrac{K'}{2}; ~k'\right) = \sqrt{\dfrac{2 \sin \beta \sin \delta}{\sin\alpha\sin\gamma + \sin\beta\sin\delta}} = \sqrt{\dfrac{\sin\beta \sin\delta}{\sin(\sigma - \beta)\sin(\sigma - \delta)}}
	\end{aligned}
\end{equation*}

\section{Square: solution at infinity}

The condition on sector angles is:
\begin{equation*}
	f_{22}=0, ~~ f_{20}=0,~~f_{02}=0, ~~f_{00} = 0 ~~ \Leftrightarrow ~~  \alpha=\beta=\gamma=\delta=\pi/2
\end{equation*}
which means:
\begin{equation*}
	\begin{aligned}
		\begin{cases}
			x_1y_1x_2y_2=0 \\
			x_1w_1x_2w_2=0 \\
			x_1^2z_2^2-x_2^2z_1^2=0 \\
			y_1^2w_2^2-y_2^2w_1^2=0
		\end{cases} 
		~~ \Rightarrow ~~ & \begin{cases}
			x_1 \neq 0, ~~ x_2=0 \\
			z_1 \neq 0, ~~ z_2=0 \\
			w = \pm y
		\end{cases} \mathrm{or} \quad \begin{cases}
			y_1 \neq 0, ~~ y_2=0 \\
			w_1 \neq 0, ~~ w_2=0 \\
			x = \pm z
		\end{cases} \\
		~~ \Rightarrow ~~ & \begin{cases}
			x = \infty \\
			z = \infty \\
			w = \pm y
		\end{cases}  \mathrm{or} \quad \begin{cases}
			y = \infty \\
			w = \infty \\
			x =\pm z
		\end{cases}
	\end{aligned}
\end{equation*}
Following the examination, the solutions at infinity correspond to two distinct branches of simple folds, each of which is diffeomorphic to a circle $S^1$.
\begin{equation*}
	\begin{cases}
		x \in \mathbb{R} \cup \{\infty\} \\
		y \equiv \infty \\
		z = -x \\ 
		w \equiv \infty \\
	\end{cases} \mathrm{or} \quad \begin{cases}
		x \equiv \infty \\
		y \in \mathbb{R} \cup \{\infty\} \\
		z \equiv \infty \\ 
		w = -y \\
	\end{cases}	 
\end{equation*}

\section{Rhombus: solution at infinity}

The condition on sector angles is 
\begin{equation*}
	f_{22}=0, ~~ f_{20}=0,~~f_{02}=0, ~~f_{00} \neq 0 ~~ \Leftrightarrow ~~  \alpha=\beta=\gamma=\delta \neq \pi/2
\end{equation*}
which means:
\begin{equation*}
	\begin{aligned}
		\begin{cases}
			(x_1y_1-x_2y_2 \cos \alpha)x_2y_2=0 \\
			(x_1w_1-x_2w_2 \cos \alpha)x_2w_2=0 \\
			x_1^2z_2^2-x_2^2z_1^2=0 \\
			y_1^2w_2^2-y_2^2w_1^2=0
		\end{cases} 
		~~ \Rightarrow ~~ & \begin{cases}
			x_1 \neq 0, ~~ x_2=0 \\
			z_1 \neq 0, ~~ z_2=0 \\
			w = \pm y
		\end{cases} \mathrm{or} \quad \begin{cases}
			y_1 \neq 0, ~~ y_2=0 \\
			w_1 \neq 0, ~~ w_2=0 \\
			x = \pm z
		\end{cases} \\
		~~ \Rightarrow ~~ & \begin{cases}
			x = \infty \\
			z = \infty \\
			w = \pm y
		\end{cases}  \mathrm{or} \quad \begin{cases}
			y = \infty \\
			w = \infty \\
			x =\pm z
		\end{cases}
	\end{aligned}
\end{equation*}
Following the examination, the solutions at infinity correspond to two distinct branches of simple folds, each of which is diffeomorphic to a circle $S^1$.
\begin{equation*}
	\begin{cases}
		x \in \mathbb{R} \cup \{\infty\} \\
		y \equiv \infty \\
		z = -x \\ 
		w \equiv \infty \\
	\end{cases} \mathrm{or} \quad \begin{cases}
		x \equiv \infty \\
		y \in \mathbb{R} \cup \{\infty\} \\
		z \equiv \infty \\ 
		w = -y \\
	\end{cases}	 
\end{equation*}

\section{Cross: solution at infinity}
The condition on sector angles is:
\begin{equation*}
	f_{22} \neq 0, ~~ f_{20}=0,~~f_{02}=0, ~~f_{00} = 0 ~~ \Leftrightarrow ~~  \alpha= \pi - \beta = \gamma = \pi - \delta \neq \pi/2
\end{equation*}
which means:
\begin{equation*}
	\begin{aligned}
		\begin{cases}
			(x_1y_1 \cos \alpha+x_2y_2)x_1y_1=0 \\
			(x_1w_1 \cos \alpha+x_2w_2)x_1w_1=0 \\
			x_1^2z_2^2-x_2^2z_1^2=0 \\
			y_1^2w_2^2-y_2^2w_1^2=0
		\end{cases} 
		~~ \Rightarrow ~~ & \begin{cases}
			x_1 \neq 0, ~~ x_2=0 \\
			z_1 \neq 0, ~~ z_2=0 \\
			w = \pm y = 0
		\end{cases} \mathrm{or} \quad \begin{cases}
			y_1 \neq 0, ~~ y_2=0 \\
			w_1 \neq 0, ~~ w_2=0 \\
			x = \pm z = 0
		\end{cases} \\
		~~ \Rightarrow ~~ & \begin{cases}
			x = \infty \\
			y = 0 \\
			z = \infty \\
			w = 0
		\end{cases}  \mathrm{or} \quad \begin{cases}
			x = 0 \\
			y = \infty \\
			z = 0 \\
			w = \infty \\
		\end{cases}
	\end{aligned}
\end{equation*}
The solutions at infinity are two isolated points, which are limit points of the finite branches.

\section{Miura I: solution at infinity}

The condition on sector angles is:
\begin{equation*}
	f_{22} = 0, ~~ f_{20} \neq 0,~~f_{02}=0, ~~f_{00} = 0 ~~ \Leftrightarrow ~~  \alpha= \pi - \beta = \pi - \gamma = \delta \neq \pi/2
\end{equation*}
which means:
\begin{equation*}
	\begin{dcases}
		x_1y_2 (x_1y_2 \cos \alpha - x_2y_1) = 0 \\
		x_1^2z_2^2 - x_2^2z_1^2 = 0 \\
		x_1w_2 (x_1w_2 \cos \alpha - x_2w_1) = 0
	\end{dcases}
	~~ \Rightarrow ~~ \begin{cases}
		y_1 \neq 0, ~~ y_2 = 0 \\
		z = \pm x	\\
		w_1 \neq 0, ~~ w_2 = 0 \\
	\end{cases}
	\Rightarrow ~~  \begin{cases}
		y = \infty \\
		z = \pm x \\
		w = \infty \\	
	\end{cases}
\end{equation*}
From post-examination, the solutions at infinity becomes another branch diffeomorphic to a circle $S^1$:
\begin{equation*}
	\begin{cases}
		x \in \mathbb{R} \cup \{\infty\} \\
		y \equiv \infty \\
		z = -x \\ 
		w \equiv \infty \\
	\end{cases}
\end{equation*}

\section{Miura II: solution at infinity}
The condition on sector angles is:
\begin{equation*}
	f_{22} = 0, ~~ f_{20} = 0,~~f_{02} \neq 0, ~~f_{00} = 0 ~~ \Leftrightarrow ~~  \alpha = \beta = \pi - \gamma = \pi - \delta \neq \pi/2
\end{equation*}
which means:
\begin{equation*}
	\begin{dcases}
		x_2y_1 (x_2y_1 \cos \alpha + x_1y_2) = 0 \\
		x_1^2z_2^2- x_2^2z_1^2 = 0 \\
		x_2w_1 (x_2w_1 \cos \alpha + x_1w_2) = 0 \\
	\end{dcases}
	~~ \Rightarrow ~~ \begin{cases}
		x_1 \neq 0, ~~ x_2 = 0 \\
		y = \pm w \\
		z_1 \neq 0, ~~ z_2 = 0 \\	
	\end{cases} 
	~~ \Rightarrow ~~ \begin{cases}
		x = \infty \\
		y = \pm w  \\
		z = \infty \\	
	\end{cases}
\end{equation*}
From post-examination, the solutions at infinity becomes another branch diffeomorphic to a circle $S^1$:
\begin{equation*}
	\begin{cases}
		x \equiv \infty \\
		y \in \mathbb{R} \cup \{\infty\} \\
		z \equiv \infty \\ 
		w = -y \\
	\end{cases}	 
\end{equation*}

\section{Isogram: solution at infinity}

The condition on sector angles is 	
\begin{equation*}
	f_{22} \neq 0, ~~ f_{20}=0,~~f_{02}=0 ~~ \Leftrightarrow ~~ \gamma=\alpha, ~~ \delta=\beta, ~~\beta \neq \alpha
\end{equation*}
which means:
\begin{equation*}
	\begin{aligned}
		& \begin{cases}
			\sin(\alpha-\beta)x_1^2y_1^2-2 x_1y_1x_2y_2\sin\alpha+\sin(\alpha+\beta)x_2^2y_2^2 =0 \\
			x_1^2z_2^2-x_2^2z_1^2=0 \\
			\sin(\beta-\alpha)x_1^2w_1^2-2 x_1w_1x_2w_2\sin\beta+\sin(\beta+\alpha)x_2^2w_2^2 =0 \\
		\end{cases} \\
		~~ \Rightarrow ~~ & \begin{cases}
			x_1 \neq 0, ~~ x_2=0 \\
			y_1 = 0, ~~ y_2 \neq 0 \\
			z_1 \neq 0, ~~ z_2=0 \\
			w_1 = 0, ~~ w_2 \neq 0 \\
		\end{cases} \mathrm{or} \quad \begin{cases}
			y_1 \neq 0, ~~ y_2=0 \\
			z_1 = 0, ~~ z_2 \neq 0 \\
			w_1 \neq 0, ~~ w_2=0 \\
			x_1 = 0, ~~ x_2 \neq 0 \\
		\end{cases} \\
		~~ \Rightarrow ~~ & \begin{cases}
			x = \infty \\
			y = 0 \\
			z = \infty \\
			w = 0
		\end{cases}  \mathrm{or} \quad \begin{cases}
			x = 0 \\
			y = \infty \\
			z = 0 \\
			w = \infty \\
		\end{cases}
	\end{aligned}
\end{equation*}
The solutions at infinity are two isolated points, which are limit points of the finite branches.

\section{Anti-isogram: solution at infinity}

The condition on sector angles is:
\begin{equation*}
	f_{22} = 0, ~~ f_{20} \neq 0,~~f_{02} \neq 0, ~~f_{00} = 0 ~~ \Leftrightarrow ~~ \gamma=\pi - \alpha, ~~ \delta= \pi - \beta, ~~\beta \neq \alpha, ~~\alpha + \beta \neq \pi
\end{equation*}
which means:
\begin{equation*}
	\begin{aligned}
		& \begin{cases}
			\sin(\alpha-\beta) x_1^2y_2^2 + 2 x_1y_1x_2y_2 \sin \alpha + \sin(\alpha+\beta) x_2^2y_1^2 = 0 \\
			x_1^2z_2^2-x_2^2z_1^2=0 \\
			\sin(\alpha-\beta) x_1^2w_2^2 + 2 x_1w_1x_2w_2 \sin \alpha + \sin(\alpha+\beta) x_2^2w_1^2 = 0 
		\end{cases}  \\
		~~ \Rightarrow ~~ & \begin{cases}
			x_1 \neq 0, ~~ x_2=0 \\
			y_1 \neq 0, ~~ y_2=0 \\
			z_1 \neq 0, ~~ z_2=0 \\
			w_1 \neq 0, ~~ w_2=0 \\
		\end{cases} 
		~~ \Rightarrow ~~  \begin{cases}
			x = \infty \\
			y = \infty \\
			z = \infty \\
			w = \infty
		\end{cases} 
	\end{aligned}
\end{equation*}
This is exactly the limit point of the two branches in the finite solution, also called a \textbf{flat-folded state}.

\section{Deltoid I: solution at infinity}

The condition on sector angles is:	
\begin{equation*}
	f_{22} = 0, ~~ f_{20} \neq 0,~~f_{02} = 0, ~~f_{00} \neq 0 ~~ \Leftrightarrow ~~ \delta=\alpha, ~~ \gamma=\beta, ~~\beta \neq \alpha, ~~\alpha + \beta \neq \pi
\end{equation*}
which means:
\begin{equation*}
	\begin{aligned}
		& \begin{cases}
			\sin (\beta-\alpha)x_1^2y_2^2-2x_1y_1x_2y_2\sin \alpha+\sin (\beta+\alpha)x_2^2y_2^2 =0 \\
			x_1^2z_2^2-x_2^2z_1^2=0 \\
			\sin (\alpha-\beta)x_1^2w_2^2-2 x_1w_1x_2w_2 \sin \beta+\sin (\alpha+\beta)x_2^2w_2^2 =0 \\
		\end{cases} \\
		~~ \Rightarrow ~~ & \begin{cases}
			y_1 \neq 0, ~~ y_2 = 0 \\
			z = \pm x	\\
			w_1 \neq 0, ~~ w_2 = 0 \\
		\end{cases}
		\Rightarrow ~~  \begin{cases}
			y = \infty \\
			z = \pm x \\
			w = \infty \\	
		\end{cases}
	\end{aligned}
\end{equation*}
From post-examination, the solutions at infinity becomes another branch diffeomorphic to a circle $S^1$:
\begin{equation*}
	\begin{cases}
		x \in \mathbb{R} \cup \{\infty\} \\
		y \equiv \infty \\
		z = -x \\ 
		w \equiv \infty \\
	\end{cases}
\end{equation*}

\section{Anti-deltoid I: solution at infinity}

The condition on sector angles is:
\begin{equation*}
	f_{22} \neq 0, ~~ f_{20} = 0,~~f_{02} \neq 0, ~~f_{00} = 0 ~~ \Leftrightarrow ~~ \delta= \pi - \alpha, ~~ \gamma = \pi - \beta, ~~\beta \neq \alpha, ~~\alpha + \beta \neq \pi
\end{equation*}
which means:
\begin{equation*}
	\begin{aligned}
		& \begin{cases}
			\sin(\beta-\alpha) x_1^2y_1^2 + 2 \sin \alpha x_1y_1x_2y_2+ \sin (\beta+\alpha)x_2^2y_1^2 =0 \\
			x_1^2z_2^2-x_2^2z_1^2=0 \\
			\sin(\alpha-\beta) x_1^2w_1^2 + 2 \sin \beta x_1w_1x_2w_2+ \sin (\alpha+\beta)x_2^2w_1^2 =0 \\  
		\end{cases}  \\
		\Rightarrow ~~  & \begin{dcases}
			x = \infty \\
			y = 0 \\
			z = \infty \\
			w = 0 
		\end{dcases} ~~ \mathrm{or} ~~ \begin{dcases}
			x = \pm \sqrt{\dfrac{\sin (\alpha + \beta)}{\sin (\alpha - \beta)}} \\
			y = \infty \\
			z = \mp \sqrt{\dfrac{\sin (\alpha + \beta)}{\sin (\alpha - \beta)}} \\
			w = \mp \sqrt{\dfrac{\sin (\alpha + \beta) \sin (\alpha - \beta)}{\sin \beta}} 
		\end{dcases} 
		~~ \mathrm{or} ~~ \begin{dcases}
			x = \pm \sqrt{\dfrac{\sin (\beta + \alpha)}{\sin (\beta - \alpha)}} \\
			y = \mp \sqrt{\dfrac{\sin (\beta + \alpha) \sin (\beta - \alpha)}{\sin \alpha}} \\
			z = \mp \sqrt{\dfrac{\sin (\beta + \alpha)}{\sin (\beta - \alpha)}} \\
			w = \infty
		\end{dcases}
	\end{aligned}
\end{equation*}
For the last two group of points, only one group is real solution, depending on the sign of $\sin (\alpha + \beta) \sin (\alpha - \beta)$.

\section{Deltoid II: solution at infinity}

The condition on sector angles is 	
\begin{equation*}
	f_{22} = 0, ~~ f_{20} = 0,~~f_{02} \neq 0, ~~f_{00} \neq 0 ~~ \Leftrightarrow ~~ \alpha=\beta, ~~ \delta=\gamma, ~~\gamma \neq \beta, ~~\beta + \gamma \neq \pi
\end{equation*}
which means:
\begin{equation*}
	\begin{aligned}
		& \begin{cases}
			\sin(\beta-\gamma)x_2^2y_1^2-2 x_1y_1x_2y_2\sin \gamma+\sin(\beta+\gamma)x_2^2y_2^2 =0 \\
			x_2^2 (z_1^2+z_2^2) \sin^2 \beta=  z_2^2(x_1^2+x_2^2) \sin^2 \gamma \\
			\sin (\beta-\gamma)x_2^2w_1^2-2 x_1w_1x_2w_2 \sin \gamma+\sin (\beta+\gamma)x_2^2w_2^2 =0 \\	
			w_1^2y_2^2-w_2^2y_1^2=0 
		\end{cases} \\
		~~ \Rightarrow ~~ & \begin{cases}
			x_1 \neq 0, ~~ x_2 = 0 \\
			y = \pm w \\
			z_1 \neq 0, ~~ z_2 = 0 \\	
		\end{cases} 
		~~ \Rightarrow ~~ \begin{cases}
			x = \infty \\
			y = \pm w  \\
			z = \infty \\	
		\end{cases}
	\end{aligned}
\end{equation*}
From post-examination, the solutions at infinity becomes another branch diffeomorphic to a circle $S^1$:
\begin{equation*}
	\begin{cases}
		x \equiv \infty \\
		y \in \mathbb{R} \cup \{\infty\} \\
		z \equiv \infty \\ 
		w = -y \\
	\end{cases}	 
\end{equation*}

\section{Anti-deltoid II: finite solution}

The condition on sector angles is 	
\begin{equation*}
	f_{22} \neq 0, ~~ f_{20} \neq 0,~~f_{02} = 0, ~~f_{00} = 0 ~~ \Leftrightarrow ~~ \alpha= \pi - \beta, ~~ \delta= \pi - \gamma, ~~\gamma \neq \beta, ~~\beta + \gamma \neq \pi
\end{equation*}
which means:
\begin{equation*}
	\begin{aligned}
		& \begin{cases}
			\sin (\beta-\gamma)x_1^2y_1^2+2 x_1y_1x_2y_2 \sin \gamma + \sin (\beta+\gamma)x_1^2y_2^2 =0 \\
			x_2^2 (z_1^2+z_2^2) \sin^2 \beta=  z_2^2(x_1^2+x_2^2) \sin^2 \gamma \\
			\sin (\beta-\gamma)x_1^2w_1^2-2 x_1w_1x_2w_2 \sin \gamma + \sin (\beta+\gamma)x_1^2w_2^2 =0
		\end{cases} \\
		\Rightarrow ~~  & \begin{dcases}
			x = 0 \\
			y = \infty \\
			z = 0 \\
			w = \infty 
		\end{dcases} ~~ \mathrm{or} ~~ \begin{dcases}
			x = \mp \sqrt{\dfrac{\sin (\beta + \gamma) \sin (\beta - \gamma)}{\sin \gamma}} \\
			y = \pm \sqrt{\dfrac{\sin (\beta + \gamma)}{\sin (\beta - \gamma)}} \\
			z = \infty \\
			w = \mp \sqrt {\dfrac{\sin (\beta + \gamma)}{\sin (\beta - \gamma)}} \\
		\end{dcases} 
		~~ \mathrm{or} ~~ \begin{dcases}
			x = \infty \\
			y = \pm \sqrt{\dfrac{\sin (\gamma + \beta)}{\sin (\gamma - \beta)}} \\
			z = \mp \sqrt{\dfrac{\sin (\gamma + \beta) \sin (\gamma - \beta)}{\sin \beta}} \\
			w = \mp \sqrt{\dfrac{\sin (\gamma + \beta)}{\sin (\gamma - \beta)}} \\
		\end{dcases}
	\end{aligned}
\end{equation*}
For the last two group of points, only one group is real solution, depending on the sign of $\sin (\beta + \gamma) \sin (\beta - \gamma)$.

\section{Conic I: solution at infinity}

The condition on sector angles is:
\begin{equation*} 
	f_{22} = 0, ~~ f_{20} \neq 0,~~f_{02} \neq 0, ~~f_{00} \neq 0  ~~ \Leftrightarrow ~~ 
	\begin{dcases}
		\alpha - \beta + \gamma - \delta = 0 \\
		\alpha - \beta - \gamma + \delta \neq 0 \\
		\alpha + \beta - \gamma - \delta \neq 0 \\
		\alpha + \beta + \gamma + \delta \neq 2\pi
	\end{dcases}
\end{equation*} 
which implies $\sigma=\alpha+\gamma=\beta+\delta$, and:
\begin{equation*} 
	\begin{aligned}
		& \Rightarrow ~~ \begin{cases}
			\sin\gamma\sin(\beta-\alpha)x_1^2y_2^2-2 x_1y_1x_2y_2 \sin\alpha\sin\gamma \\
			+\sin\alpha\sin(\beta-\gamma)x_2^2y_1^2+\sin\beta\sin(\alpha+\gamma)x_2^2y_2^2 =0 \\
			x_2^2 (z_1^2+z_2^2) \sin\alpha \sin\beta =  z_2^2(x_1^2+x_2^2) \sin\gamma \sin\delta \\
			\sin\delta\sin(\alpha-\beta)x_1^2w_2^2-2 x_1w_1x_2w_2 \sin\beta\sin\delta \\ +\sin\beta\sin(\alpha-\delta)x_2^2w_1^2+\sin\alpha\sin(\beta+\delta)x_2^2w_2^2 =0
		\end{cases} \\
		& \Rightarrow ~~ \begin{cases}
			x_1 \neq 0, ~~ x_2=0 \\
			y_1 \neq 0, ~~ y_2 = 0 \\
			z_1 \neq 0, ~~ z_2=0 \\
			w_1 \neq 0, ~~ w_2 = 0 \\
		\end{cases} ~~ \Rightarrow ~~ \begin{cases}
			x = \infty \\
			y = \infty \\
			z = \infty \\
			w = \infty \\
		\end{cases} \\
	\end{aligned}
\end{equation*}
This is exactly the limit point of the two branches in the finite solution.

\section{Conic II: solution at infinity}

The condition on sector angles is 
\begin{equation*} 
	f_{22} \neq 0, ~~ f_{20} = 0,~~f_{02} \neq 0, ~~f_{00} \neq 0  ~~ \Leftrightarrow ~~ 
	\begin{dcases}
		\alpha - \beta + \gamma - \delta \neq 0 \\
		\alpha - \beta - \gamma + \delta = 0 \\
		\alpha + \beta - \gamma - \delta \neq 0 \\
		\alpha + \beta + \gamma + \delta \neq 2\pi
	\end{dcases}
\end{equation*}
\begin{equation*}
	\sigma = \dfrac{\alpha + \beta - \gamma - \delta}{2} + \pi
\end{equation*}
\begin{equation*}
	\begin{cases}
		\sin\gamma\sin(\alpha-\beta)x_1^2y_1^2-2 x_1y_1x_2y_2\sin\alpha\sin\gamma \\ +\sin\beta\sin(\alpha-\gamma)x_2^2y_1^2+\sin\alpha\sin(\beta+\gamma)x_2^2y_2^2 =0 \\
		x_2^2 (z_1^2+z_2^2) \sin\alpha \sin\beta= z_2^2(x_1^2+x_2^2) \sin\gamma \sin\delta \\
		\sin\delta\sin(\beta-\alpha)x_1^2w_1^2-2 x_1w_1x_2w_2 \sin\beta\sin\delta \\ +\sin\alpha\sin(\beta-\delta)x_2^2w_1^2+\sin\beta\sin(\alpha+\delta)x_2^2w_2^2 =0 \\
	\end{cases}
\end{equation*}
\begin{equation*}
	\begin{aligned}
		~~ \Rightarrow ~~ & \begin{cases}
			x_1 \neq 0, ~~ x_2=0 \\
			y_1 = 0, ~~ y_2 \neq 0 \\
			z_1 \neq 0, ~~ z_2=0 \\
			w_1 = 0, ~~ w_2 \neq 0 \\
		\end{cases} \mathrm{or} \quad \begin{cases}
			y_1 \neq 0, ~~ y_2 = 0 \\
			\sin \gamma \sin (\alpha-\beta)x_1^2+\sin \beta \sin (\alpha-\gamma)x_2^2=0 \\
			\sin \alpha \sin \beta (z^2+1)= \sin \gamma \sin \delta (x^2+1) \\
			\sin \beta \sin \gamma (w^{-2}+1) = \sin \delta \sin \alpha  \\
		\end{cases} \\
		& \mathrm{or} \quad \begin{cases}
			w_1 \neq 0, ~~ w_2 = 0 \\
			\sin \delta\sin (\beta-\alpha)x_1^2+\sin \alpha \sin(\beta-\delta)x_2^2=0 \\
			\sin \delta \sin \alpha (y^{-2} +1 ) = \sin \beta \sin \gamma \\
			\sin \alpha \sin \beta (z^2+1)= \sin \gamma \sin \delta (x^2+1) \\
		\end{cases} \\
	\end{aligned}
\end{equation*}
\begin{equation*} 
	\begin{aligned}
		~~ \Rightarrow ~~ & \begin{cases}
			x = \infty \\
			y = 0 \\
			z = \infty \\
			w=0 
		\end{cases} \mathrm{or} \quad \begin{dcases}
			x = \pm \sqrt{ \dfrac{\sin\alpha\sin(\beta-\gamma)}{\sin\gamma\sin(\beta-\alpha)}-1} \\
			y = \infty \\
			z = \pm \sqrt{\dfrac{\sin\delta\sin(\gamma-\beta)}{\sin\beta\sin(\gamma-\delta)}-1} \\
			w^{-1} = \pm \sqrt{\dfrac{\sin\delta \sin\alpha}{\sin\beta \sin\gamma}-1}  \\
		\end{dcases} \quad \mathrm{or} \quad \begin{dcases}
			x = \pm \sqrt{ \dfrac{\sin\beta\sin(\alpha-\delta)}{\sin\delta\sin(\alpha-\beta)}-1} \\
			y^{-1} = \pm \sqrt{\dfrac{\sin\beta \sin\gamma}{ \sin\delta \sin\alpha}-1}  \\
			z = \pm \sqrt{ \dfrac{\sin\gamma\sin(\delta-\alpha)}{\sin\alpha\sin(\delta-\gamma)}-1} \\
			w = \infty \\
		\end{dcases}
	\end{aligned}
\end{equation*}

From post-examination, the isolated solutions at infinity are:
\begin{equation*}
	\begin{cases}
		x = \infty \\
		y = 0 \\
		z = \infty \\
		w=0 
	\end{cases}
\end{equation*}
and when $\sin \delta \sin \alpha > \sin \beta \sin \gamma$, 
\begin{equation*}
	\begin{dcases}
		x = \mathrm{sign} (\sigma - \pi)\sqrt{ \dfrac{\sin \alpha \sin(\beta-\gamma)}{\sin \gamma \sin (\beta-\alpha)}-1} \\
		y = \infty \\
		z = \mathrm{sign} (\pi - \sigma)  \sqrt{\dfrac{\sin \delta \sin(\gamma-\beta)}{\sin \beta \sin(\gamma-\delta)}-1} \\
		w^{-1} = \sqrt{\dfrac{\sin\delta \sin\alpha}{\sin\beta \sin\gamma}-1}  \\ 
	\end{dcases} , ~ \begin{dcases}
		x = \mathrm{sign} (\pi - \sigma) \sqrt{ \dfrac{\sin\alpha\sin(\beta-\gamma)}{\sin\gamma\sin(\beta-\alpha)}-1} \\
		y = \infty \\
		z = \mathrm{sign} (\sigma - \pi) \sqrt{\dfrac{\sin\delta\sin(\gamma-\beta)}{\sin\beta\sin(\gamma-\delta)}-1} \\
		w^{-1} = -\sqrt{\dfrac{\sin\delta \sin\alpha}{\sin\beta \sin\gamma}-1}  \\
	\end{dcases}
\end{equation*}
On the other hand, when $\sin \delta \sin \alpha < \sin \beta \sin \gamma$,
\begin{equation*}
	\begin{dcases}
		x = \mathrm{sign} (\sigma - \pi) \sqrt{ \dfrac{\sin \beta\sin(\alpha-\delta)}{\sin\delta\sin(\alpha-\beta)}-1} \\
		y^{-1} = \sqrt{\dfrac{\sin\beta \sin\gamma}{ \sin\delta \sin\alpha}-1}  \\
		z = \mathrm{sign} (\pi - \sigma) \sqrt{ \dfrac{\sin\gamma\sin(\delta-\alpha)}{\sin\alpha\sin(\delta-\gamma)}-1} \\
		w = \infty 
	\end{dcases} , ~ \begin{dcases}
		x = \mathrm{sign} (\pi - \sigma) \sqrt{ \dfrac{\sin\beta\sin(\alpha-\delta)}{\sin\delta\sin(\alpha-\beta)}-1} \\
		y^{-1} = - \sqrt{\dfrac{\sin\beta \sin\gamma}{ \sin\delta \sin\alpha}-1}  \\
		z = \mathrm{sign} (\sigma - \pi) \sqrt{ \dfrac{\sin\gamma\sin(\delta-\alpha)}{\sin\alpha\sin(\delta-\gamma)}-1} \\
		w = \infty \\ 
	\end{dcases}
\end{equation*}
The two sets of solutions presented above will alternate between real and complex values depending on whether $\sin\delta \sin\alpha$ is greater than $\sin\beta \sin\gamma$, or vice versa. At infinity, only real values are considered valid solutions.

\section{Conic III: solution at infinity}

The condition on sector angles is 
\begin{equation*} 
	f_{22} \neq 0, ~~ f_{20} \neq 0,~~f_{02} = 0, ~~f_{00} \neq 0  ~~ \Leftrightarrow ~~ 
	\begin{dcases}
		\alpha - \beta + \gamma - \delta \neq 0 \\
		\alpha - \beta - \gamma + \delta \neq 0 \\
		\alpha + \beta - \gamma - \delta = 0 \\
		\alpha + \beta + \gamma + \delta \neq 2\pi
	\end{dcases}
\end{equation*}
\begin{equation*}
	\sigma = \dfrac{-\alpha + \beta + \gamma - \delta}{2} + \pi
\end{equation*}
\begin{equation*}
	\begin{cases}
		\sin \alpha\sin(\gamma-\beta)x_1^2y_1^2+\sin\beta\sin(\gamma-\alpha)x_1^2y_2^2 \\
		-2x_1y_1x_2y_2\sin\alpha\sin\gamma+\sin\gamma\sin(\alpha+\beta)x_2^2y_2^2 =0 \\
		x_1^2 (z_1^2+z_2^2) \sin\alpha \sin\beta= z_1^2(x_1^2+x_2^2) \sin\gamma \sin\delta \\
		\sin\beta\sin(\delta-\alpha)x_1^2w_1^2+\sin\alpha\sin(\delta-\beta)x_1^2w_2^2 \\ -2 x_1w_1x_2w_2\sin\beta\sin\delta+\sin\delta\sin(\alpha+\beta)x_2^2w_2^2 =0 \\
	\end{cases}
\end{equation*}
\begin{equation*}
	\begin{aligned}
		~~ \Rightarrow ~~ & \begin{cases}
			x_1 \neq 0, ~~ x_2 = 0 \\
			\sin \alpha\sin(\gamma-\beta)y_1^2+\sin\beta\sin(\gamma-\alpha)y_2^2=0 \\
			\sin\alpha \sin\beta (z^{-2}+1)=\sin\gamma \sin\delta \\
			\sin\beta\sin(\delta-\alpha)w_1^2+\sin\alpha\sin(\delta-\beta)w_2^2=0  \\
		\end{cases} \mathrm{or} \quad  \begin{cases}
			x_1 = 0, ~~ x_2 \neq 0 \\
			y_1 \neq 0, ~~ y_2 = 0 \\
			z_1 = 0, ~~ z_2 \neq 0 \\
			w_1 \neq 0, ~~ w_2 = 0 \\
		\end{cases} \\
		& \mathrm{or} \quad \begin{cases}
			z_1 \neq 0, ~~ z_2 = 0 \\
			\sin\alpha\sin(\gamma-\delta)w_1^2+\sin\delta\sin(\gamma-\alpha)w_2^2=0  \\
			\sin\gamma \sin\delta (x^{-2}+1)=\sin\alpha \sin\beta \\
			\sin\delta \sin(\beta-\gamma)y_1^2+\sin\gamma \sin(\beta-\delta)y_2^2=0 \\
		\end{cases} \\
	\end{aligned}
\end{equation*}
\begin{equation*} 
	\begin{aligned}
		~~ \Rightarrow ~~ & \begin{dcases}
			x = \infty \\
			y = \pm \sqrt{ \dfrac{\sin\gamma\sin(\beta-\alpha)}{\sin\alpha\sin(\beta-\gamma)}-1} \\
			z^{-1} = \pm \sqrt{\dfrac{\sin\gamma \sin\delta}{\sin\alpha \sin\beta}-1} \\
			w = \pm \sqrt{\dfrac{\sin\delta\sin(\alpha-\beta)}{\sin\beta\sin(\alpha-\delta)}-1} 
		\end{dcases} \quad \mathrm{or} \quad  \begin{cases}
			x = 0 \\
			y = \infty \\
			z = 0 \\
			w = \infty \\
		\end{cases} 
		\mathrm{or} \quad  \begin{dcases}
			x^{-1} = \pm \sqrt{\dfrac{\sin\alpha \sin\beta}{\sin\gamma \sin\delta}-1} \\
			y = \pm \sqrt{ \dfrac{\sin\beta\sin(\gamma-\delta)}{\sin\delta\sin(\gamma-\beta)}-1} \\
			z = \infty \\
			w = \pm \sqrt{\dfrac{\sin\alpha\sin(\delta-\gamma)}{\sin\gamma\sin(\delta-\alpha)}-1}  \\
		\end{dcases}
	\end{aligned}
\end{equation*}

From post-examination, the isolated solutions at infinity are:
\begin{equation*}
	\begin{dcases}
		x = 0 \\
		y = \infty \\
		z = 0 \\
		w = \infty \\
	\end{dcases}
\end{equation*}
and, when $\sin\gamma \sin\delta > \sin\alpha \sin\beta$:
\begin{equation*}
	\begin{dcases}
		x = \infty \\
		y = \mathrm{sign} (\sigma - \pi) \sqrt{ \dfrac{\sin \gamma\sin(\beta-\alpha)}{\sin\alpha\sin(\beta-\gamma)}-1} \\
		z^{-1} = \sqrt{\dfrac{\sin\gamma \sin\delta}{\sin\alpha \sin\beta}-1} \\
		w = \mathrm{sign} (\pi - \sigma) \sqrt{\dfrac{\sin\delta\sin(\alpha-\beta)}{\sin\beta\sin(\alpha-\delta)}-1} 
	\end{dcases} , ~ \begin{dcases}
		x = \infty \\
		y = \mathrm{sign} (\pi - \sigma) \sqrt{ 	\dfrac{\sin\gamma\sin(\beta-\alpha)}{\sin\alpha\sin(\beta-\gamma)}-1} \\
		z^{-1} = - \sqrt{\dfrac{\sin\gamma \sin\delta}{\sin\alpha 	\sin\beta}-1} \\
		w = \mathrm{sign} (\sigma - \pi) 	\sqrt{\dfrac{\sin\delta\sin(\alpha-\beta)}{\sin\beta\sin(\alpha-\delta)}-1} 
	\end{dcases} 
\end{equation*}
when $\sin\gamma \sin\delta < \sin\alpha \sin\beta$:
\begin{equation*}
	\begin{dcases}
		x^{-1} = \sqrt{\dfrac{\sin\alpha \sin\beta}{\sin\gamma \sin\delta}-1} \\
		y = \mathrm{sign} (\sigma - \pi) \sqrt{ \dfrac{\sin\beta\sin(\gamma-\delta)}{\sin\delta\sin(\gamma-\beta)}-1} \\
		z = \infty \\
		w = \mathrm{sign} (\pi - \sigma) \sqrt{\dfrac{\sin\alpha\sin(\delta-\gamma)}{\sin\gamma\sin(\delta-\alpha)}-1}  \\
	\end{dcases} , ~ \begin{dcases}
		x^{-1} = - \sqrt{\dfrac{\sin\alpha \sin\beta}{\sin\gamma \sin\delta}-1} \\
		y = \mathrm{sign} (\pi - \sigma) \sqrt{ \dfrac{\sin\beta\sin(\gamma-\delta)}{\sin\delta\sin(\gamma-\beta)}-1} \\
		z = \infty \\
		w = \mathrm{sign} (\sigma - \pi) \sqrt{\dfrac{\sin\alpha\sin(\delta-\gamma)}{\sin\gamma\sin(\delta-\alpha)}-1}  \\
	\end{dcases}
\end{equation*}
The two sets of solutions presented above will alternate between real and complex values depending on whether $\sin\delta \sin\alpha$ is greater than $\sin\beta \sin\gamma$, or vice versa. At infinity, only real values are considered valid solutions. 

\section{Conic IV: solution at infinity}
The condition on sector angles is:
\begin{equation*} 
	f_{22} \neq 0, ~~ f_{20} \neq 0,~~f_{02} \neq 0, ~~f_{00} = 0  ~~ \Leftrightarrow ~~ 
	\begin{dcases}
		\alpha - \beta + \gamma - \delta \neq 0 \\
		\alpha - \beta - \gamma + \delta \neq 0 \\
		\alpha + \beta - \gamma - \delta \neq 0 \\
		\alpha + \beta + \gamma + \delta = 2\pi
	\end{dcases}
\end{equation*}
which leads to
\begin{equation*}
	\begin{aligned}
		& \begin{cases}
			x_1 \neq 0, ~~ x_2 = 0 \\
			\sin\beta\sin(\delta+\beta)y_1^2+\sin \alpha \sin(\delta+\alpha)y_2^2=0 \\
			\sin \alpha\sin(\gamma+\alpha)w_1^2+\sin \beta \sin(\gamma+\beta)w_2^2=0 \\
			\sin(\beta+\delta)\sin(\alpha+\delta)z_1^2+\sin \beta \sin \alpha z_2^2=0
		\end{cases} \\
		\mathrm{or} \quad & \begin{cases}
			y_1 \neq 0, ~~ y_2 = 0 \\
			\sin \gamma\sin(\alpha+\gamma)z_1^2+\sin\beta\sin(\alpha+\beta)z_2^2=0 \\
			\sin \beta\sin(\delta+\beta)x_1^2+\sin\gamma\sin(\delta+\gamma)x_2^2=0 \\
			\sin(\gamma+\alpha)\sin(\beta+\alpha)w_1^2+\sin \gamma \sin \beta w_2^2=0
		\end{cases} \\
		\mathrm{or} \quad & \begin{cases}
			z_1 \neq 0, ~~ z_2 = 0 \\
			\sin \delta \sin(\beta+\delta)w_1^2+\sin \gamma \sin(\beta+\gamma)w_2^2=0 \\
			\sin \gamma \sin(\alpha+\gamma)y_1^2+\sin \delta \sin(\alpha+\delta)y_2^2=0 \\
			\sin(\delta+\beta)\sin(\gamma+\beta)x_1^2+\sin \delta \sin \gamma x_2^2=0
		\end{cases} \\
		\mathrm{or} \quad & \begin{cases}
			w_1 \neq 0, ~~ w_2 = 0 \\
			\sin \alpha \sin(\gamma+\alpha)x_1^2+\sin \delta \sin(\gamma+\delta)x_2^2=0 \\
			\sin \delta \sin(\beta+\delta)z_1^2+\sin \alpha \sin(\beta+\alpha)z_2^2=0 \\
			\sin(\alpha+\gamma)\sin(\delta+\gamma)y_1^2+\sin \alpha \sin \delta y_2^2=0
		\end{cases} \\
	\end{aligned}
\end{equation*}

\begin{equation*}
	\begin{aligned}
		\Rightarrow \quad & \begin{dcases}
			x =  \infty \\
			y =  \pm \sqrt{\dfrac{\sin \gamma\sin (\beta-\alpha)}{\sin \beta\sin(\alpha+\gamma)}-1} \\
			z^{-1} = \pm \sqrt{\dfrac{\sin\gamma \sin\delta}{\sin \alpha \sin \beta}-1} \\
			w = \pm \sqrt{\dfrac{\sin\delta\sin(\alpha-\beta)}{\sin \alpha \sin(\beta+\delta)}-1} 
		\end{dcases}  ~\mathrm{or}~ \begin{dcases}
			y =  \infty \\
			z =  \pm \sqrt{\dfrac{\sin\delta\sin(\gamma-\beta)}{\sin \gamma \sin(\beta+\delta)}-1} \\
			w^{-1} = \pm \sqrt{\dfrac{\sin\delta \sin\alpha}{\sin \beta \sin \gamma}-1} \\
			x = \pm \sqrt{\dfrac{\sin\alpha\sin(\beta-\gamma)}{\sin \beta \sin(\alpha+\gamma)}-1}
		\end{dcases} \\
		\mathrm{or} ~ & \begin{dcases}
			z = \infty \\
			w = \pm \sqrt{\dfrac{\sin\alpha\sin(\delta-\gamma)}{\sin \delta \sin(\alpha+\gamma)}-1} \\
			x^{-1} = \pm \sqrt{\dfrac{\sin\alpha \sin\beta}{\sin \gamma \sin \delta}-1} \\
			y = \pm \sqrt{\dfrac{\sin\beta\sin(\gamma-\delta)}{\sin \gamma \sin(\beta+\delta)}-1} \\
		\end{dcases}
		~\mathrm{or}~ \begin{dcases}
			w =  \infty \\
			x =  \pm \sqrt{\dfrac{\sin\beta\sin(\alpha-\delta)}{\sin \alpha \sin(\beta+\delta)}-1} \\
			y^{-1} = \pm \sqrt{\dfrac{\sin\beta \sin\gamma}{\sin \alpha \sin \delta}-1} \\
			z = \pm \sqrt{\dfrac{\sin\gamma\sin(\delta-\alpha)}{\sin \delta \sin(\alpha+\gamma)}-1} \\
		\end{dcases} \\
	\end{aligned} 
\end{equation*} 

From post examination we could obtain the solution similar to Conic II and III.
\begin{equation*}
	\sigma = \dfrac{-\alpha + \beta - \gamma + \delta}{2} + \pi
\end{equation*}
When $\sin\gamma\sin\delta>\sin\alpha\sin\beta$,
\begin{equation*}
	\begin{aligned}
		\begin{dcases}
			x =  \infty \\
			y =  \mathrm{sign}(\sigma - \pi) \sqrt{\dfrac{\sin \gamma\sin (\beta-\alpha)}{\sin \beta\sin(\alpha+\gamma)}-1} \\
			z^{-1} = \sqrt{\dfrac{\sin\gamma \sin\delta}{\sin \alpha \sin \beta}-1} \\
			w = \mathrm{sign}(\pi - \sigma) \sqrt{\dfrac{\sin\delta\sin(\alpha-\beta)}{\sin \alpha \sin(\beta+\delta)}-1}
		\end{dcases}, ~~
		\begin{dcases}
			x =  \infty \\
			y =  \mathrm{sign}(\pi-\sigma) \sqrt{\dfrac{\sin \gamma\sin (\beta-\alpha)}{\sin \beta\sin(\alpha+\gamma)}-1} \\
			z^{-1} = -\sqrt{\dfrac{\sin\gamma \sin\delta}{\sin \alpha \sin \beta}-1} \\
			w = \mathrm{sign}(\sigma-\pi) \sqrt{\dfrac{\sin\delta\sin(\alpha-\beta)}{\sin \alpha \sin(\beta+\delta)}-1}
		\end{dcases}
	\end{aligned}
\end{equation*} 
When $\sin\delta\sin\alpha>\sin\beta\sin\gamma$,
\begin{equation*}
	\begin{aligned}
		\begin{dcases}
			y =  \infty \\
			z =  \mathrm{sign}(\sigma - \pi) \sqrt{\dfrac{\sin\delta\sin(\gamma-\beta)}{\sin \gamma \sin(\beta+\delta)}-1} \\
			w^{-1} = \sqrt{\dfrac{\sin\delta \sin\alpha}{\sin \beta \sin \gamma}-1} \\
			x = \mathrm{sign}(\pi-\sigma) \sqrt{\dfrac{\sin\alpha\sin(\beta-\gamma)}{\sin \beta \sin(\alpha+\gamma)}-1}
		\end{dcases}, ~~
		\begin{dcases}
			y =  \infty \\
			z =  \mathrm{sign}(\pi - \sigma) \sqrt{\dfrac{\sin\delta\sin(\gamma-\beta)}{\sin \gamma \sin(\beta+\delta)}-1} \\
			w^{-1} = - \sqrt{\dfrac{\sin\delta \sin\alpha}{\sin \beta \sin \gamma}-1} \\
			x = \mathrm{sign}(\sigma - \pi) \sqrt{\dfrac{\sin\alpha\sin(\beta-\gamma)}{\sin \beta \sin(\alpha+\gamma)}-1}
		\end{dcases}
	\end{aligned}
\end{equation*}
When $\sin\alpha\sin\beta>\sin\gamma\sin\delta$,
\begin{equation*}
	\begin{aligned}
		\begin{dcases}
			z = \infty \\
			w = \mathrm{sign}(\sigma - \pi) \sqrt{\dfrac{\sin\alpha\sin(\delta-\gamma)}{\sin \delta \sin(\alpha+\gamma)}-1} \\
			x^{-1} =  \sqrt{\dfrac{\sin\alpha \sin\beta}{\sin \gamma \sin \delta}-1} \\
			y = \mathrm{sign}(\pi - \sigma) \sqrt{\dfrac{\sin\beta\sin(\gamma-\delta)}{\sin \gamma \sin(\beta+\delta)}-1} \\
		\end{dcases}, ~~
		\begin{dcases}
			z = \infty \\
			w = \mathrm{sign}(\pi - \sigma) \sqrt{\dfrac{\sin\alpha\sin(\delta-\gamma)}{\sin \delta \sin(\alpha+\gamma)}-1} \\
			x^{-1} = - \sqrt{\dfrac{\sin\alpha \sin\beta}{\sin \gamma \sin \delta}-1} \\
			y = \mathrm{sign}(\sigma - \pi) \sqrt{\dfrac{\sin\beta\sin(\gamma-\delta)}{\sin \gamma \sin(\beta+\delta)}-1} \\
		\end{dcases}
	\end{aligned}
\end{equation*}
When $\sin\beta\sin\gamma>\sin\delta\sin\alpha$,
\begin{equation*}
	\begin{aligned}
		\begin{dcases}
			w =  \infty \\
			x =  \mathrm{sign}(\sigma - \pi) \sqrt{\dfrac{\sin\beta\sin(\alpha-\delta)}{\sin \alpha \sin(\beta+\delta)}-1} \\
			y^{-1} = \sqrt{\dfrac{\sin\beta \sin\gamma}{\sin \alpha \sin \delta}-1} \\
			z = \mathrm{sign}(\pi - \sigma) \sqrt{\dfrac{\sin\gamma\sin(\delta-\alpha)}{\sin \delta \sin(\alpha+\gamma)}-1} \\
		\end{dcases}, ~~
		\begin{dcases}
			w =  \infty \\
			x =  \mathrm{sign}(\pi - \sigma) \sqrt{\dfrac{\sin\beta\sin(\alpha-\delta)}{\sin \alpha \sin(\beta+\delta)}-1} \\
			y^{-1} = -\sqrt{\dfrac{\sin\beta \sin\gamma}{\sin \alpha \sin \delta}-1} \\
			z = \mathrm{sign}(\sigma - \pi) \sqrt{\dfrac{\sin\gamma\sin(\delta-\alpha)}{\sin \delta \sin(\alpha+\gamma)}-1} \\
		\end{dcases}
	\end{aligned}
\end{equation*}
Here only two of the four groups are real, while the rest are complex.

\section{Elliptic: solution at infinity}

The condition on sector angles is:
\begin{equation*} 
	f_{22} \neq 0, ~~ f_{20} \neq 0 , ~~ f_{02} \neq 0 ~~ \Leftrightarrow ~~ \begin{dcases}
		\alpha - \beta + \gamma - \delta \neq 0 \\
		\alpha - \beta - \gamma + \delta \neq 0 \\
		\alpha + \beta - \gamma - \delta \neq 0
	\end{dcases}
\end{equation*}
which leads to
\begin{equation*}
	\begin{aligned}
		& \begin{cases}
			x_1 \neq 0, ~~ x_2 = 0 \\
			\sin (\sigma-\beta)\sin(\sigma-\delta-\beta)y_1^2+\sin(\sigma-\alpha)\sin(\sigma-\delta-\alpha)y_2^2=0 \\
			\sin(\sigma-\alpha)\sin(\sigma-\gamma-\alpha)w_1^2+\sin(\sigma-\beta)\sin(\sigma-\gamma-\beta)w_2^2=0 \\
			\sin(\sigma-\beta-\delta)\sin(\sigma-\alpha-\delta)z_1^2+\sin(\sigma-\beta)\sin(\sigma-\alpha)z_2^2=0
		\end{cases} \\
		\mathrm{or} \quad & \begin{cases}
			y_1 \neq 0, ~~ y_2 = 0 \\
			\sin(\sigma-\gamma)\sin(\sigma-\alpha-\gamma)z_1^2+\sin(\sigma-\beta)\sin(\sigma-\alpha-\beta)z_2^2=0 \\
			\sin(\sigma-\beta)\sin(\sigma-\delta-\beta)x_1^2+\sin(\sigma-\gamma)\sin(\sigma-\delta-\gamma)x_2^2=0 \\
			\sin(\sigma-\gamma-\alpha)\sin(\sigma-\beta-\alpha)w_1^2+\sin(\sigma-\gamma)\sin(\sigma-\beta)w_2^2=0
		\end{cases} \\
		\mathrm{or} \quad & \begin{cases}
			z_1 \neq 0, ~~ z_2 = 0 \\
			\sin(\sigma-\delta)\sin(\sigma-\beta-\delta)w_1^2+\sin(\sigma-\gamma)\sin(\sigma-\beta-\gamma)w_2^2=0 \\
			\sin(\sigma-\gamma)\sin(\sigma-\alpha-\gamma)y_1^2+\sin(\sigma-\delta)\sin(\sigma-\alpha-\delta)y_2^2=0 \\
			\sin(\sigma-\delta-\beta)\sin(\sigma-\gamma-\beta)x_1^2+\sin(\sigma-\delta)\sin(\sigma-\gamma)x_2^2=0
		\end{cases} \\
		\mathrm{or} \quad & \begin{cases}
			w_1 \neq 0, ~~ w_2 = 0 \\
			\sin(\sigma-\alpha)\sin(\sigma-\gamma-\alpha)x_1^2+\sin(\sigma-\delta)\sin(\sigma-\gamma-\delta)x_2^2=0 \\
			\sin(\sigma-\delta)\sin(\sigma-\beta-\delta)z_1^2+\sin(\sigma-\alpha)\sin(\sigma-\beta-\alpha)z_2^2=0 \\
			\sin(\sigma-\alpha-\gamma)\sin(\sigma-\delta-\gamma)y_1^2+\sin(\sigma-\alpha)\sin(\sigma-\delta)y_2^2=0
		\end{cases} \\
	\end{aligned}
\end{equation*}

\begin{equation*}
	\begin{aligned}
		\Rightarrow \quad & \begin{dcases}
			x =  \infty \\
			y =  \pm \sqrt{\dfrac{\sin \gamma\sin (\beta-\alpha)}{\sin (\sigma-\beta)\sin(\sigma-\alpha-\gamma)}-1} \\
			z^{-1} = \pm \sqrt{\dfrac{\sin\gamma \sin\delta}{\sin(\sigma-\alpha)\sin(\sigma-\beta)}-1} \\
			w = \pm \sqrt{\dfrac{\sin\delta\sin(\alpha-\beta)}{\sin(\sigma-\alpha)\sin(\sigma-\beta-\delta)}-1} 
		\end{dcases}  ~\mathrm{or}~ \begin{dcases}
			y =  \infty \\
			z =  \pm \sqrt{\dfrac{\sin\delta\sin(\gamma-\beta)}{\sin(\sigma-\gamma)\sin(\sigma-\beta-\delta)}-1} \\
			w^{-1} = \pm \sqrt{\dfrac{\sin\delta \sin\alpha}{\sin(\sigma-\beta)\sin(\sigma-\gamma)}-1} \\
			x = \pm \sqrt{\dfrac{\sin\alpha\sin(\beta-\gamma)}{\sin(\sigma-\beta)\sin(\sigma-\alpha-\gamma)}-1}
		\end{dcases} \\
		\mathrm{or} ~ & \begin{dcases}
			z = \infty \\
			w = \pm \sqrt{\dfrac{\sin\alpha\sin(\delta-\gamma)}{\sin(\sigma-\delta)\sin(\sigma-\alpha-\gamma)}-1} \\
			x^{-1} = \pm \sqrt{\dfrac{\sin\alpha \sin\beta}{\sin(\sigma-\gamma)\sin(\sigma-\delta)}-1} \\
			y = \pm \sqrt{\dfrac{\sin\beta\sin(\gamma-\delta)}{\sin(\sigma-\gamma)\sin(\sigma-\beta-\delta)}-1} \\
		\end{dcases}
		~\mathrm{or}~ \begin{dcases}
			w =  \infty \\
			x =  \pm \sqrt{\dfrac{\sin\beta\sin(\alpha-\delta)}{\sin(\sigma-\alpha)\sin(\sigma-\beta-\delta)}-1} \\
			y^{-1} = \pm \sqrt{\dfrac{\sin\beta \sin\gamma}{\sin(\sigma-\alpha)\sin(\sigma-\delta)}-1} \\
			z = \pm \sqrt{\dfrac{\sin\gamma\sin(\delta-\alpha)}{\sin(\sigma-\delta)\sin(\sigma-\alpha-\gamma)}-1} \\
		\end{dcases} \\
	\end{aligned} 
\end{equation*} 

Similarly, our next step is to determine how many of the above solutions are real. Take $ x = \infty$ as an example, the condition for $y, ~z, ~w$ to be real numbers is (referring to the helpful identities in Section \ref{section: sign convention folding angle}): 
\begin{equation*}
	\begin{aligned}
		& \quad \sin(\sigma-\alpha)\sin(\sigma-\beta) - \sin\gamma \sin \delta = \sin(\sigma-\alpha-\gamma)\sin(\sigma-\beta-\gamma) \\
		& = \dfrac{\sin \alpha \sin\beta \sin\gamma \sin\delta - \sin(\sigma-\alpha)\sin(\sigma-\beta)\sin(\sigma-\gamma)\sin(\sigma-\delta)}{\sin\sigma\sin(\sigma-\alpha-\beta)} < 0
	\end{aligned}
\end{equation*} 
Following the notation introduced for the Elliptic type, 
\begin{equation}
	M = \dfrac{\alpha \beta \gamma \delta}{(\sigma-\alpha)(\sigma-\beta)(\sigma-\gamma)(\sigma-\delta)} \in (0, 1) \cup (1, +\infty)
\end{equation}
the condition for $x=\infty$ to be a real configuration is
\begin{equation*}
	(M-1)\sin\sigma\sin(\sigma-\alpha-\beta) < 0
\end{equation*} 
Let us write all the cases together:
\begin{equation*}
	\begin{dcases}
		(M-1)\mathrm{sign}(\pi - \sigma)\mathrm{sign}(\sigma-\alpha-\beta)<0 ~~ \Leftrightarrow ~~ x=\infty \mathrm{~can~be~reached} \\
		(M-1)\mathrm{sign}(\pi - \sigma)\mathrm{sign}(\sigma-\beta-\gamma)<0 ~~ \Leftrightarrow ~~ y=\infty \mathrm{~can~be~reached} \\
		(M-1)\mathrm{sign}(\pi - \sigma)\mathrm{sign}(\sigma-\gamma-\delta)<0 ~~ \Leftrightarrow ~~ z=\infty \mathrm{~can~be~reached} \\
		(M-1)\mathrm{sign}(\pi - \sigma)\mathrm{sign}(\sigma-\delta-\alpha)<0 ~~ \Leftrightarrow ~~ w=\infty \mathrm{~can~be~reached} \\
	\end{dcases}
\end{equation*}

From post-examination, when $(M-1)\mathrm{sign}(\pi - \sigma)\mathrm{sign}(\sigma-\alpha-\beta)<0$, $x$ reaches the infinity:
\begin{equation*}
	\begin{dcases}
		x =  \infty \\
		y =  \mathrm{sign} (\sigma - \pi) \sqrt{\dfrac{\sin \gamma \sin(\beta-\alpha)}{\sin(\sigma-\beta)\sin(\sigma-\alpha-\gamma)}-1} \\
		z^{-1} = \sqrt{\dfrac{\sin\gamma \sin\delta}{\sin(\sigma-\alpha)\sin(\sigma - \beta)}-1} \\
		w = \mathrm{sign} (\pi - \sigma) \sqrt{\dfrac{\sin\delta\sin(\alpha-\beta)}{\sin(\sigma-\alpha)\sin(\sigma-\beta-\delta)}-1} 
	\end{dcases} \\
\end{equation*}
\begin{equation*}
	\begin{dcases}
		x =  \infty \\
		y =  \mathrm{sign} (\pi - \sigma) \sqrt{\dfrac{\sin\gamma\sin(\beta-\alpha)}{\sin(\sigma-\beta)\sin(\sigma-\alpha-\gamma)}-1} \\
		z^{-1} = - \sqrt{\dfrac{\sin\gamma \sin\delta}{\sin(\sigma-\alpha)\sin(\sigma-\beta)}-1} \\
		w = \mathrm{sign} (\sigma - \pi) \sqrt{\dfrac{\sin\delta\sin(\alpha-\beta)}{\sin(\sigma-\alpha)\sin(\sigma-\beta-\delta)}-1} 
	\end{dcases}
\end{equation*}
When $(M-1)\mathrm{sign}(\pi - \sigma)\mathrm{sign}(\sigma-\beta-\gamma)<0$, $y$ reaches the infinity:
\begin{equation*}
	\begin{dcases}
		y =  \infty \\
		z =  \mathrm{sign} (\sigma - \pi) \sqrt{\dfrac{\sin\delta\sin(\gamma-\beta)}{\sin(\sigma-\gamma)\sin(\sigma-\beta-\delta)}-1} \\
		w^{-1} = \sqrt{\dfrac{\sin\delta \sin\alpha}{\sin(\sigma-\beta)\sin(\sigma-\gamma)}-1} \\
		x = \mathrm{sign} (\pi - \sigma) \sqrt{\dfrac{\sin\alpha\sin(\beta-\gamma)}{\sin(\sigma-\beta)\sin(\sigma-\alpha-\gamma)}-1}
	\end{dcases} 
\end{equation*}
\begin{equation*}
	\begin{dcases}
		y =  \infty \\
		z =  \mathrm{sign} (\pi - \sigma) \sqrt{\dfrac{\sin\delta\sin(\gamma-\beta)}{\sin(\sigma-\gamma)\sin(\sigma-\beta-\delta)}-1} \\
		w^{-1} = - \sqrt{\dfrac{\sin\delta \sin\alpha}{\sin(\sigma-\beta)\sin(\sigma-\gamma)}-1} \\
		x = \mathrm{sign} (\sigma - \pi) \sqrt{\dfrac{\sin\alpha\sin(\beta-\gamma)}{\sin(\sigma-\beta)\sin(\sigma-\alpha-\gamma)}-1}
	\end{dcases}
\end{equation*}
When $(M-1)\mathrm{sign}(\pi - \sigma)\mathrm{sign}(\sigma-\gamma-\delta)<0$, $z$ reaches the infinity:
\begin{equation*}
	\begin{dcases}
		z = \infty \\
		w = \mathrm{sign} (\sigma - \pi) \sqrt{\dfrac{\sin\alpha\sin(\delta-\gamma)}{\sin(\sigma-\delta)\sin(\sigma-\alpha-\gamma)}-1} \\
		x^{-1} = \sqrt{\dfrac{\sin\alpha \sin\beta}{\sin(\sigma-\gamma)\sin(\sigma-\delta)}-1} \\
		y = \mathrm{sign} (\pi - \sigma) \sqrt{\dfrac{\sin\beta\sin(\gamma-\delta)}{\sin(\sigma-\gamma)\sin(\sigma-\beta-\delta)}-1} \\
	\end{dcases}
\end{equation*}
\begin{equation*} 
	\begin{dcases}
		z = \infty \\
		w = \mathrm{sign} (\pi - \sigma) \sqrt{\dfrac{\sin\alpha\sin(\delta-\gamma)}{\sin(\sigma-\delta)\sin(\sigma-\alpha-\gamma)}-1} \\
		x^{-1} = -\sqrt{\dfrac{\sin\alpha \sin\beta}{\sin(\sigma-\gamma)\sin(\sigma-\delta)}-1} \\
		y = \mathrm{sign} (\sigma - \pi) \sqrt{\dfrac{\sin\beta\sin(\gamma-\delta)}{\sin(\sigma-\gamma)\sin(\sigma-\beta-\delta)}-1} \\
	\end{dcases}
\end{equation*}
When $(M-1)\mathrm{sign}(\pi - \sigma)\mathrm{sign}(\sigma-\delta-\alpha)<0$, $w$ reaches the infinity:
\begin{equation*}
	\begin{dcases}
		w =  \infty \\
		x =  \mathrm{sign} (\sigma - \pi) \sqrt{\dfrac{\sin\beta\sin(\alpha-\delta)}{\sin(\sigma-\alpha)\sin(\sigma-\beta-\delta)}-1} \\
		y^{-1} = \sqrt{\dfrac{\sin\beta \sin\gamma}{\sin(\sigma-\alpha)\sin(\sigma-\delta)}-1} \\
		z = \mathrm{sign} (\pi - \sigma) \sqrt{\dfrac{\sin\gamma\sin(\delta-\alpha)}{\sin(\sigma-\delta)\sin(\sigma-\alpha-\gamma)}-1} \\
	\end{dcases}
\end{equation*}
\begin{equation*} 
	\begin{dcases}
		w =  \infty \\
		x =  \mathrm{sign} (\pi - \sigma) \sqrt{\dfrac{\sin\beta\sin(\alpha-\delta)}{\sin(\sigma-\alpha)\sin(\sigma-\beta-\delta)}-1} \\
		y^{-1} = -\sqrt{\dfrac{\sin\beta \sin\gamma}{\sin(\sigma-\alpha)\sin(\sigma-\delta)}-1} \\
		z = \mathrm{sign} (\sigma - \pi) \sqrt{\dfrac{\sin\gamma\sin(\delta-\alpha)}{\sin(\sigma-\delta)\sin(\sigma-\alpha-\gamma)}-1} \\
	\end{dcases}
\end{equation*}

\section{The Grashof condition and self-intersection}

\begin{prop}
	(The spherical version of the Grashof condition) If the sum of the shortest and longest linkage of a spherical 4-bar linkage is less than or equal to the sum of the remaining two linkages, then the shortest linkage can rotate fully with respect to a neighbouring linkage.
\end{prop}

We can easily analyse this condition based on our prior discussions regarding solutions at infinity. Only the elliptic type requires further justification. Without loss of generality, suppose $\beta$ is the shortest linkage, then the condition for $x= \infty$ or $y= \infty$ is
\begin{equation*}
	\begin{gathered}
		(M-1)\mathrm{sign}(\pi - \sigma)\mathrm{sign}(\sigma-\alpha-\beta)<0 \\
		(M-1)\mathrm{sign}(\pi - \sigma)\mathrm{sign}(\sigma-\beta-\gamma)<0
	\end{gathered} 
\end{equation*}
Since $\beta$ is the shortest,
\begin{equation*}
	\sigma-\alpha-\beta>0 ~~\mathrm{or}~~ \sigma-\beta-\gamma>0
\end{equation*}
hence in order for $x= \infty$ or $y= \infty$ can be reached, we need to require:
\begin{equation*}
	(M-1)\mathrm{sign}(\pi - \sigma) < 0 \Leftrightarrow \max + \min < \sigma
\end{equation*}
The right hand side is exactly the Grashof condition.

\begin{prop}
	When no folding angle reaches zero or the infinity, a spherical 4-bar linkage is self-intersected if and only if
	\begin{equation*}
		(\mathrm{sign}(x), \mathrm{sign}(y), \mathrm{sign}(z), \mathrm{sign}(w)) = 
		\begin{dcases}
			(1, ~1, ~-1, ~-1) \\
			(-1, ~1, ~1, ~-1) \\
			(-1, ~-1, ~1, ~1) \\
			(1, ~-1, ~-1, ~1)
		\end{dcases}
	\end{equation*}  
\end{prop}

\begin{proof}
	The investigation of self-intersections in a spherical quadrilateral can be simplified to examining the intersection of two arcs, say $AB$ and $CD$. These arcs intersect if and only if the points $A$ and $B$ are positioned on opposite sides of the line $CD$, and the points $C$ and $D$ are on opposite sides of the line $AB$. The points $A$ and $B$ are on different sides of $CD$ if and only if the oriented areas of triangles $ACD$ and $BCD$ have opposite signs. Therefore, for a spherical 4-bar linkage $ABCD$, it is sufficient to calculate the oriented areas of all four triangles and analyze the sign pattern. Additionally, the signs of these oriented areas correspond to the signs of the rotational angles $x, ~y, ~z, ~w$.
\end{proof}

We have introduced this self-intersection check in the accompanying MATLAB app \citep{he_elliptic-fun-based_2023-1}.

\section*{Acknowledgement}

This work is supported by the Japan Science and Technology Agency – Core Research for Evolutional Science and Technology Grant No. JPMJCR1911.

\bibliographystyle{plainnat}

\end{document}